\documentclass[11pt, letterpaper]{article}
\pdfoutput=1
\usepackage{fullpage} 
\usepackage[T1]{fontenc}
\usepackage{xspace,xcolor,graphicx,xltabular}
\usepackage{mathtools,amsthm,amssymb} 
\usepackage{thmtools}
\usepackage{thm-restate}
\usepackage{enumitem}
\setenumerate[0]{label={\normalfont (\roman*)}}
\usepackage{booktabs}
\usepackage{bm,bbm}
\usepackage{soul}

\usepackage{tikz}
\usetikzlibrary{quantikz2}
\usepackage{caption}
\usepackage{subcaption}

\usepackage[UKenglish]{babel}
\usepackage[scaled=0.86]{helvet}
\usepackage{mathptmx}
\DeclareMathAlphabet{\mathcal}{OMS}{cmsy}{m}{n}
\usepackage{dsfont} 
\usepackage{stmaryrd}
\makeatletter
\renewcommand{\paragraph}{%
  \@startsection{paragraph}{4}%
  {\z@}{2.25ex \@plus 1ex \@minus .2ex}{-1em}%
  {\normalfont\normalsize\bfseries}%
}
\makeatother
\usepackage{authblk}

\definecolor{linkblue}{HTML}{001487}
\usepackage[colorlinks=true,allcolors=linkblue]{hyperref}
\usepackage{url}
\usepackage[nameinlink,capitalize]{cleveref}

\newtheorem{theorem}{Theorem}[section]
\newtheorem*{theorem*}{Theorem}
\newtheorem{proposition}[theorem]{Proposition}

\newtheorem{lemma}[theorem]{Lemma}

\newtheorem{fact}[theorem]{Fact}
\newtheorem{corollary}[theorem]{Corollary}
\theoremstyle{definition}
\newtheorem{remark}[theorem]{Remark}
\newtheorem{definition}[theorem]{Definition}

\numberwithin{equation}{section}
\newcommand\numberthis{\addtocounter{equation}{1}\tag{\theequation}}

\newcommand{\eps}{\epsilon}

\newcommand{\F}{\ensuremath{\mathds{F}}}
\newcommand{\N}{\ensuremath{\mathds{N}}}

\newcommand{\bits}{\ensuremath{\{0, 1\}}}

\newcommand{\ZZ}{{\mathbb{Z}}}
\newcommand{\CC}{{\mathbb{C}}}
\newcommand{\FF}{{\mathbb{F}}}
\newcommand{\RR}{{\mathbb{R}}}

\newcommand{\baa}{{\mathbf{a}}}
\newcommand{\bbb}{{\mathbf{b}}}
\newcommand{\bxx}{{\mathbf{x}}}
\newcommand{\byy}{{\mathbf{y}}}
\newcommand{\bzz}{{\mathbf{z}}}
\newcommand{\bww}{{\mathbf{w}}}

\newcommand{\bKK}{{\mathbf{K}}}

\newcommand{\CPFPC}{\mathrm{CPFPC}}
\newcommand{\FPC}{\mathrm{FPC}}
\newcommand{\PFC}{\mathrm{PFC}}
\newcommand{\SWAP}{\mathrm{SWAP}}

\newcommand{\ot}{\ensuremath{\otimes}}
\newcommand{\deq}{\coloneqq}

\DeclareMathOperator{\poly}{poly}

\DeclareMathOperator{\E}{\mathds{E}}
\DeclareMathOperator*{\avg}{\mathds{E}}

\DeclareMathOperator{\rank}{rank}

\DeclareMathOperator*{\Tr}{Tr}

\DeclarePairedDelimiter{\abs}{\lvert}{\rvert} 
\DeclarePairedDelimiter{\norm}{\lVert}{\rVert}

\DeclarePairedDelimiter{\parens}{\lparen}{\rparen}

\newcommand{\tand}{\;\textnormal{~and~}\;}

\usepackage{braket}

\newcommand{\cA}{\ensuremath{\mathcal{A}}}

\newcommand{\cD}{\ensuremath{\mathcal{D}}}

\newcommand{\cH}{\ensuremath{\mathcal{H}}}
\newcommand{\cI}{\ensuremath{\mathcal{I}}}

\newcommand{\cK}{\ensuremath{\mathcal{K}}}

\newcommand{\cO}{\ensuremath{\mathcal{O}}}
\newcommand{\cP}{\ensuremath{\mathcal{P}}}

\newcommand{\cR}{\ensuremath{\mathcal{R}}}

\newcommand{\cX}{\ensuremath{\mathcal{X}}}

\DeclareSymbolFont{greekletters}{OML}{ntxmi}{m}{it}
\DeclareMathSymbol{\alpha}{\mathord}{greekletters}{"0B}
\DeclareMathSymbol{\beta}{\mathord}{greekletters}{"0C}
\DeclareMathSymbol{\gamma}{\mathord}{greekletters}{"0D}
\DeclareMathSymbol{\delta}{\mathord}{greekletters}{"0E}
\DeclareMathSymbol{\epsilon}{\mathord}{greekletters}{"0F}
\DeclareMathSymbol{\zeta}{\mathord}{greekletters}{"10}
\DeclareMathSymbol{\eta}{\mathord}{greekletters}{"11}
\DeclareMathSymbol{\theta}{\mathord}{greekletters}{"12}
\DeclareMathSymbol{\iota}{\mathord}{greekletters}{"13}
\DeclareMathSymbol{\kappa}{\mathord}{greekletters}{"14}
\DeclareMathSymbol{\lambda}{\mathord}{greekletters}{"15}
\DeclareMathSymbol{\mu}{\mathord}{greekletters}{"16}
\DeclareMathSymbol{\nu}{\mathord}{greekletters}{"17}
\DeclareMathSymbol{\xi}{\mathord}{greekletters}{"18}
\DeclareMathSymbol{\pi}{\mathord}{greekletters}{"19}
\DeclareMathSymbol{\rho}{\mathord}{greekletters}{"1A}
\DeclareMathSymbol{\sigma}{\mathord}{greekletters}{"1B}
\DeclareMathSymbol{\tau}{\mathord}{greekletters}{"1C}
\DeclareMathSymbol{\upsilon}{\mathord}{greekletters}{"1D}
\DeclareMathSymbol{\phi}{\mathord}{greekletters}{"1E}
\DeclareMathSymbol{\chi}{\mathord}{greekletters}{"1F}
\DeclareMathSymbol{\psi}{\mathord}{greekletters}{"20}
\DeclareMathSymbol{\omega}{\mathord}{greekletters}{"21}
\DeclareMathSymbol{\varepsilon}{\mathord}{greekletters}{"22}
\DeclareMathSymbol{\vartheta}{\mathord}{greekletters}{"23}
\DeclareMathSymbol{\varpi}{\mathord}{greekletters}{"24}
\DeclareMathSymbol{\varrho}{\mathord}{greekletters}{"25}
\DeclareMathSymbol{\varsigma}{\mathord}{greekletters}{"26}
\DeclareMathSymbol{\varphi}{\mathord}{greekletters}{"27}

\newcommand{\SL}{{\mathrm{SL}}}
\newcommand{\EL}{{\mathrm{EL}}}
\newcommand{\GL}{{\mathrm{GL}}}
\newcommand{\Alt}{{\mathrm{Alt}}}
\newcommand{\Sym}{{\mathrm{Sym}}}

\newcommand{\Mat}{{\mathrm{Mat}}}
\newcommand{\Cay}{\mathrm{Cay}}
\newcommand{\one}{{\mathds{1}}}
\newcommand{\Cl}{{\mathrm{Cl}}}
\newcommand{\SU}{{\mathrm{SU}}}
\newcommand{\U}{{\mathrm{U}}}
\newcommand{\Asterisk}{{\scalebox{1.7}{\raisebox{-0.2ex}{$\ast$}}}}

\newcommand{\nurevalltoall}{{\nu_{\mathrm{rev},\mathrm{All}\to\mathrm{All}}}}
\newcommand{\nutwodesign}{{\nu_{2\text{-design}}}}

\renewcommand{\eps}{\varepsilon}

\renewcommand{\vec}{\bm}

\newcommand{\CO}{\mathcal{O}}

\newcommand{\vertiiismall}[1]{{\vert\kern-0.25ex\vert\kern-0.25ex\vert #1 
    \vert\kern-0.25ex\vert\kern-0.25ex\vert}}

\newcommand{\vertiii}[1]{{\left\vert\kern-0.25ex\left\vert\kern-0.25ex\left\vert #1 
    \right\vert\kern-0.25ex\right\vert\kern-0.25ex\right\vert}}
\newcommand{\vertiiiNoLR}[1]{{\bigg\vert\kern-0.25ex\bigg\vert\kern-0.25ex\bigg\vert #1 
    \bigg\vert\kern-0.25ex\bigg\vert\kern-0.25ex\bigg\vert}}
\newcommand{\lvertiii}{\bigg\vert\kern-0.25ex\bigg\vert\kern-0.25ex\bigg\vert }
\newcommand{\rvertiii}{\bigg\vert\kern-0.25ex\bigg\vert\kern-0.25ex\bigg\vert }
\newcommand{\lnorm}[1]{\left\Vert {#1} \right\Vert}

\newcommand{\labs}[1]{\left\vert {#1} \right\vert}
\newcommand{\indicator}{\mathds{1}}

{}

\begin{document}

\title{Incompressibility and spectral gaps of random circuits}

\date{}

\author[1]{Chi-Fang Chen}

\author[2]{Jeongwan Haah}

\author[3]{Jonas Haferkamp}

\author[4]{\\Yunchao Liu}

\author[5]{Tony Metger}

\author[6]{Xinyu Tan}

\affil[1]{Caltech}
\affil[2]{Microsoft Quantum}
\affil[3]{Harvard University}
\affil[4]{UC Berkeley}
\affil[5]{ETH Zurich}
\affil[6]{MIT}

\maketitle

\begin{abstract}
Random reversible and quantum circuits form random walks on the alternating group $\Alt(2^n)$ and unitary group $\SU(2^n)$, respectively, with each random gate as one step of the walk. 
Existing bounds on the spectral gap for the $t$-th moment of these random walks have inverse-polynomial dependence in both $n$ and $t$. 
We prove that the gap for random reversible circuits is $\Omega(n^{-3})$ for all $t\geq 1$, and the gap for random quantum circuits is $\Omega(n^{-3})$ for $t \leq \Theta(2^{n/2})$.
Importantly, these gaps are independent of $t$ in the respective regimes. 
We can further improve both gaps to $n^{-1}/\mathrm{polylog}(n, t)$ for $t\leq 2^{\Theta(n)}$, which is tight 
up to polylog factors in $n$ and $t$. Our spectral gap results have a number of consequences:
\begin{enumerate}
    \item Random reversible circuits with $\cO(n^4 t)$ gates form multiplicative-error $t$-wise independent (even) permutations for all $t\geq 1$; for $t \leq \Theta(2^{n/6.1})$, we show that $\tilde \cO(n^2 t)$ gates suffice.
    \item Random quantum circuits with $\cO(n^4 t)$ gates form multiplicative-error unitary $t$-designs for $t \leq \Theta(2^{n/2})$; for $t\leq \Theta(2^{2n/5})$, we show that $\tilde \cO(n^2t)$ gates suffice.
    \item The robust quantum circuit complexity of random quantum circuits grows linearly for an exponentially long time, proving the robust Brown--Susskind conjecture~\cite{brown2018second,brandao2021models}.
    We also show an analogous result for random reversible circuits.
\end{enumerate}

Our spectral gap bounds are proven by reducing random quantum circuits to a more structured walk: a modification of the ``$\PFC$ ensemble'' from~\cite{metger2024simple} together with an expander on the alternating group due to Kassabov~\cite{Kassabov_2007_alt}, for which we give an efficient implementation using reversible circuits. In our reduction, we approximate the structured walk with local random circuits without losing the gap, which uses tools from the study of frustration-free Hamiltonians. 

\end{abstract}

\newpage

{
\hypersetup{linkcolor=black}
\setcounter{tocdepth}{2}
\tableofcontents
}

\newpage

\section{Introduction}

Among all Boolean functions on $n$-bit strings,
the vast majority has exponential circuit complexity.
This is based on a counting argument by Shannon from the 1940's~\cite{Shannon1949synthesis},
and is often summarized by saying that a random function is complex.
A similar argument carries over to the quantum setting,
showing that a Haar random unitary is complex.
Since a sufficiently deep random circuit should produce a random function, this also means that very deep circuits of randomly chosen gates
must be almost always complex, too.
This complexity question for random circuits can be made sharper
as there is a tunable parameter, the number of random gates.
It is then natural to ask if the ``true'' circuit complexity ({\it i.e.,} the size of the smallest possible circuit) of the function implemented by a random circuit
monotonically increases as the number of random gates increases,
and if so, what the precise relation between the two is.
We give arguably the simplest answer:
\emph{a random quantum (or reversible) circuit on $n$ qubits (or bits) with $L\leq 2^{\Theta(n)}$ 
gates cannot be implemented by any quantum (or reversible) circuit with less than $L/\poly(n)$ gates.}
In other words, the complexity of a random circuit grows linearly with the number of random gates for an exponentially long time, which means that random circuits are essentially incompressible.
Though this is the most obviously plausible answer,
this statement has been a conjecture~\cite{brown2018second,brandao2021models}, which is now resolved in this work.

This linear complexity growth follows from our analysis 
of the probability distribution of unitaries (or reversible functions, {\it i.e.,} permutations)
obtained by random circuits.
We give bounds on how quickly this distribution converges to the uniform distribution
on the unitary group on $n$ qubits or the alternating group on $n$-bit strings
as the random circuit deepens.
This convergence to the uniform measures has been studied for various reasons
and often goes under the name of \textit{$t$-designs} or \textit{$t$-wise independence}.
A unitary or permutation $t$-design\footnote{In the classical literature, ``$t$-wise independent permutation'' is a more conventional term, but we will use ``$t$-design'' for both permutations and unitaries throughout.}
is a distribution whose $t$-th moments are (almost) the same as those of the uniform measure.
Although the true uniform measures are hard to sample from,
there are efficiently implementable $t$-designs,
and many randomized algorithms rely on this efficiency.
Applications include randomized benchmarking~\cite{magesan2012characterizing,knill2008randomized}, classical shadows~\cite{huang2020predicting}, and derandomization~\cite{Vadhan2012Pseudorandomness}.
Beyond the algorithmic realm, 
random circuits have been considered as 
toy models of chaotic physical systems such as black holes~\cite{Hayden_2007,brown2018second}
and dynamical systems with little structure~\cite{Fisher2023Random}.
In any of these applications, it is desirable to know higher moments,
and our results give a sharper answer
as to where one can substitute dynamics-induced or circuit-induced 
randomness with uniform distributions.

Our main technical result proves that the $t$-th moment operator of the random walk 
produced by random circuits has a spectral gap of $\Delta=\Omega(1/\poly(n))$ independent of $t$.
This has been an open question in the study of random circuits 
and is known to imply that random circuits of depth $\cO(t \poly(n))$ form $t$-designs.

At a high level, our spectral gap estimate consists of two parts.
The first is to find some random, efficiently implementable walk 
on the unitary group $\SU(2^n)$ or the alternating group $\Alt(2^n)$ 
that has a $t$-independent spectral gap,
where each step is a structured quantum and reversible circuit with $\poly(n)$ gates.
An important ingredient for us here is Kassabov's generating set for an expander on $\Alt(2^n)$~\cite{Kassabov_2007_alt}. 
We show that Kassabov's generators can be implemented efficiently by $\cO(n)$ $3$-bit reversible gates and combine this with a modification of the ``PFC ensemble'' by Metger, Poremba, Sinha, and Yuen~\cite{metger2024simple}.
The second is to convert such a structured random walk 
into a circuit where individual gates are chosen randomly without losing the spectral gap property.
To this end, we use tools from the study of frustration-free Hamiltonians,
namely the detectability lemma~\cite{aharonov2009detectability,anshu2016simple} 
and its converse~\cite{gao2015quantum,o2022quantum}; see~\cref{sec:proof_overview} for the proof overview.

\subsection{Previous work} \label{sec:prior_work}

The mixing properties of random circuits have been extensively studied.
Two important notions in this context 
are moment operators and their spectral gaps, which we briefly recall.
Random reversible and quantum circuits produce a probability distribution
on the alternating and unitary groups, respectively.
The \textit{$t$-th moment operators} of a probability distribution~$\nu$ are defined to be 
\begin{align*}
    \E_{\pi \sim \nu } P(\pi)^{\ot t}\quad \text{for}\quad \pi \in \Alt(2^n) \quad\tand\quad \E_{U \sim \nu} U^{\ot t} \ot \overline{U}^{\ot t}\quad \text{for}\quad U \in \SU(2^n)\,. 
\end{align*}
Here, $\overline{U}$ is the element-wise complex conjugate of $U$ (dual representation) 
and $P(\pi)$ is the linear operator permuting the computational basis vectors, 
{\it i.e.}, $P(\pi)\ket{z} = \ket{\pi(z)}$ for all $z \in \{0,1\}^n$.
The spectral gap for a distribution $\nu$ on $\SU(2^n)$ is 
\begin{align}
\Delta:= 1 - \norm*{\E_{U \sim \nu} U^{\ot t} \ot \overline{U}^{\ot t} - \E_{U \sim \mu(\SU(2^n))} U^{\ot t} \ot \overline{U}^{\ot t}}_\infty \,. \label{eqn:spectral_gap_informal}
\end{align}
The definition for the alternating group is analogous.
An (approximate) $t$-design is a distribution on the unitary group whose $t$-th moment operator is close to the $t$-th moment operator of the Haar measure.
One can use various distance measures to define closeness, giving rise to slightly different notions of approximate designs; see \cref{section:relativevsadditive} for more details.
Similarly, $t$-designs on the alternating groups are distributions of permutations 
whose $t$-th moment operator approximates the $t$-th moment operator for uniformly random \textit{even} permutations. For most of our purposes, the distinction between the alternating group and symmetric group is not important as their moments match exactly for all $t\le 2^n-2$ (see~\cref{lemma:equivalenceSandAlt}). 
For historical reasons, $t$-designs on permutation groups are more commonly called \emph{$t$-wise independent permutations},
but in this paper we will call them \emph{permutation $t$-designs}.
The product of $\cO\left(\frac{1}{\Delta}\left(nt + \log \frac 1 \eps \right)\right)$ 
independent samples from a distribution of spectral gap~$\Delta$
is an $\eps$-approximate $t$-design on $n$ bits or qubits.

\paragraph{Random reversible circuits.}
In 1996, Gowers~\cite{gowers1996almost} proved that 
the $t$-th moment operator of random reversible circuits 
has a spectral gap of $1/\mathrm{poly}(n,t)$.
This was subsequently improved by Hoory, Magen, Myers, and Rackoff~\cite{hoory2005simple} 
to $\tilde{\Omega}(1/n^2t^2)$ for $t \leq \cO(2^{n/4})$, 
and shortly after to $\Omega(1/n^2t)$ for all $t\leq 2^n-2$ by Brodsky and Hoory~\cite{brodsky2008simple} 
for randomly drawn DES[2] gates\footnote{Here, a DES[2] gate is any gate implementing the transformation $\{0,1\}^3\to\{0,1\}^3$ via $(x,b)\mapsto (x,b\oplus f(x))$, 
where $x \in \bits^2$ and $f:\{0,1\}^2\to \{0,1\}$ is chosen uniformly at random.}. 
More recently, the gap was improved to $\tilde \Omega(1/nt)$ for all $t\leq 2^n-2$
by He and O'Donnell~\cite{he2024pseudorandom} 
using techniques developed for random quantum circuits and unitary $t$-designs.
This implies that random reversible circuits with $\tilde \cO(n^2 t^2)$ gates form approximate permutation $t$-designs.
In~\cite{he2024pseudorandom}, these gates can be arranged in a brickwork architecture with depth $\tilde \cO(n t^2)$. 
In this paper we do not consider the brickwork architecture for reversible circuits, 
although we expect that our techniques can be used to show a gap of $\tilde{\Omega}(1/n)$ for the brickwork architecture\footnote{The $\Theta(1/n)$ scaling was conjectured in~\cite[Conjecture 1]{feng2024dynamics} and was consistent with numerics~\cite[Figure 6(b)]{feng2024dynamics}.}.

\paragraph{Random quantum circuits.}
For the case of $t = 2$, Harrow and Low~\cite{harrow2009random} proved a gap of $\Omega(1/n)$, 
implying that $\eps$-approximate unitary $2$-designs are generated after $\cO(n(n+\log \frac 1 \varepsilon))$ random gates. 
For more general $t$,
Brown and Viola~\cite{BrownViola2009} calculated a series expansion of the gap in $1/n$ 
and showed that the first term is $\propto 1/n$ independent of~$t$.
Brand\~{a}o, Harrow, and Horodecki~\cite{brandao2016local} 
showed that the gap is inverse polynomial in both $n$ and $t$.
Concretely, for random quantum circuits on $n$ qubits 
they showed that the gap of $t$-th moment operator is $\Omega(1/nt^{9.5})$ for all $t\leq \Theta(2^{2n/5})$.
This was proven using techniques by Nachtergaele~\cite{nachtergaele1996spectral} 
for lowering bound the spectral gap of frustration-free Hamiltonians.
Using a similar proof strategy, Haferkamp~\cite{haferkamp2022random} later improved this bound to $\Omega(1/nt^{4+o(1)})$. 
Most of the improvement in~\cite{haferkamp2022random} 
comes from an improved $t$-independent inverse exponential gap estimate of $\Omega(n^{-5}4^{-n})$, 
which was obtained first for an auxiliary random walk 
that interleaves random Clifford unitaries with single qubit Haar random unitaries
by Haferkamp, Montealegre-Mora, Heinrich, Eisert, Gross, and Roth~\cite{haferkamp2023efficient} 
and then translated to random quantum circuits.
It was speculated in \cite{brandao2016local} that the true gap could be $1/\mathrm{poly}(n)$ for all $t\geq 1$, 
resulting in a design depth of $\mathrm{poly}(n)(n t+\log \frac 1 \varepsilon)$. 
Evidence for an inverse-polynomial and $t$-independent gap 
was provided by Hunter-Jones~\cite{hunter2019unitary}, 
who showed that the gap becomes $\Omega(n^{-1})$ in the limit of large local dimensions. 
Using Knabe bounds for spectral gaps, 
it can be shown that this behavior already sets in once the local dimensions $q$ 
satisfy $q\geq 6t^2$~\cite{haferkamp2021improved}.

Some of the motivation for a $t$-independent gap stems from a close connection to quantum circuit complexity: 
motivated by applications in black hole physics, 
Brown and Susskind~\cite{brown2018second} argued that the circuit complexity 
({\it i.e.},~the size of the smallest possible circuit that approximately implements a given unitary) 
of random quantum circuits should grow at a linear rate for an exponentially long time 
because collisions (circuits accidentally implementing the same unitary) should be rare.
See also Roberts and Yoshida~\cite{roberts2017chaos}.
While intuitive, this sparsity of collisions seems hard to prove directly.

Towards proving the Brown--Susskind conjecture for random quantum circuits, 
\sloppy Haferkamp, Faist, Kothakonda, Eisert, and Yunger Halpern~\cite{haferkamp2022linear}
showed that the \emph{exact} quantum circuit complexity 
in random quantum circuits of depth $d$ grows as $\Omega(d/\mathrm{poly}(n))$ until $d=4^n/\mathrm{poly}(n)$, 
when it saturates; see also Li~\cite{li2022short} for two short proofs of this fact.
As the name suggests, the exact circuit complexity counts the number of gates required to implement a unitary \emph{exactly}.
Unfortunately, this is too brittle a notion of circuit complexity.
A slightly more robust result for a family of random circuits (using a non-universal gate set) was obtained in~\cite{haferkamp2023moments}, 
showing a linear growth up to exponentially small implementation errors in operator norm, 
but this is still not an operational notion of circuit complexity.

For the standard \textit{robust} notion of circuit complexity where larger approximation errors are allowed,
prior work was only able to show a scaling that is worse than linear.
It was shown in \cite{brandao2016local} (see also Brand\~ao, Chemissany, Hunter-Jones, Kueng, and
Preskill~\cite{brandao2021models}) that the circuit complexity of a unitary 
drawn from a unitary $t$-design is $\Omega(t)$ with high probability for $t\leq 2^{\Theta (n)}$ (see~\cref{section:circuitcomplexity}), 
implying that the circuit complexity of random quantum circuits with $L$ gates
must be at least $\Omega(L^{1/11}/\poly(n))$; 
the exponent~$1/11$ can be improved using~\cite{haferkamp2022random}. In addition, Oszmaniec, Kotowski, Horodecki, and Hunter-Jones~\cite{oszmaniec2024saturation} studied the long-time behavior of complexity, showing that complexity saturates at maximal value after exponential time and undergoes recurrences after double exponential time, for both random quantum circuits and stochastic local Hamiltonian evolution. This was also conjectured in~\cite{brown2018second}.

\paragraph{Other constructions of $t$-designs.}

The gap of random quantum circuits is used in multiple constructions 
to prove the efficient generation of designs for other random processes.
This includes stochastic Hamiltonian dynamics~\cite{lashkari2013towards}, 
where the local Hamiltonian randomly fluctuates depending on time~\cite{onorati2017mixing}.
Jian, Bentsen, and Swingle~\cite{Jian2023} argued that certain random time-dependent Hamiltonian evolution
converges to $t$-designs at a linear rate.
See also~\cite{guo2024complexity,tang2024brownian}.
The gap estimate in~\cite{brandao2016local} 
can also be used to show that cluster states randomly measured in the $X$-$Y$ plane 
generate approximate $t$-designs on the last unmeasured column~\cite[Appendix]{haferkamp2022randomness}.
In particular, our result improves the scaling in both of these settings to linear.
Previously, the work of Nakata, Hirche, Koashi, and Winter~\cite{nakata2017efficient} 
showed a linear scaling in $t$ for a family of stochastic Hamiltonians in the regime $t\leq \cO(\sqrt{n})$ 
without embedding random quantum circuits.

Recently, multiple advances were made on efficient constructions of approximate unitary designs 
that do not rely on random quantum circuits.
Haah, Liu, and Tan~\cite{haah2024efficient} 
showed that a circuit ensemble based on ``random Pauli rotations'' $e^{i\theta P/2}$ 
with a uniformly random $n$-qubit Pauli operator $P$ and a uniformly random angle $\theta\in (-\pi,\pi)$, 
has a gap of $1/(4t)+1/(4^n-1)$ for \emph{all} $t\geq 1$. This achieves \textit{multiplicative-error} approximate $t$-designs 
with circuits of depth $\cO(n t^2 \log n)$ using all-to-all gates.
Chen, Docter, Xu, Bouland, and Hayden~\cite{chen2024efficient_sums} 
considered products of two exponentials of random matrix sums, 
giving an overall gate complexity $\tilde{\cO}(n^2 t^2)$ 
for $t \leq 2^{\cO(n/\log n)}$ 
for approximate $t$-designs, 
but only for \textit{additive-errors} 
(which is a weaker notion of approximation that operationally corresponds to parallel applications of the unitary).
Shortly after, \cite{metger2024simple,chen2024efficient} achieved constructions of approximate unitary $t$-designs 
with $\cO(\poly(n)t)$-depth circuits and additive error $\mathrm{poly}(t)/2^n$. The construction from~\cite{metger2024simple} has a structure suitable for our setting and plays an important role in our spectral gap analysis.
In comparison, we show a $\tilde\Omega(1/n)$ gap for $t\leq \Theta(2^{2n/5})$, 
resulting in a multiplicative-error approximate design of circuit depth $\tilde{\cO}(nt)$.
Furthermore, in contrast to~\cite{metger2024simple,chen2024efficient}, 
we show this for random quantum circuits, not specially constructed ones.
We refer to~\cref{section:relativevsadditive} for details on the difference 
between multiplicative and additive errors in $t$-designs.

To our knowledge, when converted to constant multiplicative error, 
the proven upper bound on circuit sizes for all existing constructions of permutation and unitary $t$-designs 
have at least quadratic dependence in $t$. 
We achieve linear dependence in $t$ (even for exponentially small multiplicative error) due to the $t$-independence of our spectral gap.

\paragraph{Note added.} 
During the preparation of this manuscript, we became aware of independent work by Gretta, He, and Pelecanos~\cite{Gretta2024More}, 
who prove bounds for the mixing of random reversible circuits using log-Sobolev inequalities instead of spectral gaps. 
While not mentioned in their paper, their result also implies a linear growth rate for the complexity of random reversible circuits 
(analogous to \cref{cor:lineargrowthreversible}).

Separately, we became aware that the recent work of He and O'Donnell~\cite{he2024pseudorandom} 
proves a similar overlap result~\cite[Lemma 61]{he2024pseudorandom}
as our~\cref{thm:permutationoverlap}, with different techniques.

\subsection{Results: spectral gaps, \texorpdfstring{$t$}{t}-designs, and circuit complexity growth} \label{sec:results_overview}

The spectral gap of a distribution of unitaries is defined in \cref{eqn:spectral_gap_informal}; 
see also \cref{sec:essential_norm} for more details.
Here, we give the formal statements of our bounds on the spectral gap of random reversible and quantum circuits (\cref{thm:main_rev} and \cref{thm:main_quantum}), as well as the consequences of these bounds, including: (1) random quantum circuits with $\tilde \cO(n^2 t)$ gates form $t$-designs and (2) a robust version of the Brown--Susskind conjecture which says that the complexity of random quantum circuits grows linearly for an exponentially long time.
We remark that all the constants in this paper can (in principle) be estimated explicitly.

\subsubsection{Random reversible circuits}

A random reversible circuit on $n$ bits consists of a sequence of elementary gates acting on $3$ bits. 
Each gate is sampled randomly and independently as follows: pick a uniformly random permutation from $\Sym(2^3)$ and apply it to 3 randomly selected bits of the input string.\footnote{Note that even though we sample the local 3-bit permutation uniformly from $\Sym(2^3)$, not $\Alt(2^3)$, 
applying this local permutation to a subset of the $n > 3$ bits 
and leaving all other bits intact results in an even permutation.}
This procedure can be viewed as a random walk, with each step being a randomly sampled gate from the aforementioned distribution denoted as $\nurevalltoall$.
The following theorem compares $\nurevalltoall$ to the uniform measure on $\Alt(2^n)$ (denoted by $\mu(\Alt(2^n))$), establishing a $t$-independent spectral gap of $\Omega(n^{-3})$.

\begin{restatable}[Spectral gap for random reversible circuits]{theorem}{gaprev}
\label{thm:main_rev}
For all integers $n \geq 4$ and $t \geq 1$, we have 
\begin{align}\label{eq:main_reversible}
\norm*{
    \avg_{\pi \sim \nurevalltoall} P(\pi)^{\ot t}
    - 
    \avg_{\pi \sim \mu(\Alt(2^n))} P(\pi)^{\ot t}
}_\infty \leq 1 - \Omega(n^{-3}) \,.
\end{align}
\end{restatable}

This theorem is stated for random $3$-local permutation gates with all-to-all connectivity.
We can show a similar $1/\poly(n)$ gap for local random reversible circuits on any (connected) circuit layout via~\cref{lem:reversible_arbitrary_arch}.
The spectral gap in \cref{eq:main_reversible} is independent of $t$, but decays like $n^3$.
We can reduce the $n$ dependence to linear at the expense of a $\mathrm{polylog}$ factor in $nt$.

\begin{theorem}[Gap amplification]\label{thm:bootstrappedreversiblegap}
    For all positive integers $t\leq \Theta(2^{n/6.1})$, 
    the left-hand side of \cref{eq:main_reversible} is $\le 1-\Omega( n^{-1}\log(nt)^{-3})$.
\end{theorem}

The reduction from \cref{thm:main_rev} to \cref{thm:bootstrappedreversiblegap} 
relies on an ``overlap theorem'' for random permutations (\cref{thm:permutationoverlap}), 
which shows that the product of two random permutations on two different (but overlapping) subsets of bits 
approximates a random permutation on all of the bits.
An analog of this overlap lemma for random unitaries was previously known~\cite{brandao2016local};
our permutation overlap lemma may be of independent interest.

One important consequence of our $t$-independent spectral gap result 
is that random reversible circuits form an approximate permutation $t$-design with respect to multiplicative error (see~\cref{def:permutation_designs})
after only $\tilde \cO(n^2 t)$ gates. 

\begin{corollary}[Approximate permutation $t$-design in linear depth]
    For all integers $n\geq 4$ and $t\leq \Theta(2^{n/6.1})$,
    random reversible circuits on $n$ bits with $\cO(n(nt+\log \frac 1 \eps)\log^{3}(nt))$ random gates 
    are approximate permutation $t$-designs with multiplicative error $\eps$.
\end{corollary}

We also remark that for structured rather than random circuits, we can construct approximate permutation $t$-designs in depth $\cO(nt)$ (without any $\log$-factors), which might be of independent interest (see~\cref{sec:nt_depth_permutations}).

\subsubsection{Random quantum circuits}

We focus on two models of random quantum circuits, defined as follows.

\begin{definition}[Random quantum circuits] \label{def:rqc}
~
We define two distributions on $\SU(2^n)$: 
\begin{itemize}
    \item \emph{$k$-local all-to-all random quantum circuits}: $\nu_{k,\mathrm{All}\to \mathrm{All},n}$ is defined by first picking a subset $S\subseteq [n]$ with $|S|=k$ uniformly at random and then applying a Haar random unitary $U_S\in \SU(2^k)$ to the qubits in $S$.

    \item \emph{Brickwork random quantum circuits}: Assume for simplicity that $n$ is even.
    $\nu_{\mathrm{BRQC},n}$ is defined by first applying a unitary $U_{1,2}\otimes U_{3,4}\otimes \cdots \otimes U_{n-1,n}$ and then a unitary $U_{2,3}\otimes \cdots \otimes U_{n,1}$, where each $U_{i,i+1}$ is drawn independently from the Haar measure on $\SU(4)$.
    Here, we identify $n+1$ with $1$.
\end{itemize}
We refer to $\nu_{\mathrm{BRQC},n}^{*k}$ as a brickwork random quantum circuit of depth $k$ (see \cref{fig:brqc} for an illustration).\footnote{For two probability measures $\nu_1$ and $\nu_2$ on a group, $\nu_1 * \nu_2$ is their \emph{convolution},
the probability distribution of the product $U_1 U_2$ for $U_1 \sim \nu_1$ and $U_2 \sim \nu_2$. 
The notation $\nu^{*k}$ means the $k$-fold convolution of $\nu$ with itself.}
We will often drop the subscript $n$ for simplicity.
\end{definition}

Brickwork random quantum circuits are also sometimes called parallel random quantum circuits.
If we interpret the probability measures above in the context of random walks, then one step of $\nu_{k,\mathrm{All}\to \mathrm{All},n}$ corresponds to applying just one quantum gate, whereas one step of $\nu_{\mathrm{BRQC},n}$ corresponds to applying a total of $n$ gates.

Brickwork circuits are a natural model for quantum devices with $1$D locality constraints on the qubits, 
so we will usually state our main results for this model of random quantum circuits.
However, our analysis can be generalized to arbitrary architectures via~\cref{lem:quantum_arbitrary_arch}, resolving a conjecture from~\cite{mittal2023local}.
The key idea in the proof is to replace the qubit permutation with a sequence of SWAP gates drawn uniformly at random corresponding to the connectivity of qubits. 
We then analyze this process using some well-known results on the eigenvalues of Cayley graphs on the symmetric group generated by transpositions. 

\begin{figure}
    \centering
    \includegraphics[width=0.75\textwidth]{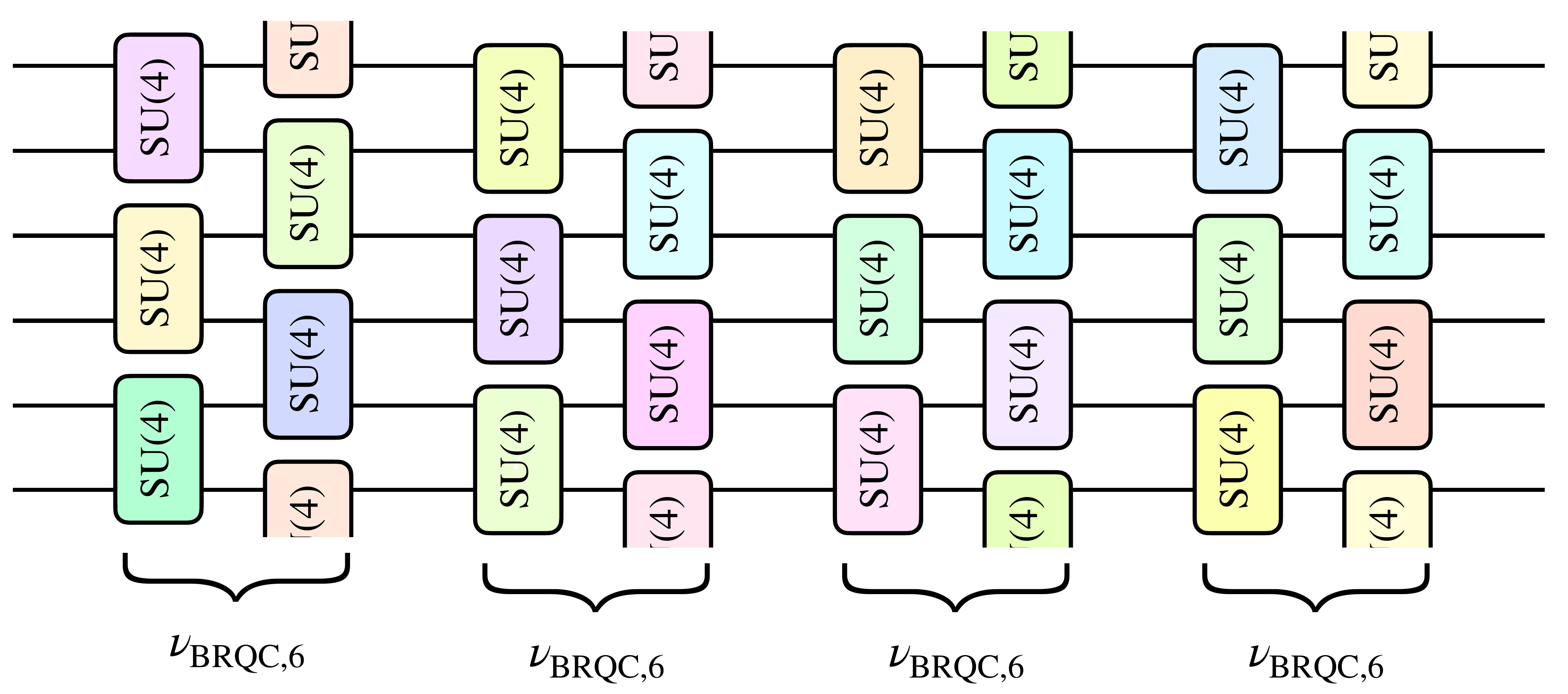}
    \caption{A depth-4 brickwork random quantum circuit, described by the probability measure $\nu_{\mathrm{BRQC},6}^{*4}$}
    \label{fig:brqc}
\end{figure}

Our main result is that both models of random quantum circuits have $t$-independent spectral gaps:
\begin{restatable}[Spectral gap for random quantum circuits]{theorem}{gaprqc}
\label{thm:main_quantum}
For all integers $n\geq 2$ and $t \leq \Theta(2^{n/2})$, we have 
\begin{align}
\norm*{\avg_{U \sim \nu_{\mathrm{2,\mathrm{All}\to\mathrm{All}},n}} U^{\ot t} \ot \overline U^{\ot t} - \avg_{U \sim \mu(\SU(2^n))} U^{\ot t} \ot \overline U^{\ot t}}_\infty &\leq 1 - \Omega(n^{-3})\,, \label{eq:main_quantum_all} \\
\norm*{\avg_{U \sim \nu_{\mathrm{BRQC},n}} U^{\ot t} \ot \overline U^{\ot t} - \avg_{U \sim \mu(\SU(2^n))} U^{\ot t} \ot \overline U^{\ot t}}_\infty &\leq 1 - \Omega(n^{-5}/\log n)\,. \label{eq:main_quantum_brickwork}
\end{align}
\end{restatable}
\noindent

Using known bootstrapping methods for the gaps~\cite{brandao2016local,haferkamp2021improved}
together with the bound in~\cref{thm:main_quantum} readily implies the following corollary.
\begin{corollary}[Gap amplification]\label{corollary:main}
    For all positive integers $t\leq \Theta(2^{2n/5})$, 
    the left-hand side of~\cref{eq:main_quantum_all} is $\leq 1-\Omega(n^{-1}(\log t)^{-4})$ 
    and the left-hand side of~\cref{eq:main_quantum_brickwork} is $\leq 1-\Omega((\log t)^{-7})$.
\end{corollary}
Further, the gap estimates in~\Cref{thm:main_quantum} and~\Cref{corollary:main} lead to multiplicative-error approximate unitary $t$-designs with optimal $t$-dependence up to polylog-factors (see~\cref{section:relativevsadditive} for the precise definitions). 

\begin{corollary}[Random quantum circuits are linear unitary $t$-designs]\label{cor:designbounds}
For all integers $n\geq 2$ and $t\leq \Theta(2^{2n/5})$,
\begin{itemize}
    \item $2$-local all-to-all random quantum circuits with $L=\cO(n(nt+\log \frac 1 \varepsilon)(\log t)^4)$ random gates and 
    \item brickwork random quantum circuits of depth $D=\cO((nt+\log \frac 1 \varepsilon)(\log t)^7)$
\end{itemize}
form approximate unitary $t$-designs with multiplicative error $\varepsilon$. Note that both bounds can be made $\cO(\poly(n) \, t)$ ({\it i.e.}, exactly linear in $t$) by inserting $t \leq \Theta(2^{2n/5})$.
\end{corollary}
It was observed in~\cite{brandao2016local} that~\cref{corollary:main} together with~\cite{bourgain2012spectral} implies a similar design depth for random quantum circuits with gates drawn from any universal gate set that contains inverses and algebraic matrix entries, although in this case we can no longer obtain explicit constants for the bounds in \cref{thm:main_quantum}.
Moreover, at the expense of some additional $\mathrm{polylog}(t)$ factor, we can obtain a similar design depth for gates drawn from any (fixed) universal probability distribution on $\SU(4)$ using~\cite{oszmaniec2021epsilon,varju2012random}.

\subsubsection{Circuit complexity growth}
For $\pi\in \Alt(2^n)$, denote by $C_R(\pi)$ the minimum number of $3$-bit reversible gates required to implement $\pi$.\footnote{Note that an odd permutation on $n$ bits with $n \ge 4$ cannot be implemented by a circuit of $3$-bit reversible gates unless one uses an ancilla bit.}
We show that random reversible circuits are not compressible (by more than a polynomial multiplicative factor in the system size $n$) up to an exponential depth.
\begin{restatable}[Linear growth of reversible circuit complexity]{corollary}{revlineargrowth}
\label{cor:lineargrowthreversible}
    Let $\pi$ be a random reversible circuit on $n$ bits with $L\leq \cO(2^{n})$ gates. 
    The reversible circuit complexity $C_R(\pi)$ must satisfy
    \begin{equation}
        C_R(\pi)\geq \Omega\left(\frac{L}{n^3\log n}\right),
    \end{equation}
with probability at least $1-2^{-\Omega(L/n^3)}$ over the choice of $\pi$.
\end{restatable}

The (robust) quantum circuit complexity $C_{Q,\delta}(\ket{\psi})$ of a state $\ket{\psi}$ is the smallest integer $R$ such that there exists a quantum circuit $V$ with $R$ two-qubit gates that satisfies $|\langle\psi|V|0^n\rangle|^2\geq 1-\delta^2$.
We can similarly define the circuit complexity for unitaries: $C_{Q,\delta}(U)$ is the smallest number of 2-qubit gates required to implement $U$ within $\delta$ error (see \cref{def:quantumcircuitcomplexity}).
Here, we state a circuit lower bound for states $\ket{\psi}=U\ket{0^n}$, where $U$ is a random quantum circuit. Note that this implies a lower bound on $C_{Q,\delta}(U)$ since $C_{Q,\delta}(U)\geq C_{Q,\delta}(U\ket{0^n})$.

\begin{restatable}[Linear growth of robust quantum circuit complexity]{corollary}{quantumlineargrowth}
\label{cor:lineargrowth}
    Let $\delta\in(0,1/2)$ be a constant and $U$ be a random quantum circuit on $n$ qubits with all-to-all connectivity and $L\leq\cO(2^{n/2})$ gates. The quantum circuit complexity $C_{Q,\delta}(U\ket{0^n})$ must satisfy
 \begin{equation}
     C_{Q,\delta}(U\ket{0^n})\geq \Omega\left(\frac{L}{n^4}\right),
 \end{equation}
with probability at least $1-e^{-\Omega(L/n^3)}$ over the choice of $U$. The big-$\Omega$ notation absorbed dependence on $\delta$.
\end{restatable}
\Cref{cor:lineargrowth} resolves the Brown--Susskind conjecture for random quantum circuits~\cite{brown2018second,brandao2021models}.
In addition, we remark that \Cref{cor:lineargrowth} also implies linear lower bounds on Nielsen's geometric complexity measure~\cite{nielsen2005geometric}.

Note that similar bounds also hold for other definitions of circuit complexity, 
such as the minimum circuit size required to implement a measurement 
that distinguishes $U$ from the completely depolarizing channel~\cite{brandao2021models}.
In addition, one may relax the definition of circuit complexity to allow nonunitary gates 
and still have essentially the same lower bound on complexity;
see~\cref{rem:nonunitaryGates}.

\subsection{Proof overview} \label{sec:proof_overview}
In this section, we highlight a few technical components of our proof. 
As mentioned, the many attempts to analyze the spectral gap directly often acquire an extra $t$-dependence. 
Instead of jumping into the spectral gap analysis of random circuit models, we turn to more \textit{structured} ensembles. 
For reversible circuit, the starting point is Kassabov's construction of an expander on the alternating group; for unitary circuits, we draw from the recent ``PFC'' ensemble that constructs a (pseudo)-random unitary from (pseudo)-random permutations. 
To translate the $t$-design properties of these structured ensembles into spectral gaps of unstructured random circuits, we further need an efficient decomposition of the structured ensembles into elementary gates.
This process includes plugging in Kassabov's choice of generating set and using tricks from the study of frustration-free Hamiltonians extensively. 
Finally, we apply ``overlap'' lemmas that further improve the $n$-dependence at the expense of $\mathrm{polylog}(t)$ for both quantum and reversible circuits. 
We discuss each component as follows:

\paragraph{Kazhdan constants and Kassabov's expander.}

Similar to spectral gaps, Kazhdan constants also characterize the mixing time of a random walk. 
In~\cite{Kassabov_2007_alt},  
Kassabov provided a family of generating sets $S_N$ for $\Alt(N)$ where $\abs{S_N} = \cO(1)$,
whose Kaszhdan constants are uniformly bounded away from zero. 
This implies that the spectral gaps of the corresponding moment operators are independent of $N$ and $t$. 
The case of $N=2^n$ is particularly relevant to us, as it can be interpreted as (even) permutations of the bitstrings $\{0,1\}^n$. 

While it is conceivable that Kassabov's result
restricted to $\{\Alt(2^n)\}_{n=4}^\infty$ 
gives a useful family of $\poly(n)$ size circuits,
we use intermediate groups $\Alt((2^{3s}-1)^6)$ where $s\geq 1$ is an integer.
In fact this is Kassabov's main focus.
Since $\Alt((2^{3s}-1)^6)$ is a subgroup of $\Alt(2^{18s})$, 
we can work with $18s$ bits directly and think of the group as permuting elements 
in $\mathbf K_s = (\{0,1\}^{3s}\setminus \{0^{3s}\})^{\times 6}$. 
While the existence of these generators was the focus of Kassabov, 
their computational complexity is important for us.
We show that the action of each generator in this scenario 
has a circuit depth-$1$ implementation using $\cO(n)$ CNOT and Toffoli gates. 

Using NOT and multiply controlled NOT gates and the property of bounded generation, 
we extend $\Alt((2^{3s}-1)^6)$ to $\Alt(2^{18s})$, and then to $\Alt(2^n)$ for arbitrary $n$.
Eventually, we show that every generator of $\Alt(2^n)$ 
can be decomposed as a product of $\cO(n)$ NOT, CNOT, and Toffoli gates (see~\cref{thm:UltimatePermutationDesign}). 

\paragraph{CPFPC.}
Recently, \cite{metger2024simple} introduced the ``PFC ensemble'' as an approximation to the Haar measure up to exponentially high moments that is based on classical random functions. Here, P stands for a uniformly random permutation of the computational basis states, 
F stands for a random phase 
({\it i.e.}, a random unitary diagonal in the computational basis with $\pm$ entries on the diagonal), 
and C stands for a random Clifford unitary. 
In our reduction from structured to random circuits, we replace the ideal permutations with short reversible circuits and the Cliffords with short quantum circuits. 

In particular, \cite{metger2024simple} proved that the $t$-fold twirl 
over the PFC ensemble is close to that over the Haar measure in diamond distance (which implies closeness in the $1 \to 1$ norm). 
Although the $1 \to 1$ norm does not immediately give a favorable bound on the spectral gap,
the noncommutative Riesz--Thorin theorem
tells us that the spectral gap that comes from $2 \to 2$ norm of the associated mixed unitary channel
is more directly related to the $1 \to 1$ norm bound,
\emph{if the channel is self-adjoint}.
We observe that this self-adjointness can be safely 
imposed by appending a unitary design with the ensemble of the inverses (see~\cref{cor:AdditiveToMultiplicativeError}).
Hence, we show that the extended ``CPFPC'' ensemble has a spectral gap of $1-\cO(t2^{-n/2})$ (see \cref{lemma:gapofCPFPC}), where each of the five components is sampled independently 
and the entire ensemble is invariant under taking inverses.

\paragraph{Decompose into local gates.}
    
Since each Kassabov's generator is a product of $\cO(n)$ $3$-bit gates, we can naturally break it down into local reversible gates. 
We show that the new walk has a spectral gap of $\Omega(n^{-3})$, which is still independent of $t$ (see~\cref{lemma:Sgateisgapped}). 
    
For random quantum circuits, we also substitute each component in the CPFPC ensemble by gapped random circuits from local subgroups:
\begin{itemize}
    \item We can replace C with any approximate unitary $2$-design with constant multiplicative error; here we use a linear depth brickwork random quantum circuit.

    \item F is not generated by local random gates; as a diagonal group, F does not have an expanding generating set.
    Instead, we ``simulate'' the phases by the ensemble $\mathrm{P}Z\mathrm{P}^{-1}$, 
    which conjugates a single-qubit Pauli $Z$ by a uniformly random permutation of bitstrings.
    We show that $\mathrm{P}Z\mathrm{P}^{-1}$ has a gap of $1-\cO(t^2 2^{-n})$ in F (see~\cref{lemma:gapforPZP}).
    The inverse $\mathrm{P}^{-1}$ can then be absorbed into the independent random permutation in CPFPC, resulting in a ``CPZPC'' ensemble (see~\cref{cor:gap_CPZPC}).

    \item P can be replaced by sampling from the $3$-bit gates, which show up in the implementation of Kassbov's generators as described above, for $\cO(n^3)$ times. 
\end{itemize}

\paragraph{Tools from the study of frustration-free Hamiltonians.}
The above consideration gives a distribution over reversible circuits, consisting of particular $3$-bit gates, with a $t$-independent spectral gap of $\Omega(n^{-3})$.
We relate the gap of this more structured walk to that of random reversible circuits using 
the detectability lemma~\cite{aharonov2009detectability,aharonov2011detectability,anshu2016simple} and its converse, Gao's quantum union bound~\cite{gao2015quantum,o2022quantum}. 
These tools enable conversion between the gaps of the convolutions and the convex combinations of local distributions 
(see~\cref{sec:hamiltonian_tools}).
Convolutions of local distributions inevitably show up when we amplify gaps, 
but convex combinations over local distributions allow us to use operator inequalities to compare spectral gaps. 
Using convex combinations,
we can safely substitute each local subgroup with $\Sym(8)$ without decreasing the spectral gap,
and by permuting qubits uniformly at random, $\Sym(8)$ at a fixed position can be 
regarded as one that acts on a random triple of bits.

A similar technique works for quantum circuits, 
where we substitute each local subgroup on $3$ qubits 
with $\SU(4)$ groups on $2$ qubits (see~\cref{lemma:reductiongeneraltorqc}).
Similar arguments are also used to prove bounds on spectral gaps 
with respect to arbitrarily connected architectures.

\paragraph{Small-size reductions for gaps.}
Once we have a $1/\poly(n)$ gap for random reversible and quantum circuits, 
we show that the $n$-dependence can be improved to $\tilde \Omega(1/n)$.
For random quantum circuits, this type of reduction already exists:
see \cite{brandao2016local} for the case of brickwork circuits using spectral gap bounds for frustration-free Hamiltonians~\cite{nachtergaele1996spectral},
and \cite{haferkamp2021improved} for the case of $2$-local all-to-all circuits using a recursion over $n$.

For random reversible circuits, we give similar reductions.
We prove a permutation analogue of the ``overlap lemma'' in~\cite{brandao2016local}, 
which concerns products of overlapping chunks of unitary (see~\cref{thm:permutationoverlap} in~\cref{section:permutationoverlap}). 
Curiously, the existing strategy for the unitary case (known as the approximate orthogonality of permutations) 
cannot be applied to the case of symmetric groups (\cref{rem:partition_non_orthogonal}).
To circumvent this issue, we use a recently introduced ``large-$N$-interpolation'' principle~\cite{chen2024efficient,chen2024efficient_sums,chen2024new} that avoided fine-grained combinatorial calculation by interpolating from the limit where the overlap is infinitely large.

\subsection{Discussion and outlook}

In this work we prove near optimal gaps of $\tilde \Omega(1/n)$ for the $t$-th moment operator of random reversible circuits for $t\leq \Theta(2^{n/6.1})$ and random quantum circuits for $t\leq \Theta(2^{2n/5})$.
This is proven for all-to-all random reversible classical circuits and both all-to-all and brickwork random quantum circuits.
We also remark that our overlap theorem (\cref{thm:permutationoverlap}) can likely be used together with the martingale technique for spectral gaps of frustration-free Hamiltonians~\cite{nachtergaele1996spectral} to prove a spectral gap of $\tilde \Omega(1/n)$ for random reversible circuits in a \emph{$1$D brickwork layout}, improving upon the $\tilde \Omega(1/t)$ spectral gap in~\cite{he2024pseudorandom} for brickwork reversible circuits.
This spectral gap estimate implies approximate unitary $t$-designs and approximate $t$-wise independent permutations after $\tilde \cO(n^2t)$ gates.

As a consequence, we show that the robust quantum circuit complexity grows at a linear rate for an exponentially long time, which was conjectured by Brown and Susskind~\cite{brown2018second,brandao2021models}. 
Moreover, we prove an analogous statement for random reversible circuits, showing that random circuits are effectively incompressible with overwhelming probability.
Our gap estimate for random quantum circuits also improves bounds on the mixing time of other random processes on the unitary group such as the dynamic of stochastic Hamiltonians~\cite{onorati2017mixing}.
Moreover, the fast convergence to designs allows us to improve predictions about random quantum circuits that involve strong concentration properties~\cite{low2009large}.
For example, our designs bounds together with the results of~\cite{cotler2022fluctuations} show that the fluctuations of the entropy $S(\rho_A(d))$ decay super-exponentially fast in the depth, where $A\subset [n]$ with $|A|=\cO(1)$ and $\rho_A=\mathrm{Tr}_{[n]\setminus A}[|\psi\rangle\langle\psi|]$, where $|\psi\rangle=U_d\cdots U_1 \ket{0^n}$ is generated by random quantum circuits of $d$ gates.
More precisely, we find 
\begin{equation}
    \Pr[S(\rho_A(d))\leq |A|-\delta]\leq (e^{2\delta}d\cdot 2^{-n})^{d/\poly(n)}.
\end{equation}

While our spectral gap estimates $n^{-1}/\mathrm{polylog}(n,t)$ are optimal up to polylog factors
(see, {\it e.g.}, the discussion in~\cite{mittal2023local}),
there are multiple avenues for future work.
In particular, it would be desirable to lift the assumption $t\leq \Theta(2^{n/2})$.
We use this condition many times and it is precisely the regime in which the permutations are approximately orthogonal~\cite{harrow2023approximate}.
In this regime, a $t\times t$ submatrix behaves approximately Gaussian, so that the $t$-th moments cannot see the correlations between Haar random matrix entries.
We cannot rule out that this condition is fundamental to the $1/\poly(n)$-gap.
On the other hand, it was suggested that the gap might be $1/\poly(n)$ for all $t\geq 1$ in~\cite{brandao2016local}, which is the quantum version of the result we have proven for random reversible circuits.

A more ambitious open problem regards the precise value of the gap for random quantum circuits.
It was pointed out to us by Nick Hunter-Jones that numerically 
the gap for large $t$ appears to be equal to the gap for $t=2$.
In other words, additional irreps of $\SU(2^n)$ 
beyond those already appearing in $U^{\otimes 2}\otimes \overline{U}^{\otimes 2}$ do not appear to contribute to the gap.
A similar phenomenon has been observed in~\cite{haah2024efficient}:
the spectral gap of random Pauli rotations on~$\SU(2)$ reaches the minimum value~$5/12$ at $t=4$.
A high level explanation for this phenomenon might be the growing size of the irreps: 
the gap of brickwork random circuits is the cosine of the angle between trivial representation spaces 
of different embeddings of $\SU(4)^{\times n/2}$ into $\SU(2^n)$.
In large dimensional representations the angles tend to be close to $90^{\circ}$ 
in the sense that the expected overlap of two random vectors is $1/\dim$.
We do not know, however, how to make this intuition precise for random quantum circuits.
We note that this phenomenon seems to be common.
Indeed, the spectral gap of Kac's random walk is equal to the case $t=4$~\cite{maslen2003eigenvalues} and similarly, the gap of random transpositions in $\Sym(N)$ is equal to the case $t=1$ 
as conjectured in 1994 by Aldous~\cite{david1995reversible} and proven in~\cite{caputo2010proof}.

The translation from spectral gaps to circuit complexity via $t$-designs 
seems to loose a factor in $n$.
For random quantum circuits on $D$-dimensional lattice, 
the design depth (in diamond distance) was improved to $n^{1/D}\mathrm{poly}(t)$ in~\cite{harrow2023approximate}.
Upcoming work~\cite{schuster2024random} constructs ensembles of unitaries, 
which generate multiplicative-error approximate unitary designs in depth $\cO(\log(n)\mathrm{poly}(t))$ 
on any geometry including a $1$D line.

\paragraph{Acknowledgements.}
We thank Aram Harrow for pointing out the reference~\cite{Kassabov_2007_alt} 
and Martin Kassabov for pointing out the reference~\cite{caprace2023tame}.
We thank Jonas Helsen, Nick Hunter-Jones, Nikhil Srivastava, and John Wright for helpful discussions. 
Most of this work was done while the authors were visiting the Simons Institute for the Theory of Computing, supported by DOE QSA grant \#FP00010905 and NSF QLCI Grant No. 2016245. C.-F.C. is supported by a Simons-CIQC postdoctoral fellowship through NSF QLCI Grant No. 2016245. J. Haferkamp acknowledges funding from the Harvard Quantum Initiative.
Y.~Liu~is supported by DOE Grant No.~DE-SC0024124, NSF Grant No.~2311733, DOE Quantum Systems Accelerator, and NSF QLCI program through grant number OMA-2016245. 
T. Metger acknowledges support from the ETH Quantum Center, NCCR SwissMAP, SNSF Grant No.~200021\_188541, and an ETH Doc. Mobility Fellowship. 
X. Tan~is supported by NSF Grant No. CCF-1729369 and by the U.S. Department of Energy, Office of Science, National Quantum Information Science Research Centers, Co-design Center for Quantum Advantage (C2QA) under contract number DE-SC0012704.

\section{Quantifying mixing rates}\label{sec:QuantifyingMixingRates}

By definition, an approximate unitary design is a distribution that resembles the Haar distribution on a unitary group. Sometimes, random draws from an approximate unitary design may be applied in sequence, {\it i.e.}, as a step in a random walk, to give even better designs.
The question of how well it reproduces $t$-th moments then
amounts to the convergence rate or mixing time of the random walk to the Haar distribution.

In this section, we give a minimal introduction to two quantities relevant to the mixing of random walks on compact groups:
\textit{Kazhdan constants}
and
\textit{spectral gaps}.
They are conceptually equivalent
and are quantitatively related to each other, as we will show.
Different results will be better described by one of the notions than the other,
and inequalities in this section will enable conversions between them. Most statements are standard, but we include self-contained proofs and rewrite them in our language. 

We will make an observation (\cref{cor:AdditiveToMultiplicativeError}) in this section,
that we use importantly later.
This appears to be not used elsewhere previously in the context of $t$-designs.
Our observation basically turns an approximate design 
that is good in terms of the trace distance (additive error)
into one that is also good in terms of a more stringent metric (spectral gap)
without any dimension-dependent factor.

\subsection{Essential norm of moment operators and the spectral gap}\label{sec:essential_norm}

We define the \emph{essential norm} of a probability measure on a group together with a representation, which is our main metric of interest.

\begin{definition}\label{def:moment_gap}
    For a finite or compact Lie group~$G$,
    we denote by~$\mu(G)$ or~$\mu_G$ the Haar probability measure on~$G$.
    If $\nu$ is a probability measure on~$G$ and 
    $\rho$ is a finite-dimensional unitary representation of~$G$ that may be reducible,
    we define
    \begin{align}
        \text{a moment operator }\quad M(\nu,\rho,G) &\deq \avg_{U \sim \nu} \rho(U), \\
        \text{and the \textit{essential norm} } \quad g(\nu,\rho,G) &\deq \norm*{M(\nu,\rho,G) - M(\mu_G,\rho,G)}_\infty. \nonumber
    \end{align}
    Here the norm is the operator norm, the largest singular value.
    The \emph{spectral gap} of~$\nu$ in~$\rho$ is defined to be
    \begin{align}
        \Delta(\nu,\rho,G) = 1 - g(\nu,\rho,G).
    \end{align}
    If $G \subseteq \SU(N)$ is a subgroup of the special unitary group, then we overload the notations to write for any integer $t \ge 1$
    \begin{align}
        \text{$t$-th moment operator } \quad M(\nu,t,G) &= M(\nu,\tau_{t,t},G),\\
        g(\nu,t,G) &= g(\nu,\tau_{t,t},G) \nonumber
    \end{align}
    where $\tau_{t,t}$ is a tensor representation
    \begin{align}
        \tau_{t,t}(U) = (U \otimes \overline U)^{\otimes t} .
    \end{align}
\end{definition}

The latter is the reason that we call $M(\nu,\rho,G)$ a moment operator.
The appearance of~$G$ in $M(\nu,\rho,G)$ and $g(\nu,\rho,G)$ is redundant because both $\nu$ and $\rho$ carry~$G$ with them.
However, we will often consider a probability measure~$\nu$ that is defined on a subgroup of~$G$,
in which case including $G$ in the notation is helpful as a reminder to view $\nu$ as a probability measure on $G$, not just on the subgroup.

\begin{remark}\label{rmk:moment_projector}
    For any unitary representation~$\rho$ of a finite or compact Lie group~$G$, 
    the moment operator $M(\mu_G,\rho,G)$ with respect to the Haar probability measure
    is the orthogonal projector onto the trivial subrepresentation of~$\rho$,
    which may be zero.
    So, the essential norm is precisely the norm of the average of the represented operators
    restricted to the orthogonal complement of the trivial subrepresentation of~$\rho$.
    An exact unitary $t$-design is one whose essential norm of the $t$-th moment operator is zero.
    The spectral gap in~\cite{bourgain2012spectral} is the infimum of our spectral gap 
    over all finite dimensional unitary representations.
\end{remark}

\begin{remark}\label{rem:conv}
A random walk on~$G$ is a sequence of random steps.
If each step~$i$ is defined by a distribution~$\nu_i$ on~$G$,
then after $n$ steps the distribution of the walker 
is the convolution~$\nu_n * \nu_{n-1} *\cdots * \nu_1$.
(This is the definition of the convolution.
The ordering here is to follow the convention that linear operators act on the left of a vector.)
The moment operator follows the same rule:
\begin{align}
    M(\nu_1 * \nu_2, \rho) = \avg_{U\sim \nu_1, ~V \sim \nu_2} \rho(UV) 
    &=  \avg_{U\sim \nu_1} \avg_{V \sim \nu_2} \rho(U)\rho(V)\\
    &= \avg_{U\sim \nu_1}  \rho(U) \avg_{V \sim \nu_2} \rho(V)
    = M(\nu_1,\rho) M(\nu_2,\rho)\nonumber
\end{align}
where the second equality is because $\rho$ is a group representation.
Since the Haar measure is left and right invariant,
we have $\nu * \mu_G = \mu_G = \mu_G * \nu$ for any probability measure~$\nu$ on~$G$.
This implies that $M(\nu,\rho) M(\mu_G,\rho) = M(\mu_G,\rho) = M(\mu_G,\rho) M(\nu,\rho)$.
Hence we have for any integer~$k \ge 1$ 
\begin{align}
    \big(M(\nu,\rho) - M(\mu_G,\rho) \big)^k &= M(\nu^{* k},\rho) - M(\mu_G,\rho) , \nonumber\\
    g(\nu,\rho)^k &\ge g(\nu^{* k}, \rho)\, , \label{eq:gapamp}
\end{align}
where in the second line the equality holds if $M(\nu,\rho)$ is diagonalizable.
This amplification in~\cref{eq:gapamp} of the spectral gap by convolutions will be frequently used below.
\end{remark}

\begin{remark} \label{rem:convex_gap}
Suppose that $\nu, \nu'$ are distributions on $G$ and $\lambda \in [0,1]$.
If $g(\nu, t, G) \leq 1 - \delta$, then the convex combination $\lambda \nu + (1 - \lambda) \nu'$ has essential norm 
\begin{align*}
    g(\lambda \nu + (1 - \lambda) \nu', t, G) 
    &\leq \lambda g(\nu, t, G) + (1 - \lambda) g(\nu', t, G) \\
    &\leq 1 - \lambda \delta \,,
\end{align*}
where we used $g(\nu', t, G) \le 1$ for the last inequality. 
\end{remark}

\subsection{Multiplicative vs additive error approximate \texorpdfstring{$t$}{t}-designs}\label{section:relativevsadditive}
There are various definitions of approximate unitary and permutation designs in the literature.
Here, we will consider the strongest, often called multiplicative (or relative) error.

\begin{definition} [Unitary design] \label{def:unitary_designs}
Let $\nu$ be a probability distribution on $\SU(N)$. 
We call $\nu$ an approximate unitary $t$-design with multiplicative error $\varepsilon$ if
\begin{equation*}
    (1-\varepsilon) \, \Phi^{(t)}_H \, \preceq \, \Phi^{(t)}_{\nu} \, \preceq \, (1+\varepsilon) \, \Phi^{(t)}_H,
\end{equation*}
for quantum channels (completely positive trace-preserving, CPTP, maps)
\begin{equation*}
    \Phi^{(t)}_{\nu}(A) = \avg_{U \sim \nu} \left[ U^{\otimes t} A (U^\dagger)^{\otimes t} \right] \,, \qquad \Phi^{(t)}_{H}(A) = \avg_{U \sim \mu(\SU(N))} \left[ U^{\otimes t} A (U^\dagger)^{\otimes t} \right] .
\end{equation*}
Here, $\Phi \preceq \Phi'$ denotes that $\Phi'-\Phi$ is completely positive.
\end{definition}
The following lemma shows that the above notion of approximate unitary $t$-design can be achieved via gap amplification. While there exist many means to achieve multiplicative-error designs, the spectral gap approach seems to give the sharpest $\eps$-dependence.
\begin{lemma}[Unitary designs from spectral gaps]\label{lemma:fromgtounitarydesign}
    If $g(\nu, \ t, \ \SU(2^n)) \leq 1-\Delta$, 
    then $\nu^{*k}$ is an approximate unitary $t$-design with multiplicative error $\varepsilon$ when $k\geq c\cdot \frac{1}{\Delta}\left(n t +\log\frac{1}{\eps}\right)$ for some constant $c>0$.
\end{lemma}
\begin{proof}
    Gap amplification implies that $g(\nu^{*k}, \ t, \ \SU(2^n) ) \leq \varepsilon\cdot 2^{-2nt}$ when $k\geq\frac{1}{\Delta}\left(2\log 2\cdot nt + \log\frac{1}{\eps}\right)$. This implies that $\nu^{*k}$ is an approximate unitary $t$-design with multiplicative error $\varepsilon$ due to~\cite[Lemma 4]{brandao2016local}.
\end{proof}

\noindent The above definition of approximate unitary design immediately implies the weaker notion of  \emph{additive error} $t$-designs
\begin{equation*}
    \norm*{\Phi^{(t)}_{\nu} - \Phi^{(t)}_H}_\diamond \leq \varepsilon.
\end{equation*}
Negligible additive errors imply the operational indistinguishability 
of the channel $\Phi^{(t)}_{\nu}$ 
from the Haar random channel $\Phi^{(t)}_H$.
However, for many applications including adaptive queries to $t$-copies of the design (see {\it e.g.},~\cite{kretschmer2021quantum}),
additive error is insufficient and a bound on multiplicative error is needed.

For our purpose of proving linear growth of quantum circuit complexity, the approximate unitary design must faithfully reproduce the moments
\begin{equation}
    \avg_{U\sim \mu(\SU(2^n))} \abs{\bra{\phi}U\ket{0^n}}^{2t} = \binom{2^n+t-1}{t}^{-1}
\end{equation}
for any state $\ket{\phi}$ (see~\cref{section:circuitcomplexity}). This can be achieved by a constant multiplicative error, but requires additive error $\varepsilon\leq 2^{-nt}$ for the same purpose.

Next, we define approximate permutation designs in a similar fashion.
\begin{definition} [Permutation design] \label{def:permutation_designs}
Let $\nu$ be a probability distribution on $\Sym(2^n)$. 
We call $\nu$ an approximate permutation $t$-design with multiplicative error $\varepsilon$ if for any distinct $t$-tuple $x_1,\dots,x_t\in\{0,1\}^n$, and distinct $t$-tuple $y_1,\dots,y_t\in\{0,1\}^n$,
\begin{equation}
    \frac{1-\eps}{N(N-1)\cdots (N-t+1)}\leq\Pr_{\sigma\sim\nu}[\sigma(x_1)=y_1,\dots,\sigma(x_t)=y_t]\leq\frac{1+\eps}{N(N-1)\cdots (N-t+1)}.
\end{equation}
\end{definition}
This is closely related to the notion of ``approximate $t$-wise independent permutations'' in the literature ({\it e.g.},~\cite[Definition 7]{he2024pseudorandom}), where the distribution of $(\sigma(x_1),\dots,\sigma(x_t))$ is $\eps$-close to uniform (on the set of distinct $t$-tuples) in total variation distance. The distinction here is that we use the stronger notion of multiplicative error instead of additive error, which is useful for the application to circuit complexity (see~\cref{section:circuitcomplexity}). 

Finally, we state a lemma similar to \cref{lemma:fromgtounitarydesign} to connect the spectral gap with \cref{def:permutation_designs}. 

\begin{lemma}[Permutation designs from spectral gaps]\label{lemma:fromgtopermutationdesign}
    If $g(\nu, \ t, \ \Sym(2^n)) \leq 1-\Delta$, 
    then $\nu^{*k}$ is an approximate permutation $t$-design with multiplicative error $\varepsilon$ when $k\geq c\cdot \frac{1}{\Delta}\left(n t +\log\frac{1}{\eps}\right)$ for some constant $c>0$.
\end{lemma}
\begin{proof}
    We can choose sufficiently large $k$ such that the total variation distance (mentioned above) is at most $\eps/N^t$, which implies \cref{def:permutation_designs}. The stated bound follows from the connection between the spectral gap and the mixing time of a Markov chain.
\end{proof}

\subsection{Trace norm to spectral norm}\label{section:tracetospectral}

While the additive error is a weaker metric for approximate unitary designs,
it is sometimes easier to establish an estimate for.
To bridge the three notions, the essential norm (the spectral gap), the additive error, and the multiplicative error,
we present a conversion method.

Recall that for a linear operator~$A \in L(\cH)$ on a finite dimensional Hilbert space~$\cH$,
the Schatten $p$-norm of~$A$ where $p \in [1,\infty]$ is defined as $\norm{A}_p = (\Tr[ (A^\dagger A)^{p/2}] )^{1/p}$.
If $p = \infty$, the norm is the operator norm (also called the spectral norm),
the largest singular value of~$A$. 
If $p=1$, the norm is the trace norm.
For a superoperator $\Phi: L(\cH) \to L(\cH')$, we may consider the $p$-to-$q$ norm:
\begin{align}
    \norm{\Phi}_{p \to q} \deq \sup_{A \neq 0} \frac{\norm{\Phi(A)}_q}{\norm{A}_p} \,.
\end{align}

\noindent We need the following lemma (see~\cite{lofstrom1976interpolation} or \cite{reed2003methods}):
\begin{lemma}[Noncommutative Riesz--Thorin]\label{thm:rieszthorin}
    For any superoperator $\Phi: L(\cH) \to L(\cH')$ and real numbers $c, p \in [1,\infty]$
    we have
    \begin{equation}
        \norm{\Phi}_{p\to p} \leq \norm{\Phi}^{\frac{1}{c}}_{1\to 1}\norm{\Phi}^{1-\frac{1}{c}}_{\infty\to\infty} \, .
    \end{equation}
\end{lemma}

We will use the following fact for the Schatten $1$- and $\infty$-norms on the vector space of operators.
The inner product on the space of operators will always be the Hilbert--Schmidt inner product.

\begin{fact}[Duality]\label{fact:dualinducednorm}
    Let $V$ be a Hilbert space, 
    and let~$\norm{\cdot}_p$ be a norm on~$V$
    (that may be different from the norm given by the inner product~$\braket{\cdot | \cdot}$).
    Let $\norm{\cdot}_q$ be its dual norm 
    : $\norm{v}_q = \sup_{w \in V: \norm{w}_p=1} \abs{\braket{w|v}}$ for any~$v \in V$.
    Then, for the induced norms of a linear map $\Phi : V \to V$ it holds that
    \begin{equation}
        \norm{\Phi}_{p \to p} = \norm{ \Phi^\dagger }_{q \to q}\, .
    \end{equation}
\end{fact}

\begin{proof}
    We recall that the double dual norm is itself as follows.
    Let $x \mapsto \norm{x}_r = \sup_{y: \norm{y}_q = 1} \abs{\braket{y|x}}$ 
    be the dual norm of $\norm{\bullet}_q$; this is the double dual of~$\norm{\bullet}_p$.
    By definition, we have $\abs{\braket{y|x}} \le \norm{y}_q \norm{x}_p$ for any $x,y \in V$.
    For any real $\delta > 0$, there is $y \in V$ such that $\norm{x}_r - \delta < \abs{\braket{y|x}}$ and $\norm{y}_q = 1$.
    So $\norm{x}_r \le \norm{x}_p$.
    For $x \neq 0$, 
    let $\tilde x \in V$ be a scalar multiple of~$x$ such that $\braket{\tilde x| x} = 1$.
    \footnote{For more general normed vector spaces, the Hahn--Banach theorem can be used in this step
    to show that the canonical injection from $V$ to its double dual is an isometry.}
    Considering the one-dimensional subspace spanned by~$x \in V$ and the orthogonal projection onto it,
    we see $\norm{\tilde x}_q = \norm{x}_p^{-1}$.
    Then, $\norm{x}_r = \sup_{y \neq 0} \frac{\abs{\braket{y|x}}}{\norm{y}_q} \ge \frac{\abs{\braket{\tilde x|x}}}{\norm{\tilde x}_q} = \norm{x}_p$.
    Hence, $\norm{x}_r = \norm{x}_p$ for all~$x \in V$. 
    Then,
    \begin{equation}
    \sup_{v:~\norm{v}_q = 1} \norm{\Phi^\dagger v}_q
    =
    \sup_{v,w:~\norm{v}_q = 1, \norm{w}_p =1} \abs{\braket{w | \Phi^\dagger v}}
    =
    \sup_{v,w:~\norm{v}_q = 1, \norm{w}_p =1} \abs{\braket{v | \Phi w }}
    =
    \sup_{w:~\norm{w}_p = 1} \norm{\Phi w}_p \, .
    \qedhere
    \end{equation}
\end{proof}

\begin{corollary}\label{cor:AdditiveToMultiplicativeError}
    Let $\Phi, \Phi' : L(\cH) \to L(\cH)$ be superoperators for a finite dimensional Hilbert space~$\cH$.
    If  $(\Phi' \circ \Phi) = (\Phi' \circ \Phi)^\dagger$ 
    and $\norm{\Phi'}_{1 \to 1} \le 1$,
    then
    \begin{align}
        \norm{\Phi' \circ \Phi}_{2 \to 2} \le \norm{\Phi}_{1 \to 1} .
    \end{align}
    In particular, for any probability distributions~$\nu, \nu'$ on a finite or compact Lie group~$G$ 
    and any integer~$t \ge 1$,
    if for any measurable set~$\{W\} \subseteq G$ 
    it holds that $(\nu' * \nu)(\{W\}) = (\nu' * \nu)(\{W^{-1}\})$, {\it i.e.}, the same probabilities,
    then
    \begin{align}
        &\norm*{
            \avg_{U \sim \nu' * \nu} (U \otimes \overline U)^{\otimes t} - \avg_{V \sim \mu_G} (V\otimes \overline V)^{\otimes t}
        }_{\infty}
        \le
        \norm*{\Phi}_{1 \to 1} \\
        \text{ where } \quad &\Phi(\eta) = \avg_{U \sim \nu} U^{\ot t} \eta U^{\dagger \otimes t} - \avg_{V \sim \mu_G} V^{\ot t} \eta V^{\dagger \otimes t} \nonumber
    \end{align}
\end{corollary}

We use this in the following way.
For a random unitary that is an approximate unitary design,
we consider another design obtained by appending the inverse of the random unitary.
If we have an estimate for the error of the former design
measured in $1 \to 1$ norm,
we will have the same estimate for the latter design
measured in $2 \to 2$ norm, that is the essential norm.
This is an important technical ingredient toward our results.

\begin{proof}
    First, note that $\norm{\Phi' \circ \Phi}_{1 \to 1} \le \norm{\Phi'}_{1 \to 1} \norm{\Phi}_{1 \to 1}$,
    which is at most $\norm{\Phi}_{1 \to 1}$ by assumption.
    Hence, it suffices to prove the claim with $\Phi' = \one$ (the identity map),
    which means~$\Phi = \Phi^\dagger$.
    Now, \cref{thm:rieszthorin} says that 
    $\norm{\Phi}_{2 \to 2} \le \norm{\Phi}_{1 \to 1}^{1/2} \norm{\Phi}_{\infty \to \infty}^{1/2}$.
    \Cref{fact:dualinducednorm} says that $\norm{\Phi}_{\infty \to \infty} = \norm{\Phi^\dagger}_{1 \to 1}$,
    but we have assumed~$\Phi^\dagger = \Phi$.
    This proves the first claim.

    The second claim relies on the left invariance of the Haar measure, 
    which implies that
    \begin{align}
    (\avg_{U \sim \nu'} \tau_{t,t}(U) ) (\avg_{V \sim \mu_G} \tau_{t,t}(V)) 
    = (\avg_{V \sim \mu_G} \tau_{t,t}(V))
    \end{align}
    where $\tau_{t,t}$ is the tensor representation.
    Define superoperators for~$(\CC^N)^{\otimes t}$
    \begin{align}
        \Phi &: \eta \mapsto \avg_{U \sim \nu} U^{\otimes t} \eta U^{\dagger \otimes t} - \avg_{V \sim \mu_G} V^{\otimes t} \eta V^{\dagger \otimes t},\\
        \Phi' &: \eta \mapsto \avg_{U \sim \nu'} U^{\otimes t} \eta U^{\dagger \otimes t}.\nonumber
    \end{align}
    The assumption on~$\nu' * \nu$ 
    implies that $(\Phi' \circ \Phi)^\dagger = \Phi' \circ \Phi$.
    It is obvious that $\norm{\Phi'}_{1 \to 1} \le 1$.
    The Schatten $2$-norm is the inner product norm on the space of operators,
    so the $2 \to 2$ norm of the superoperator 
    is the spectral norm of the linear operator acting on the vector space of operators.
    Therefore, the second claim follows from the first part.
\end{proof}

\subsection{Kazhdan constants}\label{sec:kaszhdan}

Kazhdan constants are in some sense similar to spectral gaps (see~\cref{lem:gapviakazhdan}) 
but have a nice geometric interpretation.

\begin{definition}[Kazhdan constants] \label{def:kazhdan_const}
    Consider a group $G$ generated by a set $S$.
    The \emph{Kazhdan constant for $G$ with respect to $S$}, denoted by $\cK(G;S)$, 
    is the largest $\eps \ge 0$ that satisfies the following property: 
    for every unitary representation $\rho: G \to \U(\cH)$
    and every unit vector $\ket{\psi} \in \cH$, 
    where $\cH$ is a Hilbert space without $G$-invariant vectors 
    ({\it i.e.}, $\rho$ does not have a trivial subrepresentation), 
    there exists $g \in S$ such that $\norm{\rho(g) \ket{\psi} - \ket{\psi}} \geq \eps$. 
\end{definition}

\noindent
For compact groups at least, 
one can think about the Kazhdan constant for each nontrivial irreducible unitary representation
and take the infimum of all those.

The Kazhdan constant is useful because we can compare it for different generating sets
using the following basic properties. We include self-contained proofs, partly following~\cite{Kassabov_2007_alt}.

\begin{lemma}[{\cite[Proposition 1.3]{Kassabov_2007_alt}}]\label{lem:KGG}
    For any group~$G$ that admits a Haar probability measure,
    \begin{align}
        \cK\left(G;G\right) \ge \sqrt{2}.
    \end{align}
\end{lemma}

\begin{proof}
    Otherwise, the orbit of a normalized vector~$\ket \psi$ under the $G$ action would be in some half space,
    so the convex hull of the orbit would not contain the zero vector,
    and $\avg_{g \sim \mu_G} \rho(g) \ket \psi$ would be an invariant nonzero vector.
\end{proof}

\begin{lemma}[Special case of {~\cite[Proposition 1.6 b]{Kassabov_2007_alt}}]\label{lem:chain_rule_pii}
    Let $H$ be a group generated by a set~$S \subseteq H$.
    If $\phi_i : H \to G$ are group homomorphisms, we have
    \begin{align}
        \cK\left(G;\bigcup_i \phi_i(S)\right) \ge \frac{1}{2} \cK\left(G;\bigcup_i \phi_i(H) \right) \cdot \cK(H;S).
    \end{align}
\end{lemma}

\begin{proof}
    Let $\rho : G \to \U(\cH)$ be a unitary representation with no $G$-invariant vectors.
    Suppose $\ket \psi \in \cH$ is an $\eps$-invariant normalized vector with respect to~$\bigcup_i \phi_i(S)$,
    {\it i.e.}, $\norm{\rho(s) \ket \psi - \ket \psi} < \eps$ for all $s \in \bigcup_i \phi_i(S)$.
    By the definition of~$\cK(G; \bigcup \phi_i(H))$,
    there exists $h_0 \in \bigcup_i \phi_i(H)$ such that
    \begin{align}
        \cK\left(G; \bigcup_i \phi_i(H)\right)
        \le
        \norm{\rho(h_0) \ket \psi - \ket \psi}. \label{eq:kgh}
    \end{align}
    Let $j$ be such that $h_0 \in \phi_j(H)$.
    Let $\Pi_j$ be the orthogonal projector onto the subspace of all $\phi_j(H)$-invariant vectors,
    so $\rho(h)\Pi_j = \Pi_j = \Pi_j \rho(h)$ for all $h \in \phi_j(H)$.
    Then, $H \xrightarrow{~\phi_j~} G \xrightarrow{~\rho~} \U(\cH)$ 
    contains an $H$-representation on~$(\one - \Pi_j)\cH$ with no $H$-invariant vectors.
    By the definition of~$\cK(H;S)$ there exists $s_0 \in \phi_j(S) \subseteq \phi_j(H)$
    such that
    \begin{align}
        \cK(H;S) \norm{(\one - \Pi_j)\ket \psi}
        \le
        \norm{\rho(s_0)(\one - \Pi_j) \ket \psi - (\one - \Pi_j) \ket \psi}
        =
        \norm{\rho(s_0) \ket \psi - \ket \psi} 
        <
        \eps \, . \label{eq:khseps}
    \end{align}
   From~\cref{eq:kgh} and~\cref{eq:khseps} we have
    \begin{align}
        \cK\left(G; \bigcup_i \phi_i(H)\right)
        \le
        \norm{\rho(h_0) \ket \psi - \ket \psi}
        \le
        2 \norm{(\one-\Pi_j) \ket \psi}
        <
        2 \eps \cK(H;S)^{-1}.
    \end{align}
    Therefore, $\eps \ge \frac 1 2 \cK(H;S)\cK(G;\bigcup_i \phi_i(H))$.
\end{proof}

\begin{lemma}[Short product{~\cite[Proposition 1.4]{Kassabov_2007_alt}}]\label{lem:shortproduct}
    Let $S$ and $S'$ be two generating sets of a group $G$ such that $S'\subseteq \bigcup_{j = 0}^k S^j$, 
    {\it i.e.}, every element of~$S'$ can be written as a product of at most $k$ elements of~$S$. 
    Then,
    \begin{align}
        \cK(G;S) \ge \frac{1}{k} \cK(G;S').
    \end{align}
\end{lemma}

Remarkably, the Kazhdan constant may only degrade by a factor of~$k$,
instead of the cardinality of $S'$, which can be exponential in $k$. Roughly, this is because a good Kazhdan constant (\cref{def:kazhdan_const}) only requires the existence of \textit{one} $g\in S$ to move the vector, whereas spectral gap requires \textit{most} elements to move the vector. 
We will use this short product lemma many times later.

\begin{proof}[Proof of~\cref{lem:shortproduct}]
    Let $\rho:G\to \U(\cH)$ be a unitary representation with no $G$-invariant vectors.
    Suppose, on the contrary to the claim, 
    that there is a normalized vector $\ket{\psi}$ 
    such that $\norm*{\rho(g)\ket{\psi} -\ket{\psi}} < \cK(G,S') / k$ for all $g\in S$.
    By definition, there is an element~$h \in S'$, 
    which must be a product $h = g_1 \cdots g_m \in S'$ of $g_j \in S$ where $m \le k$,
    such that $\cK(G; S') \le \norm*{ \rho(g_1\cdots g_m) \ket{\psi} - \ket{\psi} }$.
    But this is a contradiction:
    \begin{align}
		\cK(G; S') 
        \le \norm*{ \rho(g_1\cdots g_m) \ket{\psi} - \ket{\psi} } 
        &\leq \sum_{j=1}^{m} \norm*{\rho(g_1)\cdots\rho(g_{j}) \ket \psi - \rho(g_1)\cdots\rho(g_{j-1})\ket{\psi} }\nonumber\\
		&= \sum_{j=1}^{m} \norm*{ \rho(g_j)\ket{\psi} - \ket{\psi} }\\
		&< m \cdot \cK(G;S')/k \nonumber\\
            & \leq \cK(G;S'). \nonumber \qedhere
    \end{align}
\end{proof}

\begin{corollary}\label{cor:GoodKazhdanSubgroupToBigger}
    Let a group~$H$ be generated by a subset~$S$, 
    and $\phi_i : H \to G$ be group homomorphisms.
    If every element of~$G$ can be written as a product of at most~$k$ elements of~$\bigcup_i \phi_i(H)$,
    then
    \begin{align}
        \cK(G; \bigcup_i \phi_i(S) ) \ge \frac 1 {k\sqrt{2}} \cK(H;S) .
    \end{align}
\end{corollary}
\begin{proof}
    Combine \cref{lem:KGG}, \cref{lem:chain_rule_pii}, and \cref{lem:shortproduct}.
\end{proof}

A Kazhdan constant gives a lower bound on the spectral gap:

\begin{lemma}[Spectral gap from Kazhdan constant]\label{lem:gapviakazhdan}
    Let $G$ be a compact group (finite or Lie) generated by a finite subset~$S$ 
    where $S$ contains the identity
    and is closed under taking inverses.
    Let $\nu_S$ be a probability measure on~$G$ that is uniform over~$S$, 
    {\it i.e.}, $\nu_S$ assigns an equal probability for every element of~$S$.
    Then for any finite-dimensional unitary representation $\rho$ of $G$,
   	\begin{equation}\label{eq:gapviakazhdan}
   		g(\nu_S, \rho, G) \leq 1 - \frac{\cK(G;S)^2}{2 \abs S}.
   	\end{equation}
\end{lemma}

Note that $\cK(G; S) = \cK(G; S\cup \{\one\})$ for any $S$, but $g(\nu_S, \rho, G) \leq g(\nu_{S\cup \{\one\}}, \rho, G)$,
so it is not essential that the generating set $S$ contains the identity.

\begin{proof}[Proof of~\cref{lem:gapviakazhdan}]
    Since $G$ is compact,
    it suffices to prove the lemma for a nontrivial and irreducible~$\rho$ and nonzero $\cK(G;S)$.
    This forces $\abs S \ge 2$.
    It follows from~\cref{def:moment_gap} and~\cref{rmk:moment_projector} that
    $g(\nu_S, \rho, G) = \norm*{ \avg_{U \sim \nu_S} \rho(U) }_\infty$.
    Since $S$ is closed under taking inverses, the moment operator is hermitian.
    Hence, the essential norm is
    \begin{align}
         g(\nu_S,\rho, G) = \frac{1}{\abs S} \max_{\ket \psi} \abs*{\sum_{s \in S} \bra \psi \rho(s) \ket \psi}
    \end{align}
    where $\ket{\psi}$ is a unit vector in the representation space of~$\rho$.
    By the definition of Kazhdan constant $\cK(G; S)$, 
    for any such $\ket{\psi}$, 
    there exists $s_\psi \in S$ such that
    \begin{align}
        \cK(G;S)^2 
        \leq \norm*{\rho(s_\psi)\ket{\psi} - \ket{\psi}}^2 
        = 2 - \bra{\psi} \rho(s_\psi) \ket{\psi} - \bra{\psi} \rho(s_\psi^{-1}) \ket{\psi}
    \end{align}
    where $s_\psi^{-1}\in S$ by assumption. 
    If $s_\psi \neq s_\psi^{-1}$, then
    \begin{align}
        \sum_{s \in S} \bra{\psi} \rho(s) \ket{\psi} 
        &\le 
        \bra \psi \rho(s_\psi) \ket \psi + \bra \psi \rho(s_\psi^{-1}) \ket \psi + 
        \sum_{s \neq s_\psi,s_\psi^{-1} } \bra \psi \rho(s) \ket \psi \nonumber\\
        &\le 
        \bra \psi \rho(s_\psi) \ket \psi + \bra \psi \rho(s_\psi^{-1}) \ket \psi
        + (\abs S -2 )   \\
        &\le
        |S| - \cK(G;S)^2\, . \nonumber
    \end{align}
    Similarly, if $s_\psi = s_\psi^{-1}$, then
    \begin{align}
        \sum_{s \in S} \bra{\psi} \rho(s) \ket{\psi} 
        &\le 
        \bra \psi \rho(s_\psi) \ket \psi + 
        \sum_{s \neq s_\psi } \bra \psi \rho(s) \ket \psi \nonumber\\
        &\le 
        \frac12 \left(2 - \cK(G;S)^2\right)
        + (\abs S - 1 )   \\
        &=
        |S| - \frac12\cK(G;S)^2\, . \nonumber
    \end{align}
    On the other hand, since $S$ contains the identity and $\cK(G;S)^2 \le 4$,
    \begin{equation}
        \sum_{s\in S} \bra{\psi}  \rho(s) \ket{\psi} 
        \geq  1 + (-1) \cdot (\abs S - 1) = - \abs S + 2 \geq -\abs S + \frac 1 2 \cK(G; S)^2. \qedhere
    \end{equation}
\end{proof}

\begin{remark}
    The quadratic dependence of the spectral gap bound by the Kazhdan constant cannot be improved.
    Consider~$G = \ZZ/m\ZZ$, the additive group of integers modulo $m$,
    and $S = \{-1,0,1\}$.
    This is an abelian group and every irrep is of form $G \ni j \mapsto e^{2\pi i p j / m} \in \U(1)$
    for some integer~$p$.
    An irrep does not have any $G$-invariant nonzero vector if and only if $p \neq 0 \bmod m$.
    The Kazhdan contant is $\min_{p \neq 0 \bmod m} \norm{e^{2\pi i p/m} - 1} = \abs{e^{2\pi i / m} - 1}$.
    For large $m$, this is $\approx 2\pi / m$.
    On the other hand, the essential norm maximized over all irreps 
    is~$(1+2\cos \frac{2\pi}{m} )/3 \approx 1 - 4\pi^2 / 3m^2$.
\end{remark}

\subsection{Gaps of convolutions and convex combinations of subgroup distributions} \label{sec:hamiltonian_tools}

We use the detectability lemma~\cite{aharonov2009detectability} and its converse, Gao's quantum union bound~\cite{gao2015quantum,o2022quantum}.
Here, we follow the presentation in~\cite{anshu2016simple}:

\begin{lemma}[Detectability lemma and quantum union bound~\cite{anshu2016simple}]\label{lemma:detectabilityandunionbound}
    Let $H=\sum_i Q_i$ be a finite sum of orthogonal projectors $Q_i$ on a Hilbert space~$\cH$.
    Then, we have for every vector $\ket \psi \in \cH$
    \begin{equation}
        1-4\braket{ \psi | H | \psi}\leq \norm*{\prod_i (\one-Q_i)\ket{\psi}}^2\,. \label{eqn:quantumunionbound}
    \end{equation}
    If furthermore each projector $Q_i$ commutes with all but $\ell$ other projectors, where $\ell \ge 1$,
    then we have an upper bound
    \begin{align}
        \norm*{\prod_i (\one-Q_i)\ket{\psi}}^2 \leq \frac{1}{ \ell^{-2}\braket{ \psi | H | \psi}  + 1}\,. \label{eqn:detectabilitylemma}
    \end{align}
    Here, the ordering of projectors in the product~$\prod_i (\one - Q_i)$ can be arbitrary.
\end{lemma}

\noindent These bounds originate from the study of gapped frustration-free Hamiltonians~\cite{aharonov2011detectability}.
A reader would enjoy a half-page proof of \cref{eqn:detectabilitylemma} in \cite{anshu2016simple} as well as a half-page proof of \cref{eqn:quantumunionbound} in \cite{o2022quantum}.

In the context of random quantum circuits, 
\cref{lemma:detectabilityandunionbound} helps to relate spectral gaps of two types of random quantum circuits.
Here is our adaptation where averaging over subgroups naturally gives projectors:
\begin{lemma}[A detectability lemma and its converse for subgroups] \label{lem:local-vs-parallel}
    Let $G_1,\ldots,G_L$ be compact subgroups of a compact group~$A$,
    where each $G_i$ commutes element-wise with all but $\ell$ others.
    Suppose $\ell \ge 1$.
    Consider the convolution and average of the subgroup Haar measures:
    \begin{align}
        \Asterisk \deq \mu_{G_1} * \cdots *\mu_{G_L}\qquad \text{ and } \qquad   
        \Sigma \deq \frac{1}{L}\big ( \mu_{G_1} + \cdots + \mu_{G_L} \big). 
    \end{align}
    Then, for any unitary representation~$\rho$ of~$A$,
    the spectral gaps $\Delta(\Asterisk) = 1- g(\Asterisk, \rho, A)$ and $\Delta(\Sigma) = 1 -  g(\Sigma, \rho, A)$
    are related as
    \begin{align} \label{eq:local-vs-parallel}
        \Delta(\Sigma) \ge \frac{1}{4L}\Delta(\Asterisk) \quad\text{ and }\quad 
        \Delta(\Asterisk) 
        \ge 
        1 - \frac{1}{\sqrt{1 + L \ell^{-2} \Delta(\Sigma)}}
        \ge \frac 1 4 \min(1,~ L\ell^{-2} \Delta(\Sigma)) \, .
    \end{align}
\end{lemma}

For our purposes it suffices to assume that $A$ is a subgroup of~$\SU(N)$.
The coefficient~$1/4$ is not sharp.

\begin{proof}
As remarked earlier, the Haar average of represented operators of a compact group is an orthogonal projector
onto the trivial subrepresentation.
Since we are only interested in the essential norms, 
we may assume that $\rho$ does not contain any trivial subrepresentation of~$A$,
{\it i.e.}, there is no nonzero $A$-invariant vector.
There may still be $G_i$-invariant vectors.
Then, the essential norm is just the norm of the moment operator,
and the spectral gap is one minus this norm.

The moment operator for~$\Sigma$ is $M(\Sigma) = \frac 1 L \sum_{i=1}^L P_i$ 
where $P_i = \avg_{g \sim \mu(G_i)} \rho(g)$ is an orthogonal projector (onto the trivial $G_i$-subrepresentation).
Put $Q_i = \one - P_i$ to conform with the notation in~\cref{lemma:detectabilityandunionbound}.
We write the moment operator as
\begin{align}
    M(\Sigma) = \one - \frac 1 L H \quad \text{ where } \quad H = \sum_i Q_i.
\end{align}
Likewise, the moment operator for~$\Asterisk$ is $M(\Asterisk) = M(\mu_{G_1}) \cdots M(\mu_{G_L})$ by \cref{rem:conv},
so 
\begin{align}
    M(\Asterisk) = \prod_{i=1}^L (\one - Q_i).
\end{align}
The assumption on the commutativity among $G_i$ implies that each $Q_i$ commutes with all but $\ell$ others.
The spectral gaps are
\begin{align}
    \Delta(\Sigma) = \frac 1 L \inf_{\ket \psi} \bra \psi H \ket \psi
    \quad\text{ and }\quad
    \Delta(\Asterisk) = \inf_{\ket \psi} \parens*{1 -  \norm*{\prod_{i} (\one - Q_i) \ket \psi} }.
\end{align}

Write $y = \norm{\prod_i (\one - Q_i) \ket \psi}$ and $x = \frac 1 L \bra \psi H \ket \psi$ for short.
Since $y^2 \le y$,
\cref{eqn:quantumunionbound} implies that $1 - 4 x L \le y$,
which means $1 - y \le 4 L x$.
This implies that $\Delta(\Asterisk) \le 4L \Delta(\Sigma)$.
On the other hand,
\cref{eqn:detectabilitylemma} says that $y \le (1 + \ell^{-2} L x )^{-1/2}$.
Rearranging, we have $1 - y \ge 1 - (1+ \ell^{-2} L x )^{-1/2}$.
The final (easy) inequality follows since
a real function $z \mapsto 1 - (1+z)^{-1/2}$ for $z \ge 0$ 
is increasing and concave.
\end{proof}

The following shows that by enlarging groups from which we take random instances, 
we only increase the spectral gap. This uses the subgroup structure and fails for subsets.

\begin{lemma}[Gap from subgroup]\label{fact:gfromoperatorinequality}
    Let $G_i \le G'_i \le H$ be compact groups,
    so each has its own Haar probability measure~$\mu(G_i)$, $\mu(G'_i)$, or $\mu(H)$.
    Then, 
    \begin{align}
        g(\avg_i \mu(G'_i), \ t, \ H) \leq g(\avg_i \mu(G_i), \ t, \ H).
    \end{align}
\end{lemma}

\begin{proof}
    For any compact group~$G$ and any unitary representation~$\rho$ of~$G$,
    the operator $\avg_{h \sim \mu(G)} \rho(h)$ is the orthogonal projector onto its trivial subrepresentation. 
    A trivial action from a group implies a trivial action from its subgroup. Hence,
    \begin{align*}
        0 \preceq \avg_{h \sim \mu(G'_i)} \rho(h) &\preceq  \avg_{h \sim \mu(G_i)} \rho(h)\,,\\
       \implies \quad 0 \preceq M(\avg_i \mu(G'_i), \rho) = \avg_i \avg_{h \sim \mu(G'_i)} \rho(h)
        &\preceq 
        \avg_i \avg_{h \sim \mu(G_i)} \rho(h) = M(\avg_i \mu(G_i), \rho)\, ,\\\implies \quad
        0 \preceq M(\avg_i \mu(G'_i),\rho,H) - M(\mu(H),\rho,H)
        &\preceq
        M(\avg_i \mu(G_i),\rho,H) - M(\mu(H),\rho,H) .
    \end{align*}
    The claim follows. 
    (If $0 \preceq A \preceq B$ for two hermitian operators $A$ and $B$,
    then for any real number $\delta > 0$, there exists a normalized vector $\ket \psi$ 
    such that $\norm{A} - \delta < \bra \psi A \ket \psi$, implying
    $\norm{A} - \delta < \bra \psi A \ket \psi \le \bra \psi B \ket \psi \le \norm{B}$.)
\end{proof}

\section{Spectral gaps for random  reversible circuits} \label{sec:gap_revckt}

We begin with a model of random reversible circuits with all-to-all connectivity:

\begin{definition}[Random reversible circuits]\label{def:revalltoall}
    Let $\Alt(2^n)$ be the alternating group that permutes elements of~$\{0,1\}^n$.
    For $n > 3$ we define a probability distribution  $\nurevalltoall$
    on $\Alt(2^n)$ obtained by 
    sampling three distinct indices $\{i_1,i_2,i_3\} \subset \{1,2,\ldots,n\}$ 
    uniformly at random 
    and then 
    applying a uniformly random permutation in $\Sym(2^3)$ on the bits $\{ i_1, i_2, i_3 \}$. 
\end{definition}

\noindent
Schematically, we may write
\begin{align}
    \nurevalltoall = \avg_{\{i_1, i_2, i_3\}} \mu(\Sym(\{i_1, i_2, i_3\})) \, .
\end{align}

We restate the goal of this section for readers' convenience.

\gaprev*

\noindent
To prove this, we use crucially Kassabov's generators~\cite{Kassabov_2007_alt} for alternating groups,
whose Kazhdan constants are uniformly bounded away from zero.
Specializing the generators on the set of all $n$-bit strings,
we will show in \cref{section:efficientreversiblecircuits} 
that Kassabov's generators can be implemented by $\cO(n)$ classical reversible gates.
\Cref{section:efficientreversiblecircuits} logically precedes this section,
but to read this section it suffices to note a result in~\cref{section:efficientreversiblecircuits},
which we copy here:

\begin{restatable*}[Kassabov's generators are short reversible circuits]{theorem}{kascircuitrev}
\label{thm:UltimatePermutationDesign}
    There exists a real $\eps >0$ and an integer $k \ge 1$ 
    such that for every $n \ge 1$
    there exists a generating set $S$ of $k$ elements for $\Alt(2^n)$ with respect to which 
    the Kazhdan constant is at least $\eps$.
    Moreover, every generator in $S$ is a strongly explicit reversible circuit on $n$ bits
    consisting of $\cO(n)$ NOT, control-NOT, and Toffoli gates without any ancilla bit,
    and every bit in this circuit is acted upon by $\cO(1)$ gates.
\end{restatable*}

Since this section focuses on classical circuits, there is no distinction between bras and kets.
Still, we find it convenient to work with tensor product vector spaces~$(\CC^2)^{\otimes \cdots}$
because then we may simply quote the results from the earlier sections.
We represent the permutation groups $\Alt(2^n)\le \Sym(2^n)$ in a standard way.
The representation space is a complex vector space of dimension~$2^n$ (or its tensor power),
and a permutation group element~$\pi$ permutes the basis vectors as
\begin{align}
    P(\pi) \ket z = \ket{\pi(z)} \qquad\text{where}\quad  z \in \{0,1\}^n, \quad \pi \in \Sym(2^n)\, .
\end{align}
We will use a shorthand in this section:
\begin{align}
    \tau : \pi \mapsto P(\pi)^{\otimes t} .
\end{align}

\subsection{A \texorpdfstring{$t$}{t}-independent spectral gap from Kassabov's generators}

Using tools noted in~\cref{sec:hamiltonian_tools} 
and Kazhdan constants, 
we turn this structured circuit of \cref{thm:UltimatePermutationDesign} 
into a random circuit where each step is a $3$-bit reversible gate chosen uniformly at random.

We first replace each NOT, CNOT, and Toffoli gate in the circuit of a Kassabov generator 
with a random element from a group $\Sym(2^3)$.

\begin{proposition}\label{lemma:Sgateisgapped}
    For each subset $T = \{i_1,i_2,i_3\}$ of three distinct elements from~$\{1,2,\ldots,n\}$,
    let $G_T \cong \Sym(2^3)$ be the subgroup of $\Alt(2^n)$
    that acts on the bits $i_1,i_2,i_3$.
    For any collection~$\cI$ of such triples, 
    let $\nu_\cI$ denote the convex combination $\nu_\cI = \abs{\cI}^{-1} \sum_{T \in \cI} \mu(G_T)$.
    For any integer $n\geq 4$,
    there exists a strongly explicit collection $\cI_\mathrm{Kas}$ of $\cO(n)$ triples of distinct bit indices
    such that for all integers $t \geq 1$,
    \begin{align}
        g\left(\nu_{\cI_\mathrm{Kas}},~ \tau,~ \Alt(2^n) \right) = 1 - \Omega(n^{-3}). 
    \end{align}
\end{proposition}

Note that because $n> 3$, each $G_T$ is an embedding of $\Sym(2^3)$ into $\Alt(2^n)$. 

\begin{proof}
    Let $S$ be the set of Kassabov's generators described in~\cref{thm:UltimatePermutationDesign}. 
    For each circuit $h \in S$ of order $2$, let $H_h = \{\one, h\}$ be the group generated by this gate.
    It follows that
    \begin{align}
        &\avg_{h\in S} \mu(H_h) = \frac12 \mu( \{\one\} ) + \frac12 \nu_S \, ,\\
        &g\big(
            \avg_{h\in S}\mu(H_h),~ \tau,~ \Alt(2^n)
        \big) = 1 - \Omega(1)\,  .\nonumber
    \end{align}
    Each circuit $h\in S$ is a product of $\cO(n)$ NOT, CNOT, and Toffoli gates.
    Let $S'$ denote the set of all these elementary $3$-bit gates that appear in at least one circuit of $S$.
    (CNOT and NOT are not $3$-bit gates, but we choose arbitrary one or two other bits 
    to insist that they are each a $3$-bit gate.)
    Let $\cI_\mathrm{Kas}$ be the collection of all triples of indices that support the gates of $S'$. 
    \cref{thm:UltimatePermutationDesign} implies that $|S'| = \cO(n)$, so $\abs{\cI_\mathrm{Kas}} = \cO(n)$.
    Using~\cref{lem:shortproduct,lem:gapviakazhdan}, we have
    \begin{align}
        g(\E_{h'\in S'}\mu(H_{h'}), \ \tau, \ \Alt(2^n)) 
        \leq 
        1 - \frac{\cK(\Alt(2^n); S')^2}{2|S'|} 
        \leq 
        1 - \frac{\cK(\Alt(2^n); S)^2}{2|S'|^3} = 1 - \Omega(n^{-3}).
    \end{align}
    Since for each $h'\in S'$, the group $H_{h'}$ is a subgroup of a copy of $\Sym(8)$ acting on the same three bits as $h'$, it follows from~\cref{fact:gfromoperatorinequality} that
    \begin{equation}
    g\left(\nu_{\cI_\mathrm{Kas}},~ \tau,~ \Alt(2^n) \right) \leq g(\E_{h'\in S'}\mu(H_{h'}), \ \tau, \ \Alt(2^n)) = 1 - \Omega(n^{-3}).  \qedhere
    \end{equation}
\end{proof}

\begin{proof}[Proof of~\cref{thm:main_rev}]
    By convexity of the operator norm, we have $g(\avg_i \nu_i, \rho, G) \le \avg_i g(\nu_i,\rho,G)$
    for any distribution~$\nu_i$ on a group~$G$ and any representation~$\rho$ of~$G$.
    Let $\Sym(n)$ permute bit labels (not $n$-bit strings).
    For each $\pi \in \Sym(n)$ we have an obvious distribution $\pi(\nu_{\cI_\mathrm{Kas}})$,
    obtained by permuting bit labels, with the same spectral gap.
    Since $\cI_\mathrm{Kas}$ is nonempty, 
    we have $\avg_{\pi \sim \mu(\Sym(n))} \pi(\nu_{\cI_\mathrm{Kas}}) = \nurevalltoall$.
    \begin{align}
        g(\nurevalltoall,~ \tau, ~\Alt(2^n)) 
        &=
        g(\avg_{\pi \sim \mu(\Sym(n))} \pi(\nu_{\cI_\mathrm{Kas}}), ~\tau,~\Alt(2^n)) \\
        &\le
        \avg_{\pi \sim \mu(\Sym(n))} g( \pi(\nu_{\cI_\mathrm{Kas}}), ~\tau,~\Alt(2^n))
        =
        \avg_{\pi \sim \mu(\Sym(n))} (1 - \Omega(n^{-3})) . \qquad \qquad \nonumber \qedhere
    \end{align}
\end{proof}

\subsection{Bootstrapping spectral gap for random reversible circuits: 
a proof of \texorpdfstring{\cref{thm:bootstrappedreversiblegap}}{Theorem 1.2}}\label{sec:reversiblebootstraping}

The goal is to replace the gap estimate $\Omega(n^{-3})$ of~$\nurevalltoall$ in \cref{thm:main_rev} 
to $\Omega(n^{-1}\log^{-3}(nt))$,
where the exponent on $\log(nt)$ is the exponent on~$n$ in \cref{thm:main_rev}.
This is based on the following relation between the spectral gap of~$\nurevalltoall$ on $n$ bits
and that on a smaller subsystem with $\Theta(\log(nt))$ bits.
To clarify the number of bits, we write $\nurevalltoall(n)$, not just $\nurevalltoall$,
to denote the random reversible circuit on~$n$ bits.

\begin{proposition}[Bootstrapping spectral gap for better $n$-dependence]\label{lemma:bootstrapingforpermutations}
    For any positive integer $m$, let $\Delta(m)$ be the spectral gap of the random reversible circuit on $m$ bits:
    \begin{align}
        \Delta(m) = 1 - g(\nurevalltoall(m),~\tau,~\Alt(2^m))
    \end{align}
    Then, for any $t \leq \Theta(2^{n/6.1})$,
     \begin{equation}
        \Delta(n)  \ge \frac{\Delta(\lceil 11\ln(nt) \rceil)}{n} (1 - \cO(n^{-1}))
     \end{equation}
\end{proposition}

\noindent
\cref{thm:bootstrappedreversiblegap} is an immediate consequence of this proposition and \cref{thm:main_rev}.

To prove \cref{lemma:bootstrapingforpermutations} we use an auxiliary distribution (a random walk)
\begin{align}
    \beta = \frac{1}{n} \sum_{i=1}^n \mu(\Alt(2^{n-1})_{[n]\setminus\{i\}})
\end{align}
where $\Alt(2^{n-1})_{[n]\setminus\{i\}}$ is the alternating group that acts on $n-1$ bits,
which are all bits but the $i$-th.
Let $Q_i = M(\mu(\Alt(2^{n-1})_{[n]\setminus\{i\}}),~ \tau)$ be the moment operator.
This is an orthogonal projector onto the subspace of all $\Alt(2^{n-1})$-invariant vectors
and is supported on all but the $i$-th qubit with $t$-th tensor power.
Denote the spectral gap of $\beta$ by
\begin{align}
    \delta(n) 
    := 1 - g( \beta,~ \tau,~ \Alt(2^n)) 
    = 1 - \norm*{\left(\frac{1}{n} \sum_{i=1}^n Q_i\right) - \avg_{\pi \sim \mu(\Alt(2^n))} P(\pi)^{\otimes t} }_\infty \, .
\end{align}

\begin{lemma}\label{lemma:recursion}
    $\Delta(n) \ge \delta(n) \Delta(n-1)$ for all $n \ge 4$.
\end{lemma}
The proof of this lemma is similar to~\cite{haferkamp2021improved}.
\begin{proof}
    For any~$T \subset \{1,2,\ldots,n\}$ with $\abs{T} = 3$,
    let $P_T = M(\mu(\Sym(2^3)_T),\tau)$ be the orthogonal projector onto the $\Sym(2^3)$-trivial subspace.
    $P_T$ is supported on $3t$ qubits that form the $t$-fold tensor power of those on~$T$.
    The moment operator $M(\nurevalltoall(n),~\tau)$ equals $\avg_T P_T$.
    If $i \notin T$, then $P_T Q_i = Q_i = Q_i P_T$.
    Let $\ket \psi$ be a normalized vector 
    orthogonal to the $\Alt(2^n)$-trivial subspace,
    {\it i.e.}, $M(\mu(\Alt(2^n)),\tau) \ket \psi = 0$.
    Then, for any qubit $i$, we have
    \begin{align}
        \begin{split}\label{eq:recursionfirstcalculation}
        \bra{\psi} \sum_{T:~ i \notin T}  P_T \ket{\psi}
        &= 
        \bra{\psi} \sum_{T:~i\notin T}  P_T Q_i \ket{\psi} 
            + 
        \bra{\psi} \sum_{T:~i\notin T}  P_T(\one - Q_i) \ket{\psi}\\
        &= 
        \binom{n-1}{3} \bra{\psi} Q_i \ket{\psi}
            +
        \bra{\psi} (\one - Q_i) \sum_{T:~i\notin T}  P_T (\one - Q_i) \ket{\psi}
        \\
        &\leq  
        \binom{n-1}{3} \bra{\psi} Q_i \ket{\psi} 
            + 
        \binom{n-1}{3} (1-\Delta(n-1)) \bra{\psi} (\one - Q_i) \ket{\psi}\\
        &=
        \binom{n-1}{3}\big( 1-\Delta(n-1) + \Delta(n-1) \bra \psi Q_i \ket \psi \big)\, .
        \end{split}
    \end{align}
    Summing this over all $i$, we obtain
    \begin{align}\begin{split}
        (n-3)\bra{\psi} \sum_{T}  P_T \ket{\psi}
        &\leq 
        \binom{n-1}{3} \left( n(1-\Delta(n-1)) + \Delta(n-1)\bra{\psi}\sum_i Q_i \ket{\psi} \right)\\
        &\leq 
        \binom{n-1}{3} n \big( 1-\Delta(n-1) + \Delta(n-1) (1- \delta(n)) \big)\\
        &=
        \binom{n}{3}(n-3)\big( 1 - \Delta(n-1)\delta(n) \big) \, .
    \end{split}
    \end{align}
    Since $1 - \Delta(n) = \binom{n}{3}^{-1} \sup_{\ket \psi} \bra \psi \sum_T P_T \ket \psi$,
    we complete the proof.
\end{proof}

\begin{lemma}\label{lemma:gapboundonaux}
    $g(\beta, ~\tau, ~\Alt(2^n)) \leq \frac{1}{n} + \xi_n$ 
    where $\xi_n = \cO( t^{3}/2^{0.495 n})$.
\end{lemma}

\begin{proof}
    Put $\gamma = 1 - \delta(n) = g(\beta, \ \tau, \ \Alt(2^n))$.
    Since the moment operator~$\avg_i Q_i$ for $\beta$ is hermitian, we have
    \begin{align}
    \begin{split}
        \gamma^2 
        &= \norm*{ \frac{1}{n^2}\sum_{i,j} Q_i Q_j - M(\mu(\Alt(2^n)),\tau) }_\infty\\
        &= \norm*{ \frac{1}{n^2} \sum_i Q_i - \frac 1 n M(\mu(\Alt(2^n)),\tau) + \frac{1}{n^2}\left(\sum_{i \neq j} Q_i Q_j - (n^2-n)M(\mu(\Alt(2^n)),\tau) \right)}_\infty\\
        &\leq \frac{\gamma}{n}  +  \frac{n^2-n}{n^2} \norm*{Q_1Q_n  - M(\mu(\Alt(2^n)),t)}_\infty \, .
        \end{split}
    \end{align}
    If we denote the last norm by $\eta$, \cref{thm:permutationoverlap} says $\eta = \cO(t^3 2^{-0.495 n} )$
    where $0.495 < 1/2$ is arbitrarily chosen.
    Solving the quadratic equation, we see $\gamma \le \frac{1}{2n} + \frac{1}{2n}\sqrt{1 + 4n^2 \eta}$.
\end{proof}

\begin{proof}[Proof of~\cref{lemma:bootstrapingforpermutations}]
    We combine the lemmas above. Suppose $4 \le n_0 < n$. Then,
    \begin{align}
       \Delta(n)
       &\geq \Delta(n-1) \delta(n) \geq \Delta(n_0) \prod_{m = n_0}^n \delta(m) & \text{by \cref{lemma:recursion}} \nonumber\\
       &\geq \Delta(n_0) \prod_{m=n_0}^n
       \left[ 
            1 - \frac{1}{m} - \xi_m 
        \right] &\text{by \cref{lemma:gapboundonaux}} \nonumber\\
        &\ge \Delta(n_0) \left[ \left(\prod_{m = n_0}^n \frac{m - 1}{m}\right) - \sum_{m=n_0}^n \xi_m \right]\\
        &\ge \Delta(n_0) \left[ \frac{1}{n} - \cO(n t^3 2^{-0.495 n_0}) \right] \, .\nonumber
    \end{align}
    Setting $n_0 = \Theta(\log(nt))$ we complete the proof.
\end{proof}

\section{Spectral gaps for random quantum circuits} \label{sec:gap_qckt}

Recall from \cref{def:rqc} that we consider two models of random quantum circuits: 2-local gates with all-to-all connectivity ($\nu_{\mathrm{2,\mathrm{All}\to\mathrm{All}},n}$)
and 2-local gates arranged on a brickwork ($\nu_{\mathrm{BRQC}, n}$).
We shall show that both models of random quantum circuits have $t$-independent spectral gaps that are inverse-polynomial in $n$.
The statement is copied from the introduction:

\gaprqc*

We will frequently use the shorthand $g(\nu, t)$ for the essential norm $g(\nu, t, \SU(2^n))$ 
and $M(\nu, t)$ for the moment operator $M(\nu, t, \SU(2^n))$. 
With this, we can restate the bounds in \cref{thm:main_quantum} compactly as
\begin{align*}
    g(\nu_{\mathrm{2,\mathrm{All}\to\mathrm{All}},n}, \ t) \leq 1 - \Omega(n^{-3}) \,, 
    \qquad \qquad
    g(\nu_{\mathrm{BRQC}, n}, \ t) \leq 1 - \Omega(n^{-5} / \log n) \,.
\end{align*}
Throughout, we write $U^{\otimes t,t}$ to mean $(U \otimes \overline U)^{\otimes t}$ for any operator~$U$.

\subsection{Gap for CPFPC ensemble} \label{sec:cpfpc}

Let $\Cl_n$ be the $n$-qubit Clifford group, 
and let $F$ be the subgroup of $\U(2^n)$ which consists of all diagonal $(\pm 1)$-matrices. 
We use $\Sym^U(2^n)$ to denote the group of all $2^n$-by-$2^n$ permutation matrices $P(\pi) \in \U(2^n)$ 
where $\pi \in \Sym(2^n)$.
The matrix $P(\pi)$ permutes computational basis vectors.
Note that for any $A \in F$ and $P \in \Sym^U(2^n)$ there is $B \in F$ such that $AP = PB$.
This implies that commutativity of two distributions' convolution:
$\mu(\Sym^U(2^n)) * \mu(F) = \mu(F) * \mu(\Sym^U(2^n))$.
We define
\begin{align}
\nu_{\CPFPC}&\deq \mu(\Cl_n)*\mu(F)*\mu(\Sym^U(2^n))*\mu(\Cl_n),\label{eq:definitionCPFPC}\\
\nu_{\FPC} &\deq \mu(F) * \mu(\Sym^U(2^n)) * \mu(\Cl_n).  \nonumber
\end{align}
Using the commutativity 
we see that $\nu_\CPFPC = \mu(\Cl_n)*\mu(\Sym^U(2^n))*\mu(F)*\mu(\Sym^U(2^n))*\mu(\Cl_n)$,
which is manifestly self-adjoint and justifies the notation ``$\CPFPC$.''
(The adjoint of a distribution $\nu$ on a unitary group is the distribution of $U^\dag$ where $U \sim \nu$.)
It is shown in \cite{metger2024simple} that $\nu_{\PFC}$
is close to the Haar measure $\mu(\U(2^n))$ in diamond distance: 

\begin{lemma}[{\cite[Theorem 3.1]{metger2024simple}}] \label{lem:PFC_1norm}
    Let $\Phi_{\nu}^{(t)}$ be the twirling channel with respect to measure $\nu$ defined in \cref{def:unitary_designs}. 
    Then
    \begin{equation}\label{eq:1to1normbound}
        \norm*{\Phi_{\nu_{\FPC}}^{(t)} - \Phi_{\mu_H}^{(t)}}_\diamond = \cO\left(\frac{t}{2^{n/2}}\right)\,.
    \end{equation}
\end{lemma}

Note that the only property of the Clifford group used in the proof of~\cite[Theorem 3.1]{metger2024simple} is that it forms an exact unitary $2$-design. Moreover, with a closer look at the proof of~\cite[Lemma 3.2]{metger2024simple}, one can substitute $\mu(\Cl_n)$ with any approximate unitary $2$-design with multiplicative error $\eps=\cO(1)$. For example, one can use a brickwork random quantum circuit with depth $\cO(n)$ to achieve $\eps=\frac12$ due to the following result.
\begin{lemma}[{\cite[Theorem 2 with $t=2$]{haferkamp2022random}}]
There exists a constant $c>0$ such that for any integer $n\geq 1$, 
    $g(\nu_{\mathrm{BRQC},n}, \ 2)\leq 1 - c$.  
\end{lemma}

Hence, throughout the section, we can safely replace $\mu(\Cl_n)$ with a distribution 
\begin{align}
    \nutwodesign :=(\nu_\mathrm{BRQC})^{*\cO(n)}, \label{eq:nutwodesign}
\end{align}
and the ``$\mathrm{C}$'' in $\nu_{\CPFPC}$ will refer to $\nutwodesign$.
Note that the second ``$\mathrm{C}$'' in $\CPFPC$ is technically the adjoint of $\nutwodesign$ 
({\it i.e.}, the distribution of $U^\dagger$ where $U\sim \nutwodesign$), 
but we will not mention this technicality from now on by assuming that $\nutwodesign$ is self-adjoint.

We have a bound on the essential norm of $\nu_\CPFPC$, which is the foundation of our analysis: 

\begin{lemma}\label{lemma:gapofCPFPC}
    $g(\nu_{\CPFPC}, \ t) = \cO\left(\frac{t}{2^{n/2}}\right)$.
\end{lemma}
\begin{proof}
    Since a bound on the diamond norm implies the same bound on the $1 \to 1$-norm, \cref{lem:PFC_1norm} implies that 
    \begin{align*}
        \norm*{\Phi_{\nu_{\FPC}}^{(t)} - \Phi_{\mu_H}^{(t)}}_{1 \to 1} = \cO\left(\frac{t}{2^{n/2}}\right)\,.
    \end{align*}
    The claim follows from \cref{cor:AdditiveToMultiplicativeError}.
\end{proof}

\subsection{Simulating random phases using permutations} \label{sec:random_phases_via_perm}

The following lemma shows that for all applications in this work, we can work with the alternating group $\Alt(N)$ instead of the full symmetric group $\Sym(N)$.
\begin{lemma}[Moments of small order are equal]
\label{lemma:equivalenceSandAlt}
	For all $t\leq N-2$ we have
	\begin{equation}
		\avg_{\pi\sim \mu(\Alt(N))} P(\pi)^{\otimes t}=\avg_{\pi\sim \mu(\Sym(N))} P(\pi)^{\otimes t},
	\end{equation}
	where $P(\pi)\! \ket{x} = \ket{\pi(x)}$ for all $x\in [N]$ and $\pi \in \Sym(N)$. 
\end{lemma}

\begin{proof}
This is a direct consequence of the fact that $\Alt(N)$ is $t$-transitive for all $t \le N-2$;
$t$-transitivity means that any tuple of $t$ distinct elements can be mapped to any other such tuple under the group action.
In detail,
note that $\avg_{\pi\sim \mu(\Sym(N))} P(\pi)^{\otimes t}$ and $\avg_{\pi\sim \mu(\Alt(N))} P(\pi)^{\otimes t}$ 
are orthogonal projectors onto the subspaces of $\Sym(N)$- and $\Alt(N)$-invariant vectors, respectively.
Since any $\Sym(N)$-invariant vector is $\Alt(N)$-invariant, 
it suffices to show that every $\Alt(N)$-invariant vector in the $t$-fold tensor power vector space is $\Sym(N)$-invariant.
Since $\Alt(N)$ together with any transposition generates $\Sym(N)$, this is equivalent to showing that
for any $x_1, \cdots, x_t\in [N]$ and $\pi \in \Alt(N)$, we can find a transposition $\pi_0$ such that
\begin{align*}
    \Big(\pi(x_1), \cdots, \pi(x_t)\Big) = \Big(\pi_0(\pi(x_1)), \cdots, \pi_0(\pi(x_t))\Big). 
\end{align*}
Since $t\leq N-2$, there must exist distinct $y_1, y_2\in [N]$ 
different from all of $\pi(x_1), \cdots, \pi(x_t)$. Hence, we can take $\pi_0$ as transposing $y_1$ and $y_2$.
\end{proof}

\begin{lemma}\label{lemma:gapforPZP}
    Let $Z$ be the Pauli operator $Z$ on a qubit.
    Let $\nu_{\mathrm{P}Z\mathrm{P}^{-1}}$ be the probability distribution of 
    $P(\pi)\cdot (Z\otimes \one_{n-1}) \cdot P(\pi^{-1}) \in F$ 
    where $\pi \sim \mu(\Alt(2^n))$. 
    Then, $g(\nu_{\mathrm{P}Z\mathrm{P}^{-1}}, \ t, \ F)
    \leq 8t^2/2^n$. 
\end{lemma}

\begin{proof}
We assume $t\leq (2^{n}-2)/4$ as otherwise the claim is trivial.
Writing out the definition, we have
\begin{align*}
	\begin{split}
	\avg_{\pi\sim \mu(\Alt(2^n))} \left(P(\pi)(Z\otimes \one_{n-1})P(\pi^{-1})\right)^{\otimes t,t}&=\avg_{\pi\sim \mu(\Alt(2^n))} P(\pi)^{\otimes 2t}(Z\otimes \one_{n-1})^{\otimes 2t}(P(\pi)^{\top})^{\otimes 2t} \,.
	\end{split}
\end{align*}
Then for all computational basis states $\ket{x_1},\ket{x_2}, \ket{y_1}, \ket{y_2}$, it follows from~\cref{lemma:equivalenceSandAlt} that
\begin{align*}
    \avg_{\pi\sim\mu(\Alt(2^n))}\bra{ x_2} P(\pi)^{\otimes 2t}\ket{x_1}\!\bra{y_1}(P(\pi)^{\top})^{\otimes 2t}\ket{y_2} 
    &= \avg_{\pi\sim\mu(\Alt(2^n))} \bra{x_2} P(\pi)^{\otimes 2t}\ket{x_1}\!\bra{y_2}P(\pi)^{\otimes 2t}\ket{y_1}\\
    &= \avg_{\pi\sim\mu(\Sym(2^n))} \bra{x_2}P(\pi)^{\otimes 2t}\ket{x_1}\!\bra{y_2}P(\pi)^{\otimes 2t}\ket{y_1}\\
    &= \avg_{\pi\sim\mu(\Sym(2^n))} 
    \bra{x_2} P(\pi)^{\otimes 2t} \ket{x_1} \!\bra{y_1} (P(\pi)^{\top})^{\otimes 2t}  \ket{y_2} \, .
\end{align*}
Consequently, 
\begin{align*}
\avg_{\pi\sim \mu(\Alt(2^n))} (P(\pi)(Z\otimes \one_{n-1})P(\pi^{-1}))^{\otimes t,t} = \avg_{\pi\sim \mu(\Sym(2^n))} (P(\pi)(Z\otimes \one_{n-1})P(\pi^{-1}))^{\otimes t,t}
\end{align*}

We identify $F\cong (\mathbb{Z}_2)^{\times 2^n}$ by writing each $D\in F$ as 
\begin{equation}
D=\sum_{x\in \{0,1\}^n} (-1)^{f_D(x)}\proj{x}
\end{equation}
for a Boolean function $f_D:\{0,1\}^n\to \{0,1\}$.
For any probability distribution $\nu$ on $F$ we have 
\begin{align}
\avg_{D\sim \nu} D^{\otimes 2t}= \sum_{x_1,\cdots,x_{2t}\in \{0,1\}^n}\avg_{D\sim \nu} (-1)^{\sum_{j=1}^{2t}f_D(x_j)}\ket{x_1,\cdots,x_{2t}}\!\bra{x_1,\cdots,x_{2t}}.
\end{align}
Moreover,
\begin{align}
	\avg_{D\sim \nu} (-1)^{\sum_{j=1}^{2t}f_D(x_j)}
    &= 2\Pr_{D\sim\nu}\left[ \sum_{j=1}^{2t} f_D(x_j) = 0 \mod 2\right] - 1.
\end{align}
Let $N_x(a) = |\{j\in [2t]: x_j = a\}|$. Then,
\begin{align}
    \Pr_{D\sim \mu(F)} \left[\sum_{j=1}^{2t}f_D(x_j)=0\mod 2\right] = \begin{cases}
        1, & \text{if }N_x(a) \text{ is even for all }a\in \{0,1\}^n, \\
        \frac12, & \text{otherwise.}
    \end{cases}
\end{align}
Hence,
\begin{align}
    g(\nu_{\mathrm{P}Z\mathrm{P}^{-1}}, \ t, \ F) = \max_{x} \ \left|2\Pr_{D\sim \nu_{\mathrm{P}Z\mathrm{P}^{-1}}} \left[\sum_{j=1}^{2t} f_D(x_j) = 0 \mod 2\right] - 1\right| \label{eqn:maxg}
\end{align}
where the maximum is over $x=(x_1, \cdots, x_{2t})$ such that there exists $a\in \{0,1\}^n$ with $N_x(a)$ being odd. 
Let us now assume that $x_1, \cdots, x_{2t}$ are distinct, because repeated $x_j$'s in \cref{eqn:maxg} will cancel and reduce to the same problem for a lower value of $t$.
Hence,
\begin{align}
    \Pr_{D\sim \nu_{\mathrm{P}Z\mathrm{P}^{-1}}} \left[\sum_{j=1}^{2t}f_D(x_j)=0\mod 2\right] 
    & = \sum_{i=0}^t \Pr_{D\sim \nu_{\mathrm{P}Z\mathrm{P}^{-1}}} \left[\sum_{j=1}^{2t}f_D(x_j)=2i \right] \nonumber\\
    &= \sum_{i=0}^t \frac{\binom{2t}{2i} \binom{2^n-2t}{2^{n-1}-2i} }{\binom{2^n}{2^{n-1}}}.
\end{align}
For the last equality, note that sampling a function $f_D$ is equivalent to assigning 0 or 1 to each of the $2^n$ possible inputs $x \in \bits^n$ uniformly at random, under the restriction that there are $2^{n-1}$ many 0's and $2^{n-1}$ many 1's in total.
The total number of such assignments is $\binom{2^n}{2^{n-1}}$ (choose the $2^{n-1}$ positions out of $2^n$ where there is a 1).
To count the number of assignments that satisfy $\sum_{j=1}^{2t}f_D(x_j)=2i$, we can view this as choosing $2i$ out of the first $2t$ positions that get assigned 1, and $2^{n-1} - 2i$ out of the remaining $2^{n} - 2t$ positions that also get assigned a 1.
Hence, the total number of choices is $\binom{2t}{2i}\binom{2^n-2t}{2^{n-1}-2i}$. 

Let $m = 2^{n-1}$. 
We can rearrange and simplify
\begin{align}
    \frac{\binom{2t}{2i} \binom{2m-2t}{m-2i} }{\binom{2m}{m}} = \frac{(2t)!}{(2i)! (2t-2i)!} \frac{(2m-2t)!}{(m-2i)!(m-2t+2i)!}\frac{m!m!}{(2m)!} = \frac{\binom{m}{2i} \binom{m}{2t-2i}}{\binom{2m}{2t}}.
\end{align}
To sum the numerator over $i$,
we examine the coefficients of polynomials $(1+h)^m (1+h)^m = (1+h)^{2m}$ and $(1+h)^m (1-h)^m = (1-h^2)^m$ in variable $h$
to see that
\begin{align}\label{eq:combo_id_one}
    \sum_{i=0}^{2t} \binom{m}{i} \binom{m}{2t-i} = \binom{2m}{2t} \quad\text{ and }\quad \sum_{i=0}^{2t} (-1)^i \binom{m}{i}\binom{m}{2t-i} = (-1)^t\binom{m}{t}.
\end{align}
Summing the two, we have
\begin{align}
    \sum_{i=0}^t \binom{m}{2i} \binom{m}{2t-2i} = \frac12 \left((-1)^t \binom{m}{t} + \binom{2m}{2t}\right),
\end{align}
and thus
\begin{align}
    g(\nu_{\mathrm{P}Z\mathrm{P}^{-1}}, \ t, \ F) &= \max_{k=1,\cdots, t} \ \frac{\binom{m}{k}}{\binom{2m}{2k}}.
\end{align}
Note that this is an exact expression for $g(\nu_{\mathrm{P}Z\mathrm{P}^{-1}}, t, F)$. 
To bound this quantity by something that is easier to deal with, we use the fact that $\binom{m}{k}\leq m^k$ and $\binom{2m}{2k}\geq \frac{(2m-2k)^{2k}}{(2k)!} \geq (\frac{m-k}{k})^{2k}$.
With this, we have
\begin{align}
    \frac{\binom{m}{k}}{\binom{2m}{2k}} \leq \frac{m^k}{(\frac{m-k}{k})^{2k}} = \left(\frac{mk^2}{(m-k)^2}\right)^{k} \leq  \left(\frac{mt^2}{(m-t)^2}\right)^{k}. 
\end{align}
We can assume that $t \leq \frac{\sqrt{m}}{2}$ because otherwise the claim is trivial. Then, $t\sqrt{m} \leq m-t$ and thus
\begin{equation}
    \left(\frac{mt^2}{(m-t)^2}\right)^{k} \leq \frac{mt^2}{(m-t)^2} \leq \frac{4t^2}{m}. \qedhere
\end{equation}
\end{proof}

\begin{corollary}[Gap for unitary design from alternating group]
\label{cor:gap_CPZPC}
Let $\delta(Z)$ denote the probability distribution of applying $Z$ with probability $1$ to the first qubit. 
For all integers $n\geq 1$ and $t \leq \Theta(2^{n/2})$,
\begin{align}
    g(\nutwodesign * \mu(\Alt^U(2^n)) * \delta(Z) * \mu(\Alt^U(2^n)) * \nutwodesign, \ t) = \cO\left(\frac{t}{2^{n/2}}\right).
\end{align}
\end{corollary}
\begin{proof}
By~\cref{lemma:equivalenceSandAlt} we can substitute $\mu(\Alt^U(2^n))$ with $\mu(\Sym^U(2^n))$
in the interested regime of $t$.
In addition,
because of the right- and left-invariance of $\mu(\Sym^U(2^n))$,
we have 
\begin{align}
    \mu(\Sym^U(2^n)) * \nu_{\mathrm{P} Z \mathrm{P}^{-1}} * \mu(\Sym^U(2^n))
    = \mu(\Sym^U(2^n))*\delta(Z)*\mu(\Sym^U(2^n)).
\end{align}
Hence, the essential norm in the claim is equal to
the essential norm of 
\begin{align}
    \nutwodesign * \mu(\Sym^U(2^n)) * \nu_{\mathrm{P}Z\mathrm{P}^{-1}} * \mu(\Sym^U(2^n)) * \nutwodesign . \label{eq:fivedist}
\end{align}
The moment operator of this distribution is the product of five moment operators,
the middle $M(\nu_{\mathrm{P}Z\mathrm{P}^{-1}},~t)$ of which is close to
the projector~$\Pi_F = M(\mu_F,t) = \avg_{D \sim \mu(F)} D^{\otimes t,t}$ by~\cref{lemma:gapforPZP}.
If we replace $\nu_{\mathrm{P}Z\mathrm{P}^{-1}}$ in \cref{eq:fivedist} with $\mu_F$,
then we obtain $\nu_\CPFPC$.
That is,
with $P_H = \E_{U \sim \mu(\SU(2^n))} U^{\ot t, t}$ the Haar projector, we have
\begin{align}
    \label{eq:firstPZPcalculation}
    & g(\nutwodesign * \mu(\Sym^U(2^n)) * \nu_{\mathrm{P}Z\mathrm{P}^{-1}} * \mu(\Sym^U(2^n)) * \nutwodesign, \ t) \nonumber\\
    = \ & \norm*{M(\nutwodesign * \mu(\Sym^U(2^n)), \ t)  \cdot M(\nu_{\mathrm{P}Z\mathrm{P}^{-1}}, \ t) \cdot M(\mu(\Sym^U(2^n)) * \nutwodesign, \ t) - P_H}_{\infty} \nonumber\\
    \leq \ & \norm*{M(\nutwodesign * \mu(\Sym^U(2^n)), \ t) \cdot \left[M(\nu_{\mathrm{P}Z\mathrm{P}^{-1}}, \ t) - M(\mu_F, \ t)\right] \cdot M(\mu(\Sym^U(2^n)) * \nutwodesign, \ t)}_{\infty} \nonumber\\
    & \qquad + g(\nu_{\CPFPC}, \ t) \nonumber\\
    \leq \ & \norm*{M(\nu_{\mathrm{P}Z\mathrm{P}^{-1}}, \ t)-M(\mu_F, \ t)}_{\infty}+g(\nu_{\CPFPC}, \ t) \nonumber\\
    = \ & g(\nu_{\mathrm{P}Z\mathrm{P}^{-1}}, \ t, \ F)+g(\nu_{\CPFPC}, \ t).
\end{align}
\cref{lemma:gapforPZP} says $g(\nu_{\mathrm{P}Z\mathrm{P}^{-1}}, \ t, \ F) = \cO\left(\frac{t^2}{2^n}\right)$
and
\cref{lemma:gapofCPFPC} says $g(\nu_{\CPFPC}, \ t)= \cO\left(\frac{t}{2^{n/2}}\right)$.
\end{proof}

\subsection{All roads lead to local random quantum circuits}\label{section:allroads}

Here,
we develop a general machinery to reduce the spectral gap of random quantum circuits to that of a wide range of probability distributions.
The gaps of random quantum circuits are universal in this sense: 
a gap for any probability distribution that can be approximated 
by polynomial-size quantum circuits with each gate drawn independently from a subgroup of $\U(2^n)$ 
implies a gap for local random quantum circuits.
This reduction has appeared in our bound (\cref{thm:main_rev})
for random classical reversible circuit of the previous section,
and will be a key step in the proof of~\cref{thm:main_quantum}, 
where we will first show a spectral gap for a more structured probability distribution 
and then use that to show a spectral gap for random circuits.

The reduction proceeds in two steps (recall~\cref{def:rqc} for the various random quantum circuits):
\begin{align}
    g(\Asterisk_i \mu(G_i), \ t) \xrightarrow{\Cref{lemma:reductiongeneraltorqc}} 
    g(\nu_{2, \mathrm{All}\to\mathrm{All}}, \ t) \xrightarrow{\Cref{lemma:reductionfromalltoalltolocal}}
    g(\nu_{\mathrm{BRQC}}, \ t),
\end{align}
where $G_i$ are a list of $k$-local unitary groups, 
and the arrow means that a gap for one distribution implies a gap for another distribution. 

We will use the symmetric group~$\Sym(n)$ on $n$ qubit indices,
which is represented by $r : \Sym(n) \to \U(2^n)$ as
\begin{align}
    r(\pi)\ket{x_1, \cdots, x_n} = \ket{x_{\pi(1)}, \cdots, x_{\pi(n)}}.
\end{align}

\begin{lemma}[Gap reduction from general $k$-local subgroup circuits to $2$-local all-to-all random circuits]\label{lemma:reductiongeneraltorqc}
Let $G_1,\ldots,G_L\subseteq \U(2^n)$ be groups 
such that each $G_i$ acts on at most $k$ qubits. 
If $g(\Asterisk_i\mu(G_i), \ t) \leq 1 - \delta$ for some $\delta > 0$, 
then 
\begin{align}
    g(\nu_{2,\mathrm{All}\to\mathrm{All}}, \ t)\leq 1-\frac{\delta \xi}{4L},
\end{align}
for a positive constant $\xi=\Omega(4^{-k}k^{-6})$ that only depends on $k$ but not on $n$ or $t$. 
\end{lemma}
\begin{proof}
    By~\cref{lem:local-vs-parallel} we have
    \begin{equation}\label{eq:unionboundingeneralreduction}
        g(\avg_i\mu(G_i), \ t)\leq 1-\frac{\delta}{4L} \,.
    \end{equation}
    By \cref{fact:gfromoperatorinequality} we may enlarge each $G_i$ to a larger group,
    which we choose to be the full unitary group on the support of~$G_i$.
    Enlarging $G_i$ if necessary, 
    we may further assume that the support of $G_i$ contains exactly $k$ qubits.

    Now, consider an intermediate distribution $\nu$
    defined by replacing each Haar measure $\mu(G_i)$ on a $k$-qubit unitary group
    with $\nu_{2, \mathrm{All}\to\mathrm{All}, k}$ on the same set of qubits as~$G_i$.
    To clarify the support of the group with respect to which we consider $\nu_{2, \mathrm{All}\to\mathrm{All}, k}$,
    for any subset $A \subseteq [n]$ we let 
    \begin{align}
        P_A := M(\mu_A, t) = \avg_{U \sim \mu(\SU(2^k))} (U_A  \otimes \overline U_A)^{\otimes t} \otimes \one_{[n]\setminus A}^{\otimes 2t}
    \end{align}
    be the moment operator of the Haar measure on $k$-qubit unitary group on $A$,
    where we have used subscripts to denote the support.
    The moment operator $P_A$ is an orthogonal projector.
    An unconditional gap theorem~\cite[Theorem 1]{haferkamp2022random} 
    says that for any $t \geq 1$ and $\xi' = 120000^{-1}4^{-k}k^{-5}$,
    \begin{align}
        \one - \frac{1}{k}\sum_{i=1}^k P_{\{i,i+1\,\mathrm{mod}\, k\}} \succeq \xi' (\one - P_{\{1, \cdots, k\}}).
    \end{align}
    If $A=\{j_1,\dots,j_k\}$ where $k \ge 2$ and $j_i< j_{i+1}$, we have 
    \begin{align*}
        \one - \frac{1}{\binom{k}{2}}\sum_{j_1,j_2\in A, \ j_1<j_2}P_{\{j_1,j_2\}}&=\frac{1}{\binom{k}{2}}\sum_{j_1,j_2\in A, \ j_1<j_2}\left(\one - P_{\{j_1,j_2\}}\right)
        \succeq\frac{1}{\binom{k}{2}}\sum_{i=1}^k\left(\one - P_{\{j_i,j_{i+1 \bmod k}\}}\right)
        \succeq \frac{2\xi'}{k-1} (\one - P_A). \numberthis
    \end{align*}
    Let $A_i \subseteq [n]$ be the support of $G_i$. We have assumed $\abs{A_i} = k$.
    Let $B_i = \{\{j_1,j_2\}~|~ j_1,j_2\in A_i \text{ and } j_1 \neq j_2\}$ 
    be the collection of all pairs of qubit indices of $A_i$.
    Then,  with $\xi = 2\xi' / (k-1)$,
    \begin{align}
        M(\avg_i \mu(G_i), \ t) = \one - \frac{1}{L}\sum_{i=1}^L (\one - P_{A_i})
        \succeq \one - \frac{1}{\xi L}\sum_{i=1}^L \left(\one - \frac{1}{\binom{k}{2}} \sum_{(j_1,j_2)\in B_i} P_{\{j_1,j_2\}}\right)
        = \one - \frac{1}{\xi} ( \one - M(\nu, \ t)).
    \end{align}
    where the last equality is by the definition of~$\nu$.
    This immediately implies that $g(\nu, \ t) \leq 1 - \frac{\delta \xi}{4L}$.
    
    Since $\nu$ is a convex mixture of Haar measures of $2$-qubit unitary groups,
    the qubit permutation $\pi \in \Sym(n)$ transforms $\nu$ into another such a mixture that we may denote by $\pi(\nu)$.
    Since $\nu_{2,\mathrm{All}\to\mathrm{All}}$ is the unique mixture of the $2$-qubit Haar measures
    such that $\pi(\nu_{2,\mathrm{All}\to\mathrm{All}}) = \nu_{2,\mathrm{All}\to\mathrm{All}}$ for all $\pi \in \Sym(n)$,
    we have $\avg_{\pi \sim \mu(\Sym(n))} \pi(\nu) = \nu_{2,\mathrm{All}\to\mathrm{All}}$.
    The essential norm is convex with respect to convex combinations of distributions,
    simply because the norm is convex.
    Hence, the spectral gap of $\nu_{2,\mathrm{All}\to\mathrm{All}}$ is at least that of~$\nu$.
\end{proof}

Let $T$ be a set of some transpositions that generates $\Sym(n)$.
This set $T$ represents the connectivity of an $n$-qubit architecture. 
We will borrow results on the eigenvalues of Cayley graphs to study how fast $T$ generates $\Sym(n)$.
Let $\Cay(\Sym(n), T)$ denote its associated Cayley graph,
a directed graph whose nodes are the group elements 
and where a vertex is connected to another by left multiplication by an element of a generating set.
Since transpositions are involutions, $\Cay(\Sym(n), T)$ is undirected.
Let $G_T$ denote the undirected graph with a vertex set $[n]$ 
and distinct vertices $i,j\in [n]$ are joined by an edge 
if and only if $(i,j) \in T$. 
Let $\alpha(G_T)$ denote the algebraic connectivity of $G_T$, 
{\it i.e.}, second smallest eigenvalue of the Laplacian $D-A$,
where $A = (A_{ij})$ is the adjacency matrix with $A_{ij} = 1$ if $(i,j) \in T$ and $0$ otherwise,
and $D$ is diagonal whose entries are the vertex degrees.
Let $\lambda_i(\cdot)$ denote the $i$-th largest eigenvalue (ignoring multiplicity) 
of a matrix or the adjacency matrix of a graph.

\begin{fact}
\label{fact:aldous_transpositions}
    For any set $T$ of transpositions that generates $\Sym(n)$, 
    \begin{align}
        \lambda_2(\Cay(\Sym(n), T)) = \max_{\rho} \ \lambda_1\left(\sum_{a\in T} \rho(a)\right) 
        =|T| - \alpha(G_T),
    \end{align}
    where the maximum is taken over all nontrivial irreps~$\rho$ of $\Sym(n)$ 
    and is attained by the $(n-1)$-dimensional standard representation.
\end{fact}
\begin{proof}
    The first equality follows from the spectral graph theory of Cayley graphs.
    The second equality is given by the well-known Aldous' spectral gap property, 
    which states that the maximum is achieved at the standard $(n-1)$-dimensional 
    representation of $\Sym(n)$ corresponding to the partition $(n-1,1)$~\cite{caputo2010proof,liu2022cayleyeigenvalues}.
\end{proof}

We can thus directly equate the essential norm with the normalized Cayley graph eigenvalue. 
Let $S_n:=r(\Sym(n))$ be the group of all represented unitary matrices of $\Sym(n)$.
\begin{lemma}\label{lem:relate_gap_cayley}
    For any positive integers $n$ and $t$ and any set $T$ of transpositions that generates $\Sym(n)$,
    if $\mu_T$ is the uniform probability measure on~$r(T)$, we have 
    \begin{align}
        g\left(\mu_T, \ t, \  S_n\right) = \frac{1}{\abs T} \lambda_2 (\Cay(\Sym(n), T)).
    \end{align}
\end{lemma}
\begin{proof}
    Note that $\lambda_1(\Cay(\Sym(n), T)) = |T|$.
    It then follows from the first equality in~\cref{fact:aldous_transpositions} that
    \begin{align}
        g\left(\mu_T, \ t, \  S_n\right) \leq \frac{1}{|T|} \lambda_2 (\Cay(\Sym(n), T)).
    \end{align}
    The representation $\Sym(n) \ni \pi \mapsto r(\pi)^{\otimes t}\otimes \overline{r(\pi)}^{\otimes t} = r(\pi)^{\otimes 2t}$ 
    can be decomposed into a direct sum of irreps of $\Sym(n)$,
    and we have to show that the $(n-1)$-dimensional standard representation is always a summand.
    The representation $r$ has a subrepresentation on the span 
    of $\ket{1 0^{n-1}}, \ket{01 0^{n-2}}, \cdots, \ket{0^{n-1}1}$.
    This subrepresentation is $n$-dimensional 
    and is isomorphic to one that acts on $\mathrm{span}\{ \ket a \,:\, a=1,2,\ldots,n\}$.
    We see that $\sum_{i=1}^n \ket{0^{i-1} 1 0^{n-i}}$ is a trivial subrepresentation
    and its orthogonal complement is isomorphic to the $(n-1)$-dimensional standard representation.
    Hence, any tensor power of $r$ contains the standard representation.
    We conclude the lemma using~\cref{fact:aldous_transpositions} again.
\end{proof}

Define a group $G_{\SWAP,i,j} = \{ \one, \SWAP_{i,j} \} \cong S_2$ for distinct $i,j\in [n]$. 

\begin{lemma}[Spectral gaps of SWAPs]\label{lemma:gapforswap_arb_arch}
    Given any connected simple graph~$T$ of $n$ vertices,
    we regard the edges as transpositions of $\Sym(n)$
    and these transpositions generate~$\Sym(n)$,
    so we identify $T$ with a group generating set.
    For all integer $t\geq 1$,
    \begin{equation}
        g(\E_{(i,j)\in T} \mu(G_{\SWAP,i,j}), \ t, \ S_n)\leq 1 - \Omega(n^{-2}|T|^{-1}).
    \end{equation}
\end{lemma}

\begin{proof}
    Let $T_0$ be a spanning tree of~$T$.
    $T_0$ can also be considered as a generating set of transpositions of $\Sym(n)$. 
    Then, it follows from~\cref{lem:relate_gap_cayley} and~\cref{fact:aldous_transpositions} that
    \begin{align}
        g\left(\mu\left(\{\SWAP_{i,j} \}_{(i,j)\in T_0}\right), \ t, \ S_n\right) = 1-\frac{\alpha(G_{T_0})}{|T_0|}.
    \end{align}
    Using~\cref{rem:convex_gap},
    \begin{align}
        g\left(\E_{(i,j)\in T} \mu(G_{\SWAP, i,j}\right), \ t, \ S_n) 
        &\le \frac12 + \frac12 g\left(\mu\left(\{\SWAP_{i,j} \}_{(i,j)\in T}\right), \ t, \ S_n\right) \nonumber\\
        &\leq \frac12 + \frac12\cdot \left( 1-\frac{|T_0|}{|T|}\cdot\frac{\alpha(G_{T_0})}{|T_0|} \right) \nonumber\\
        &= 1 - \frac{\alpha(G_{T_0})}{2|T|}.
    \end{align}
    It is well known that the algebraic connectivity of a tree on $n\geq 3$ vertices 
    is maximized with the value of $1$ if and only if the tree is a ``star,'' 
    {\it i.e.}, $T_0=\{(1,2), (1,3), \cdots, (1,n)\}$, 
    and minimized with the value of $2-2\cos(\pi / n)$ 
    if and only if the tree is a ``path,'' 
    {\it i.e.}, $T_0=\{(1,2), (2,3), \cdots, (n-1,n)\}$~\cite{Grone1990}. 
    This means that we always have $2-2\cos(\pi/n)\leq \alpha(G_{T_0})\leq 1$. 
    Since $1 - \cos(\pi/n)=\Theta(n^{-2})$, we conclude the lemma.
\end{proof}

For the brickwork architecture, since $|T|=n$, 
using~\cref{lem:local-vs-parallel} with $L=n$ and $\ell=2$, we have

\begin{corollary}\label{lemma:gapforswap_brickwork}
    For any positive integers $n$ and $t$ 
    it holds that $g(\Asterisk_{i} \mu(G_{\SWAP, i,i+1}), \ t, \ S_n)\leq 1 - \Omega(n^{-2})$ 
    where the order in the convolution is given by the brickwork layout.
\end{corollary}

Next, we show that a gap for $2$-local all-to-all random quantum circuits 
implies a gap for ($1$D) brickwork random quantum circuits. 
The proof can be easily adapted to arbitrary quantum and classical architectures. 

It is helpful to define a local ($1$D) random quantum circuits model, 
which induces a probability distribution $\nu_{\mathrm{LRQC},n}$ on $\SU(2^n)$ 
by first picking an index $i\in [n]$ uniformly at random 
and then applying a Haar random unitary $U_{i,i+1}$ from $\SU(4)$.
We identify $n+1$ with $1$.

\begin{lemma}[Gap reduction from all-to-all to brickwork connectivity]\label{lemma:reductionfromalltoalltolocal}
    If $g(\nu_{2,\mathrm{All}\to \mathrm{All}}, \ t)\leq 1-\delta$ 
    for some $\delta>0$, 
    then $g(\nu_{\mathrm{LRQC}}, \ t)\leq 1-\Omega\left(\delta n^{-3}/\log \frac 1 \delta\right)$ 
    and $g(\nu_{\mathrm{BRQC}}, \ t) \leq 1- \Omega\left(\delta n^{-2}/ \log \frac 1 \delta \right)$.
\end{lemma}

\begin{proof}
    Note that $\nu_{2, \mathrm{All}\to \mathrm{All}}$ is the distribution of
    $r(\pi) (U_{1,2}\otimes \one_{n-2}) r(\pi^{-1})$ where $\pi\sim \mu(\Sym(n))$ and $U_{1,2}\sim\mu(\SU(4))$. 
    Let us replace $r(\pi) (U_{1,2}\otimes \one_{n-2}) r(\pi^{-1})$ by $r(\pi) (U_{1,2}\otimes \one_{n-2}) r(\sigma)$ 
    for another independent permutation $\sigma\in \Sym(n)$ uniformly at random; 
    that is, we take the convolution $\mu(S_n)*\mu_{H,12}*\mu(S_n)$
    of $\nu_{2, \mathrm{All}\to \mathrm{All}}$ and $\mu(S_n)$.
    We claim that this does not decrease the gap.
    Indeed, if $P_H = \E_{U \sim \mu(\SU(2^n))} U^{\ot t, t}$ is the Haar projector,
    the essential norm of $\mu(S_n)*\mu_{H,12}*\mu(S_n)$ 
    is the norm of $(M(\nu_{2, \mathrm{All}\to \mathrm{All}}) - P_H)M(\mu(S_n))$,
    where we used the invariance of the Haar measure $\mu(\SU(2^n))$ 
    under the convolution by $\mu(S_n)$ on the right.
    Norm is submultiplicative, and $\norm{M(\mu(S_n))} \le 1$.
    Therefore, 
    \begin{align}
        g(\mu(S_n)*\mu_{H,1,2}* \mu(S_n),\,t) \le g(\nu_{2, \mathrm{All}\to \mathrm{All}},\, t) \le 1 - \delta.
    \end{align}
    
    Our goal now is to upper-bound $g(\nu_{\mathrm{LRQC}}, \ t)$ 
    using $g(\mu(S_n)*\mu_{H,12}*\mu(S_n), \ t)$.
    We first replace $\mu(S_n)$ with random local SWAP gates. 
    Let us use the shorthand $G_{\SWAP,i}$ for $G_{\SWAP, i, i+1}$. 
    \cref{lemma:gapforswap_brickwork} implies that it suffices to apply gates from $\Asterisk_i\mu(G_{\SWAP,i})$ 
    (the order in the convolution is given by the brickwork architecture) 
    for $m = \Theta( n^2 \log \frac 1 \delta)$ times 
    where $m$ is sufficiently large such that
    \begin{align}
        \norm*{\E_{\pi\sim \mu(\Sym(n))} r(\pi)^{\otimes t, t} - \E_{U\sim [\Asterisk_i\mu(G_{\SWAP,i})]^{*m}} U^{\otimes t, t}}_\infty &= g(\Asterisk_i\mu(G_{\SWAP,i}), \ t, \ S_n)^{m} \leq \frac \delta 4 \, .
    \end{align}
    Using triangle inequality, 
    \begin{align*}
        g\left([\Asterisk_i\mu(G_{\SWAP,i})]^{*m}*\mu_{H,1,2}*[\Asterisk_i\mu(G_{\SWAP,i})]^{*m}, \ t\right) 
        &\leq g\left(\mu(S_n)*\mu_{H,1,2}* \mu(S_n), \ t\right) + \frac \delta 2 
        \leq 1 - \frac \delta 2 \,. \numberthis \label{eq:sparallelhaarestimate}
    \end{align*}
    
    Next, we will reduce from $2$-local all-to-all random quantum circuits to local random quantum circuits.
    To achieve this, we apply~\Cref{lem:local-vs-parallel}.
    Consider a list of the $2nm+1$ subgroups $G_j$ which are $\{\SWAP_{i,i+1}, \one\}$ and $\SU(4)$ 
    that comprises the convolution $[\Asterisk_i\mu(G_{\SWAP,i})]^{*m}*\mu_{H,1,2}*[\Asterisk_i\mu(G_{\SWAP,i})]^{*m}$.
    Just to make the list perfectly fit into the brickwork structure,
    we insert $n-1$ auxiliary $2$-qubit trivial groups in the convolution after $\mu_{H,1,2}$.
    Let $G_1,\ldots, G_{L}$ be the list of groups with the trivial group insertions,
    where $L = n(2m+1) = \Theta(n^3 \log \frac 1 \delta)$. 
    Let $G'_j$ be $\SU(4)$ sitting on the same 2 qubits as $G_j$;
    the artificial insertion of the trivial groups was to explain this replacement.
    Notice that $\nu_{\mathrm{LRQC}} = \avg_j \mu(G'_j)$ because $G'_j$ are arranged on a brickwork.
    Then,
    \begin{align}
        g(\nu_{\mathrm{LRQC}}, \ t) &\leq g(\avg_j\mu(G_j), \ t) & \text{by \Cref{fact:gfromoperatorinequality}}\nonumber\\
        &\leq 1-\frac{1}{4L}(1-g(\Asterisk_i\mu(G_i), \ t)) & \text{by \Cref{lem:local-vs-parallel}}\nonumber\\
        &= 1-\frac{1}{4L}\left(1-g\left([\Asterisk_i\mu(G_{\SWAP,i})]^{*m}*\mu_{H,1,2}*[\Asterisk_i\mu(G_{\SWAP,i})]^{*m}, \ t\right)\right)       \\
        &\leq 1-\frac{\delta}{8L} & \text{by \cref{eq:sparallelhaarestimate}} \, .\nonumber
    \end{align}
    Finally, we relate the gap of $\nu_{\mathrm{BRQC}}$ to the gap of $\nu_{\mathrm{LRQC}}$. Using~\cref{lem:local-vs-parallel} with $\ell=2$, we have
    \begin{equation}
        g(\nu_{\mathrm{BRQC}}, \ t)\leq 1-\frac{n}{16}(1-g(\nu_{\mathrm{LRQC}}, \ t))\leq 1-\frac{n\delta}{128L}. \qedhere
    \end{equation}
\end{proof}

By completely analogous manipulation we have the following for arbitrary architectures.
We state one for quantum random circuits and another for classical reversible random circuits.
We omit proofs.

\begin{proposition}\label{lem:quantum_arbitrary_arch}
    Let $T$ denote the edges of a connected architecture. 
    Let $\nu_{\mathrm{RQC},T}$ be a probability distribution on $\SU(2^n)$ 
    that first sample a pair $(i,j)\in T$ uniformly at random and 
    then apply a Haar random unitary $U\in \SU(4)$ to the qubits indexed by $i$ and $j$.
    If $g(\nu_{2,\mathrm{All}\to \mathrm{All}}, \ t)\leq 1-\delta$ for some $\delta>0$, 
    then $g(\nu_{\mathrm{RQC}, T}, \ t)\leq 1-\Omega\left(\delta n^{-7} / \log\frac 1 \delta \right)$. 
\end{proposition}

\begin{proposition}\label{lem:reversible_arbitrary_arch}
    Let $T$ denote the edges of any connected architecture. 
    Let $\nu_{\mathrm{rev},T}$ be a probability distribution on $\Sym(2^n)$ 
    that first uniformly sample a triple $(i,j,k)\in [n]^{3}$ 
    such that $(i,j), (j,k)\in T$ and then apply a uniformly random permutation from $\Sym(8)$ to the bits indexed by $i,j,k$. 
    If $g(\nu_{\mathrm{rev},\mathrm{All}\to \mathrm{All}}, \ t)\leq 1-\delta$ for some $\delta>0$, 
    then $g(\nu_{\mathrm{rev}, T}, \ t)\leq 1-\Omega\left(\delta n^{-7} / \log\frac 1 \delta \right)$. 
\end{proposition}

\subsection{Proof of \texorpdfstring{\Cref{thm:main_quantum}}{Theorem 1.5} and \texorpdfstring{\cref{corollary:main}}{Corolary 1.6}} \label{sec:proof_main_thm}

We first estimate the spectral gap for the convolution of uniform measures on local subgroups.

\begin{lemma}[Spectral gap for 3-qubit subgroups] \label{lem:main_18s}
    For any integer $n\geq 4$, there exist $\cO(n^3)$-many 3-qubit unitary subgroups $G_{Q,i} \subset \U(2^n)$ such that,
    \begin{align}
    g(\Asterisk_i\mu(G_{Q,i}), \ t)\leq \frac45 \quad \text{for all}\quad t \leq \Theta(2^{n/2}).
    \end{align}
\end{lemma}
 The proof is a combination of the circuit implementation of Kassabov's generators (\cref{thm:UltimatePermutationDesign}), the gap for unitary design given the full alternating group (\cref{cor:gap_CPZPC}), and detectability Lemmas (\cref{lem:local-vs-parallel})
\begin{proof}[Proof of~\cref{lem:main_18s}]
In the following, $c>0$ will denote a small constant and $C>0$ a large constant.
From~\Cref{lemma:Sgateisgapped} we obtain a list of $3$-qubit subgroups of $\Alt(2^n)$ for any $n\geq 4$, {\it i.e.}, $G_K = \left(G_{K,i} \cong \Sym(2^3)\right)_{i\in \cI_\mathrm{Kas}}$, such that 
\begin{align}
    g\left(\E_{i\in \cI_\mathrm{Kas}}\mu(G_{K,i}),~ \tau,~ \Alt(2^n) \right) = 1 - \Omega(n^{-3}).
\end{align}
Since $|\cI_{\mathrm{Kas}}| = \cO(n)$,
and any $i=(i_1, i_2, i_3)\in I_\mathrm{Kas}$ only intersects with at most $\cO(1)$ other elements in $\cI_\mathrm{Kas}$ (by~\cref{thm:UltimatePermutationDesign}), 
it follows from~\Cref{lem:local-vs-parallel} using $L=\cO(n)$ and $\ell=\cO(1)$ that 
\begin{equation}
    g(\Asterisk_{i\in I_\mathrm{Kas}} \mu(G_{K, i}), \ t)\leq 1-\Omega(n^{-2}).
\end{equation}
Recall from \cref{cor:gap_CPZPC} that
\begin{equation}
    g(\nutwodesign * \mu(\Alt^U(2^n)) * \delta(Z) * \mu(\Alt^U(2^n)) * \nutwodesign, \ t) \leq \cO\left(\frac{t}{2^{n/2}}\right),
\end{equation}
where $\nutwodesign = (\nu_{\mathrm{BRQC}})^{*C_2n}$ for some constant $C_2>0$; see \cref{eq:nutwodesign}.
Let $\nu_Z$ denote the probability measure of applying $Z$ with probability $1/2$ to the first qubit.
It is the Haar measure of a group $\{\one, Z\}$ of order~$2$.
We are going to show that a probability measure
\begin{equation}\label{eq:definitionoftildeG}
	\Asterisk_i\mu(G_{Q,i}) \deq (\nu_{\mathrm{BRQC}})^{*C_2n} *[\Asterisk_i\mu(G_{K,i})]^{* C_1n^{2}}*\nu_{Z}*[\Asterisk_i\mu(G_{K,i})]^{* C_1n^{2}}*(\nu_{\mathrm{BRQC}})^{*C_2n}
\end{equation}
for some large constants $C_1,C_2  > 0$ (by which we define the list $G_{Q,i}$) satisfies that 
\begin{equation}\label{eq:estimateforQparallel}
    g(\Asterisk_i\mu(G_Q), \ t)\leq \frac45,
\end{equation}
We first amplify the gap estimate from~\cref{lemma:Sgateisgapped}:
\begin{align}
    g\left([\Asterisk_i\mu(G_{K,i})]^{* C_1n^{2}}, \ t, \ \Alt^U(2^n)\right) &\leq g(\Asterisk_i\mu(G_{K,i}), \ t, \ \Alt^U(2^n))^{C_1 n^{2}} \leq \frac{1}{100},
\end{align}
where $C_1>0$ is chosen large enough.
We can now repeat the argument in~\cref{eq:firstPZPcalculation} twice,
which is basically triangle inequality,
for $\Asterisk_i\mu(G_{K,i})$ and obtain the bound:
\begin{equation}\label{eq:gapestimatewithdeltas}
    g\left((\nu_{\mathrm{BRQC}})^{*C_2n} *[\Asterisk_i\mu(G_{K,i})]^{* C_1n^{2}}*\delta(Z)*[\Asterisk_i\mu(G_{K,i})]^{* C_1n^{2}}*(\nu_{\mathrm{BRQC}})^{*C_2n}, \ t\right)
    \leq \frac{2}{100} + \cO\left(\frac{t}{2^{n/2}}\right).
\end{equation}
Lastly, we need to substitute $\delta(Z)$ with $\nu_Z$. Since $\Asterisk_i\mu(G_{Q,i})$ is a convex mixture of two probability distributions, one of which satisfies~\cref{eq:gapestimatewithdeltas}, we have
\begin{equation}
   g(\Asterisk_i\mu(G_{Q,i}), \ t)\leq \frac12+\frac{1}{2}\left(\frac{2}{100} + \cO\left(\frac{t}{2^{n/2}}\right)\right)\leq \frac45,
\end{equation}
for $t\leq c2^{n/2}$, for a constant $c>0$.
Enlarging groups does not decrease the gap (\Cref{fact:gfromoperatorinequality}).
\end{proof}

\begin{proof}[Proof of~\cref{thm:main_quantum}]
It follows from~\Cref{lemma:reductiongeneraltorqc} and~\cref{lem:main_18s} that
\begin{align}
    g(\nu_{2,\mathrm{All}\to\mathrm{All}}, \ t) = 1-\Omega(n^{-3}).
\end{align}
Then, \Cref{lemma:reductionfromalltoalltolocal} implies that
\begin{equation*}
     g(\nu_{\mathrm{LRQC}}, \ t) = 1 - \Omega(n^{-6}/\log n) \qquad \text{and} \qquad g(\nu_{\mathrm{BRQC}}, \ t) = 1-\Omega(n^{-5}/\log n ). \qedhere
\end{equation*}
\end{proof}

To prove \cref{corollary:main}, we need the following bootstrapping result from~\cite{brandao2016local}, which shows that a $\Omega(f(n))$ gap for some local random quantum circuits can be reduced to a $1/(n \cdot f(C\log(t)))$ gap.
We will use this to turn a $1/\poly(n)$ gap into a near-optimal $1/n\cdot \mathrm{polylog}(t)$ gap.
\begin{lemma}[Bootstrapping spectral gap from log-sized patches~\cite{brandao2016local}]\label{lemma:bhhreduction}
    For $t\leq \Theta(2^{2n/5})$,
    \begin{align*}
        1 - g(\nu_{\mathrm{LRQC,n}}, \ t) \ge \Omega\left( \frac{1 - g(\nu_{\mathrm{LRQC}, \cO(\log t)}, \ t)}{n \log t} \right) \,.
    \end{align*}
    In particular, if $1 - g(\nu_{\mathrm{LRQC,n}}, \ t) \geq \Omega(n^{-c})$ for some constant $c$, 
    then 
    \begin{align*}
        1 - g(\nu_{\mathrm{LRQC,n}}, \ t) \ge \Omega\left(n^{-1} (\log t)^{-c-1} \right) \,.
    \end{align*}
\end{lemma}

\begin{proof}[Proof (actually transcription of our statement to the language of \cite{brandao2016local})]
    For $i = 1, \dots, n-1$, define the projectors 
    $ P_{i,i+1} = \avg_{U \sim \mu(\SU(4))} (U_{i,i+1} \otimes \overline{U_{i,i+1}})^{\otimes t} $.
    Here, $U_{i,i+1}$ is a shorthand for the 2-qubit unitary $U$ on qubits $(i,i+1)$ and identity on the rest.
    Then we define 
    $ H_{n,t} = \sum_{i=1}^{n-1} (\one - P_{i,i+1}) $.
    We regard this as a geometrically 2-local Hamiltonian in one-dimensional array of $n$ qudits of dimension $2^{2t}$ each.
    By \cite[Lemma 16]{brandao2016local}, 
    $H_{n,t}$ is a frustration-free Hamiltonian 
    and its spectral gap $\Delta(H_{n,t})$ ({\it i.e.}, the smallest nonzero eigenvalue) 
    is related to the spectral gap of the random walk $\nu_{\mathrm{LRQC,n}}$: 
    \begin{align}
        1- g(\nu_{\mathrm{LRQC,n}}, \ t) = \frac{\Delta(H_{n,t})}{n} \,. \label{eqn:gap_to_ham}
    \end{align}
    Using spectral gap techniques~\cite{nachtergaele1996spectral}, 
    \cite[Lemma 18]{brandao2016local} shows that when $n \geq \lceil 2.5 \log_2 (4t) \rceil$, 
    \begin{equation*}
        \Delta(H_{n,t})
        \geq \frac{\Delta(H_{\lceil 2.5\log_2(4t)\rceil,t})}{4\lceil 2.5 \log_2(4t)\rceil} 
        = \Omega\left( \frac{\Delta(H_{\cO(\log t),t})}{\log(t)} \right) \,. \qedhere
    \end{equation*}
\end{proof}

This phenomenon is also true for all-to-all random quantum circuits:
\begin{lemma}[Bootstrapping spectral gap for all-to-all connectivity~\cite{haferkamp2021improved}]\label{lemma:HHJreduction}
    For $t\leq \Theta(2^{2n/5})$,
    \begin{align*}
    1 - g(\nu_{\mathrm{2,\mathrm{All}\to\mathrm{All},n}}, \ t) \ge \Omega\left( \frac{1 - g(\nu_{\mathrm{LRQC},\cO(\log t)}, \ t)}{n} \right) \,.
    \end{align*}
    In particular, if $g(\nu_{\mathrm{LRQC,n}}, \ t) \leq 1 - \Omega(n^{-c})$ for some constant $c$, then it also holds that 
    \begin{align*}
        g(\nu_{\mathrm{LRQC,n}}, \ t) \le 1 - \Omega\left( n^{-1}(\log t)^{-c} \right) \,.
    \end{align*}
\end{lemma}
\noindent\cref{corollary:main} follows directly from plugging 
the bound of~\cref{thm:main_quantum} into~\cref{lemma:bhhreduction,lemma:HHJreduction}.

\section{Efficient reversible circuits generating \texorpdfstring{$\Alt(2^n)$}{Alt(2ⁿ)}}\label{section:efficientreversiblecircuits}

In this section we are concerned with the circuit complexity of generators for alternating groups
that make the Kazhdan constants uniformly bounded below by a universal constant.
All that we need for other results of this paper is contained in \cref{thm:UltimatePermutationDesign},
so a reader may skip the rest of this section to understand the other sections.
We also describe a different set of generators for the alternating group from~\cite{caprace2023tame} 
in \cref{app:caprace_kassabov}. 
The generators in \cref{app:caprace_kassabov} are much simpler, 
but the generators in this section is more efficient and do not use any ancillas.
It is conceivable that the generators in \cref{app:caprace_kassabov}
can be implemented by efficient reversible circuit 
but it appears nontrivial to determine if it can be done without any ancillas
and to adapt it to $\Alt(2^n)$.
All these features will contribute to our main results on random circuits.

We use a result of Kassabov as follows.

\begin{theorem}[Theorem~2.1 of~\cite{Kassabov_2007_alt}]\label{thm:Kassabov2007Alt}
    There exists a real~$\eps > 0$ and an integer~$k \ge 1$
    such that for any integer $s > 6$
    there exists a generating set~$S$ of $k$ elements for~$\Alt((2^{3s}-1)^6)$
    such that the Kazhdan constant $\cK(\Alt((2^{3s}-1)^6); S)$ is at least~$\eps$.
\end{theorem}

Our result is the following.

\kascircuitrev

\noindent
Here, a circuit is \textit{strongly explicit} if the circuit's description
(a list of tuples of the locations and the names of gates)
is obtained in $\poly(n)$ time given a number $n$.
By repeating randomly chosen generators from \cref{thm:UltimatePermutationDesign} in sequence, one can generate approximately $t$-wise independent permutations. 
However, it turns out that if one is only interested in approximately $t$-wise independent permutations (and not random circuits), there is a more efficient construction, which we describe in \cref{sec:nt_depth_permutations} and which might be of independent interest.

To a large extent, our exposition until we discuss circuits
is our own detailed, though still limited, guide to~\cite{Kassabov_2007_lattices},
focusing only on the necessary algebraic elements.
Many implicit proofs of~\cite{Kassabov_2007_lattices} will be made hopefully more accessible.
The circuits we derive are in some sense presented in~\cite{Kassabov_2007_alt,Kassabov_2007_lattices};
the contribution of the first half of this section is to spell them out explicitly 
and make the construction of \cite{Kassabov_2007_alt} via~\cite{Kassabov_2007_lattices} more accessible.
We remark that 
our circuit construction and its modification would not make sense 
if we used Kassabov's result in a blackbox fashion.
The remainder of this section constitutes a proof of \cref{thm:UltimatePermutationDesign}.

\subsection{Review of Kassabov's generators}

In~\cite{Kassabov_2007_alt}, an alternating group is identified 
with an even permutation group on a rather specially structured set of points. Namely, the alternating group permutes $K^6$ points
arranged in a 6-dimensional hypercube of linear dimension~$K = 2^{3s} - 1$ 
for a positive integer~$s$.
So, the group of interest is $\Alt((2^{3s}-1)^6) = \Alt(K^6)$. The number $2^{3s}-1$ naturally arises since 
\begin{align}
\SL(3s;\FF_2)   \quad \text{permutes}\quad \FF_2^{3s} \setminus \{ 0\}.
\end{align}
That is, the special linear group~$\SL(3s;\FF_2) = \GL(3s; \FF_2)$ acting on the vector space~$\FF_2^{3s}$ can be regarded as permutations. In particular, non-zero vectors are mapped bijectively to non-zero vectors, and hence, the set of $2^{3s}-1$ nonzero vectors in~$\FF_2^{3s}$ are permuted. 

The six-dimensional hypercube is a collection of $K^5$ ``lines'' parallel to a coordinate axis,
each line consisting of~$K = 2^{3s} -1$ points.
Each line is identified with the set of all nonzero vectors in~$\FF_2^{3s}$.
Define 
\begin{align}
    \bKK _s = (\FF_2^{3s} \setminus \{ 0\} )^{\times 6} \subseteq (\FF_2^{3s})^{\times 6}.
\end{align}
Every element of~$\Alt(\bKK_s)$ is a bijective function from~$\bKK_s$ onto itself.
The group $\Alt(\mathbf K_s)$ is \emph{not} the full alternating group~$\Alt(2^{18s})$,
but this will be resolved toward the end of this section.

\subsubsection{A large generating set for $\Alt((2^{3s}-1)^6)$}

The special linear group $\SL(3s;\FF_2)$ is generated by \emph{elementary matrices} of form
\begin{align}
    E_{i,j} = \one_{3s} + r e_{i,j}
\end{align}
where $1\le i \neq j \le 3s$, 
and $r$ is an element of the coefficient ring, which is currently just~$\FF_2$,
and $e_{i,j}$ is a matrix with a sole~$1$ at $(i,j)$.
(Think of Gauss eliminations.)
Every elementary row operation is order~$2$ (because of~$\FF_2$),
so its disjoint cycle decomposition as a permutation on the line of~$K$ points
consists of transpositions.
The number of the transpositions is~$2^{3s-2}$, which is always an even number because~$3s \ge 3$.
Therefore, $E_{i,j}$ gives an even permutation.
This defines an embedding (injective group homomorphism) $\SL(3s;\FF_2) \hookrightarrow \Alt(K)$,
and furthermore we have 
\begin{align}
    \Gamma = \SL(3s; \FF_2)^{\times K^5} \hookrightarrow \Alt(\mathbf K_s)
\end{align}
where $\Gamma$ is a direct product of $K^5 = (2^{3s}-1)^5$ copies of~$\SL(3s; \FF_2)$. 
We have not specified the axis to which the $K^5$ lines are parallel.
Indeed, there are six different choices of the axis,
and each gives an embedding of~$\Gamma$ into the alternating group of interest:
\begin{align}
    \pi_i : \Gamma \hookrightarrow \Alt(\mathbf K_s), \qquad i = 1,2,\ldots,6 .
\end{align}

It is asserted in~\cite{Kassabov_2007_alt} that the \emph{union} $\bigcup_{i=1}^6 \pi_i(\Gamma)$
generates the full~$\Alt(\mathbf K_s)$.
For a proof, we find it convenient to use a result of Pyber~\cite{Pyber1993}.
\footnote{Pyber notes that this result is weaker than what is implied by the classification of all doubly transitive groups,
but Pyber's result comes with a simple proof.}
Recall that an action of a group~$G$ on a set~$S$ is said to be \emph{doubly transitive}
if for any four elements~$a,b,a',b' \in S$ such that $a \neq b$ and $a' \neq b'$, 
there exists~$g \in G$ such that~$g \cdot a = a'$ and~$g \cdot b = b'$.
Recall also that a \emph{permutation group} of degree~$N$ is a subgroup of the symmetric group on $N$ letters,
so the symmetric group is the largest permutation group of degree~$N$,
which has order $N! \approx (N/e)^N$.
Any permutation group of degree~$N$ is endowed with the obvious group action on a set of~$N$ elements.
The result of Pyber is that every doubly transitive permutation group of degree~$N \ge 400$
either contains~$\Alt(N)$ or has order at most~$N^{33 (\log_2 N)^2}$.
Observe that this bound on the group order is much smaller than~$(N/e)^N$ asymptotically.

Let us show that 
the permutation group~$G = \langle \bigcup_{i=1}^6 \pi_i(\Gamma) \rangle$ of degree~$N = K^6$
is doubly transitive.
Note that the group $\SL(3s;\FF_2)$ is transitive on a line of~$K$ points
since two different nonzero vectors over~$\FF_2$ are linearly independent.
For any pair of distinct points $a=(a_1,\ldots, a_6)$ and $b=(b_1, \ldots, b_6)$ 
of the six-dimensional hypercube and for any $x, y\in \FF_2^{3s}\setminus \{0\}$, 
if $(a_2, \ldots, a_6)\neq (b_2, \ldots, b_6)$, then
there exists $g\in \Gamma$ such that
\begin{align}
    g\cdot a = (x, a_2, \ldots, a_6), \qquad   g\cdot b = (y, b_2, \ldots, b_6).
\end{align}
Hence, we can use this action to move any pair of points, $a$ and~$b$, to a fixed pair of points, say $x_0$ and $y_0$.
This shows the double transitivity: $(a,b) \leftrightarrow (x_0,y_0) \leftrightarrow (a',b')$.

On the other hand, we can lower bound the order of the group generated by~$\bigcup_i \pi_i(\Gamma)$.
This is obviously bigger than the order of~$\Gamma$.
Since the group~$\SL(3s;\FF_2)$ contains all upper triangular matrices 
where $(3s)(3s-1)/2$ off-diagonal elements are arbitrary,
we have $\abs \Gamma = \abs{\SL(3s;\FF_2)}^{K^5} \ge (2^{3s(3s-1)/2})^{K^5}$.
Since this lower bound is bigger than Pyber's upper bound for~$s \ge 1$,
the permutation group~$G$ must contain~$\Alt(\mathbf K_s)$.

A generating set for~$\Alt(\mathbf K_s)$ with respect to which the Kazhdan constant is uniformly bounded below,
is defined to be the \emph{union} $\bigcup_{i=1}^6 \pi_i(\Gamma_0)$ 
where $\Gamma_0 \subseteq \Gamma$ is a generating set for~$\Gamma$.
In view of the action of~$\Alt(\mathbf K_s)$ on $(\FF_2^{3s})^{\times 6}$,
different embeddings can be realized by permuting bits.
Hence, it will suffice for us to look into the generator set~$\Gamma_0$ of~$\Gamma$.
The group~$\Gamma$ is a direct product of many (exponential in~$s$) 
copies of~$\SL(3s; \FF_2)$.
Naively, a generating set for~$\Gamma$ must be large too,
but a certain quotient polynomial ring will save the situation.

\subsubsection{Product of many matrix rings}

One of the important ideas in~\cite{Kassabov_2007_alt} 
is that the direct product of copies of~$\SL(3s;\FF_2)$
is an image of some small rank matrix group.
To better understand this, we have to relate a product of linear groups to a linear group over a product ring.
A slightly abstract treatment will make the construction more accessible.

Let $A$ and $B$ be unital rings that may be noncommutative.
The product ring $A \times B$ is set-theoretically the Cartesian product set,
and the ring structure is given by component-wise operation.
The multiplicative identity is~$(1,1)$.
It is important to note that there exist \emph{orthogonal idempotents} $e_A = (1,0) = e_A^2$ and $e_B = (0,1) = e_B^2$ 
with $e_A e_B = 0$ and $e_A + e_B = (1,1)$.
These idempotents commute with every element of $A \times B$.
One may write that $(A\times B)e_A \cong A$ and $(A \times B) e_B \cong B$.
That is, the idempotents are projections (actually ring homomorphisms) onto the ring factors.
We are going to work with
\begin{definition}
    An \emph{elementary matrix group} over a unital ring~$R$ (that may be noncommutative) of rank~$m$,
    denoted by~$\EL(m;R)$,
    is a multiplicative matrix group generated by all elementary matrices $E_{i,j} = \one_m + r e_{i,j}$
    where $1 \le i \neq j \le m$ and $r \in R$.
\end{definition}
Note that even if~$r \in R$ is invertible (in a general ring), 
a diagonal matrix with diagonal entries~$r,1,1,\ldots,1$
may or may not be a member of the elementary matrix group. 
For a commutative unital ring~$R$, the special linear group~$\SL(m;R)$ is defined,
and $\EL(m;R)$ is a subgroup of~$\SL(m;R)$.
However, the two may be different in general.
If $R$ is a field, then~$\SL(m;R) = \EL(m;R)$.
\footnote{Proof: The first row of an $m \times m$ matrix from a special linear group over a field must be nonzero,
and hence by column operations (multiplication by elementary matrices on the right)
one can bring the row to, \emph{e.g.}, $(0,r,0,0,\ldots,0)$ where $r \neq 0$.
By another column operation, this row is brought to~$(1,r,0,0,\ldots,0)$
and then to~$(1,0,0,0,\ldots,0)$.
Then, the lower right $(m-1) \times (m-1)$ submatrix must have determinant~$1$,
as seen by the cofactor expansion of the determinant on the first row,
and inductively the whole matrix can be brought to a lower triangular matrix by column operations,
which in turn can be brought to the identity matrix.}
We will find a surjection from an elementary matrix group over a noncommutative ring of a constant rank 
onto~$\EL(m;\FF_2) = \SL(m;\FF_2)$.

Now, consider an elementary matrix group~$\EL(m; A \times B)$ over a product ring~$A \times B$.
Given $M \in \EL(m; A \times B)$, taking the entry-wise projection, we have two matrices $M e_A$ and $M e_B$
in the respective rings. 
Conversely, given $M_A \in \EL(m;A)$ and $M_B \in \EL(m;B)$,
the matrix $(M_A,M_B)$ constructed by putting corresponding elements into tuples,
is a member of $\EL(m;A \times B)$.
It is routine to check that this correspondence obeys the group operation of the matrix group~$\EL$,
and hence the idempotents give a group isomorphism
$\EL(m; A \times B) \cong \EL(m; A) \times \EL(m ; B)$.
It follows that for any positive integer~$k$
\begin{align}
    \EL(m; R )^{\times k} \cong \EL(m; R^{\times k} ) \label{eq:ProductGroupProductRing}
\end{align}
for any associative unital ring~$R$.
Hence, the product of many copies of a elementary matrix group
is just one elementary matrix group over a ring with many idempotents.
Note that the set of all idempotents has a partial ordering: $e \le e'$ if $e e' = e$,
so we may speak of minimal idempotents.

\begin{lemma}\label{lem:ProductRing}
    Let $p$ be a prime.
    The quotient polynomial ring~$\FF_p[x]/(x^p - x)$
    is ring isomorphic to the product ring~$\FF_p ^p$.
    The $p$ orthogonal minimal idempotents are labeled by $a \in \FF_p$:
    \begin{align}
        e_a = \prod_{j \in \FF_p \setminus\{a\}} \frac{x - j}{a - j} .
    \end{align}
\end{lemma}

\begin{proof}
    We have to show that $e_a^2 = e_a$ for each~$a$ and $e_a e_b = 0$ if $a \neq b$.
    Once this is shown, $\FF_p$-dimension counting implies that there is no smaller idempotent.
    In $\FF_p[x]/(x^p -x)$ every element is represented by a polynomial of degree at most $p-1$.
    In particular, $e_a^2$ is represented by a degree~$p-1$ polynomial~$f \in \FF_p[x]$.
    That is, $e_a^2 - f = r (x^p -x)$ for some $r \in \FF_p[x]$.
    Fermat's little theorem says that $y^p = y \bmod p$ for all integer~$y$.
    So, the three polynomials $f,e_a, e_a^2 \in \FF_p[x]$ all have the same $p-1$ distinct roots ($\FF_p \setminus \{a\}$).
    Therefore $f$ and $e_a$ must be the same up to an overall scalar.
    That is, $e_a^2$ is a multiple of $e_a$ in the quotient ring.
    Since they both evaluate to $1$ at $x = a$, the multiplier is~$1$.
    This shows that $e_a^2 = e_a$ in the quotient ring.
    Similarly, the product $e_a e_b$ vanishes for all $x \in \FF_p$,
    so it is a multiple of $(x-1)(x-2) \cdots (x-p)$,
    which is a degree~$p$ polynomial 
    that must be represented by a degree $p-1$ polynomial in the quotient ring,
    where the latter cannot have more than $p-1$ distinct roots unless it is identically zero.
    So, $(x-1)(x-2) \cdots (x-p) = x^p - x \in \FF_p[x]$, 
    and $e_a e_b$ is zero in the quotient ring.
\end{proof}

\begin{corollary}\label{cor:RPolyQuotient}
    Let $R$ be any algebra over~$\FF_p$ for a prime~$p$.
    Then, $R[x_1,\ldots,x_k] / (x_1^p - x_1, \ldots, x_k^p - x_k) \cong R^{\times p^k}$ as rings.
\end{corollary}

The quotient ring is a unital ring that may be noncommutative, but the indeterminants~$x_j$ are in the center.

\begin{proof}
    It suffices to understand $k=2$.
    Observe that for any $a,b \in \FF_p$,
    the polynomial 
    \begin{align}
        \left(\prod_{i \in \FF_p \setminus\{a\}} \frac{x_1 - i}{a - i}\right)
        \left(\prod_{j \in \FF_p \setminus\{b\}} \frac{x_2 - j}{b - j}\right)
    \end{align}
    is a minimal idempotent.
\end{proof}

\begin{remark}\label{rem:DeltaFunctionAlgGeom}
It is the best to regard the quotient ring $R[x_1,\ldots,x_k] / (x_1^p - x_1, \ldots, x_k^p - x_k)$
as the set of all $R$-valued functions on an algebraic variety,
which in this case is a $k$-dimensional box with $p$ points along each dimension.
Indeed, $\FF_p[x_1,\ldots,x_k]$ is the set of all polynomial functions valued in~$\FF_p$ over a $k$-dimensional affine space,
and the zero locus of the ideal $(x_i^p -x_i)$ consists of points whose coordinates are in~$\FF_p$ by Fermat's little theorem.
In this perspective, the idempotents are the functions that are supported on exactly one point,
at which the function assumes~$1$. The product of two different such delta functions must be zero,
and the role of this delta function is to read off a general function's value at the support point.
This remark will be useful below when we consider algorithmic aspects with product rings.
\end{remark}

\begin{lemma}\label{lem:product_ring_square_gen_set}
    For a prime~$p$, the product ring~$\Mat(m;\FF_p)^{\times p^{k m^2}}$ 
    of $p^{k m^2}$ copies of the matrix ring~$\Mat(m; \FF_p)$
    is generated by $k+2$ elements.
\end{lemma}

It follows that any ring homomorphic image of~$\Mat(m;\FF_p)^{\times p^{k m^2}}$ is also generated by $k+2$ elements.
In particular, $\Mat(m; \FF_p)^{\times l}$ is generated by~$k+2$ elements as long as $l \le p^{k m^2}$.

\begin{proof}
    Consider a quotient polynomial ring with coefficients in the matrix algebra:
    \begin{align*}
        \Mat(m;\FF_p)[x_1,\ldots, x_u] / (x_1^p - x_1, \ldots, x_u^p - x_u).
    \end{align*}
    This is clearly generated by $\Mat(m;\FF_p)$ and $u$ indeterminants.
    Following~\cref{cor:RPolyQuotient}, this quotient polynomial ring is isomorphic to $\Mat(m;\FF_p)^{\times p^u}$.
    Hence, it suffices to show that $\Mat(m;\FF_p)[x_1,\ldots, x_u]$ is generated by $k+2$ elements where $u=km^2$. 
    Notice that the matrix algebra~$\Mat(m; \FF_p)$ is generated by two elements
    \begin{align}
        A = \begin{pmatrix}
            0 & 1 &   &   &   &    \\
              & 0 & 1 &   &   &    \\
              &   &   & \ddots & \ddots & \\
              &   &   &        &   0    & 1\\
            1 &   &   &   &   &   0          
        \end{pmatrix}, \quad
        B = \begin{pmatrix}
            1 & 0 &   &   &   &    \\
            0 & 0 &  &   &   &    \\
              &   & 0 &  &  & \\
              &   &   &  \ddots &       & \\
             &   &   &   & 0  &  
        \end{pmatrix} . \label{eq:matricesAB}
    \end{align}
    The first is a cyclic permutation and the second is a diagonal matrix with sole nonzero element at the top left.
    By multiplying $B$ by $A$ from the left, we can bring the one to any row,
    and by multiplying $B$ by $A$ from the right, we can bring the one to any column.

    It remains to generate $\Mat(m;\FF_p)[x_1,\ldots, x_u]$ where $u = k m^2$
    by only $k+2$ elements.
    The trick~\cite[Footnote 5]{Kassabov_2007_lattices} is to use matrices to ``store'' indeterminants.
    Define
    \begin{align}
        y_i = \begin{pmatrix}
            x_{i,1,1} & x_{i,1,2} & \cdots & x_{i,1,m} \\
                      &           & \cdots &           \\
            x_{i,m,1} & x_{i,m,2} & \cdots & x_{i,m,m}
        \end{pmatrix} \label{eq:matrixYi}
    \end{align}
    for $i = 1,2,\ldots,k$.
    Since $A,B$ generate $\Mat(m;\FF_p)$, 
    we have a projector (a diagonal matrix) that singles out any row from~$y_i$ by acting on the left.
    Similarly, we can single out any column.
\end{proof}

When we say that a ring is generated by a set $\{y_j\}$,
we use addition and multiplication of an arbitrary finite number of elements from~$\{y_j\}$ in an arbitrary order.
This can be formalized by saying 
that there is a surjective ring homomorphism from the free ring~$\ZZ \langle \{ \byy_j \} \rangle$,
where $\byy_j$ are \emph{noncommutative} indeterminants with absolutely no relations among them.
The preceding lemmas are summarized by saying that
we have a chain of surjective ring homomorphisms:
\begin{align}
    \ZZ \langle \baa, \bbb, \byy_1, \ldots, \byy_k \rangle
    \xrightarrow{\quad \varphi_1 \quad}
    \Mat(m; \FF_p)[x_{1,1,1}, \ldots, x_{k,m,m}]
    \xrightarrow{\quad \varphi_2 \quad}
    \Mat(m; \FF_p)^{\times p^{k m^2}}
\end{align}
where
\begin{align}
    \varphi_1(\baa) = A, \quad    \varphi_1(\bbb) = B \text{ of \cref{eq:matricesAB}}, \quad \text{ and }
    \varphi_1(\byy_j) = y_j \text{ of \cref{eq:matrixYi}}.
\end{align}

\begin{remark}\label{rmk:ring_iso_2}
The second map~$\varphi_2$ needs further unpacking.
The last product ring is isomorphic to
\begin{align}
    \Mat(m; \FF_p)[x_{1,1,1}, \ldots, x_{k,m,m}] / (x_{i,a,b}^p - x_{i,a,b}).
\end{align}
There is a canonical ring homomorphism from $\Mat(m; \FF_p)[x_{1,1,1}, \ldots, x_{k,m,m}]$ onto this quotient ring,
from which the projection onto each component of the product ring $\Mat(m; \FF_p)^{\times p^{k m^2}}$ is achieved by a minimal idempotent.
The $p^{k m^2}$ components are indexed by a string~$(c_{1,1,1},\ldots,c_{k,m,m}) \in \FF_p^{k m^2}$.
The corresponding idempotent is 
$\prod_{i=1}^k \prod_{a=1}^m \prod_{b=1}^m \prod_{j \in \FF_p, j \neq c_{i,a,b}} \frac {x_{i,a,b} - j}{c_{i,a,b}-j}$.
While this formula looks clumsy,
\cref{rem:DeltaFunctionAlgGeom} provides a nice way of thinking about them.
The idempotents can be (and should be) regarded as
delta functions that evaluate to~$1$ at $(x_{i,a,b}) = (c_{i,a,b})$ and vanish elsewhere.
Hence, we can retrieve a matrix for each component of~$\Mat(m;\FF_p)^{\times p^{k m^2}}$ by
simply evaluating an $m \times m$ matrix with entries in $\FF_p[x_{1,1,1}, \ldots, x_{k,m,m}]$
at $x_{i,a,b} = c_{i,a,b}$.
\end{remark}

\begin{remark}\label{rmk:kas_gen_set_diag}
    If we remove all the off-diagonal variables in $y_i$ and only keep the diagonal variables,
    we can easily show that $\Mat(m;\FF_p)[x_1,\ldots, x_{km}]$ can be generated by a different set
    \footnote{
        The generating set for $\Mat(m;\FF_p)^{\times p^{k m^2}}$ in~\cref{lem:product_ring_square_gen_set} was the one used in~\cite[Theorem 8b]{Kassabov_2007_lattices} but its explicit form was never discussed. However, the generating set for $\Mat(m;\FF_p)^{\times p^{k m}}$ in~\cref{rmk:kas_gen_set_diag} was described 
        explicitly in~\cite[Lemma 4.1]{Kassabov_2007_lattices} but it appears that Kassabov did not use it.
        It worth mentioning that the proof of~\cite[Theorem 8b]{Kassabov_2007_lattices} is unaffected.
    }
    of $k+2$ elements: $\{A, B, z_i \,|\, i=1, 2, \cdots, k\}$ where
    \begin{align}
        z_i = \text{diag}(x_{i,1,1}, \ x_{i,2,2}, \ \cdots, \ x_{i,m,m}) = \begin{pmatrix}
            x_{i,1,1} & &  & 0 \\
             & x_{i,2,2} & & \\
             & & \ddots & \\
            0 & & & x_{i,m,m}
        \end{pmatrix}. 
    \end{align}
    This observation is the key to~\cref{sec:kas_depth_one} when we reduce each generator's circuit depth to $1$.
\end{remark}

\subsubsection{Generating elementary matrix groups}

Above we have discussed how a certain product ring of matrices is generated.
We have to use these ring generators to generate an elementary matrix group~$\EL$,
which in turn generates an alternating group.
The surjective ring homomorphisms~$\varphi_1, \varphi_2$ 
give a chain of \emph{surjective} group homomorphisms
\begin{align}
    &\EL(3; \ZZ \langle \baa, \bbb, \byy_1, \ldots, \byy_k \rangle ) \nonumber\\
    &\xrightarrow{\quad \varphi_1 \quad}
    \EL(3; \Mat(m; \FF_p)[x_{1,1,1}, \ldots, x_{k,m,m}]) 
    =
    \EL(3m; \FF_p[x_{1,1,1}, \ldots, x_{k,m,m}] ) \label{eq:ChainOfEL}\\
    &\xrightarrow{\quad \varphi_2 \quad}
    \EL(3m; \FF_p^{\times p^{k m^2}} ) =
    \EL(3m; \FF_p)^{\times p^{k m^2}}\nonumber
\end{align}
where we have abused notations to overload $\varphi_1,\varphi_2$ 
to denote the group homomorphisms defined by
applying the ring homomorphisms to every entry of a matrix.
There are two equalities in this chain of mappings,
the second of which is just \cref{eq:ProductGroupProductRing}.
The first equality needs a proof.

\begin{lemma}
    Let $R$ be any unital ring that may be noncommutative.
    Then, for any $n,m \ge 2$ we have
    \begin{align}
        \EL(n; \Mat(m; R)) = \EL(nm; R ).
    \end{align}
\end{lemma}

In the left-hand side
we have $n \times n$ matrices whose elements are $m \times m$ matrices.
So, it is naturally an $nm \times nm$ matrix, and algorithmically there is nothing to do.

\begin{proof}
    It is obvious that ``$\subseteq$'' is true;
    the set of all elementary matrices on the right-hand side is bigger.
    For the other inclusion, we have to show that any elementary matrix~$\one_{nm} + r e_{i,j}$ with $i \neq j$ and $r \in R$
    is a product of matrices of form~$\one_{nm} + M$ where $M$ is zero 
    except for one of the $n^2 - n$ off-diagonal $m \times m$ blocks.
    The only nontrivial case is where the entry $r$ at $(i,j)$ 
    is placed in one of the diagonal $m \times m$ blocks.
    It suffices to consider $1 \le i < j \le m$.
    Define $E = \one + r e_{i,m+1}$ and $F = \one + e_{m+1,j}$.
    Then, $E^{-1} = \one - r e_{i,m+1}$ and $F^{-1} = \one - e_{m+1,j}$.
    These are clearly members of the left-hand side,
    and so is $E F E^{-1} F^{-1} 
    = \one + r e_{i,j} $.
\end{proof}

It remains to find a generating set for the first elementary matrix group 
$\EL(3; \ZZ \langle \baa, \bbb, \byy_1, \ldots, \byy_k \rangle )$
over a free noncommutative ring.
This will in turn specify generators of $\EL(3m; \FF_p)^{\times p^{km^2}}$.
We will eventually set $p = 2$.

A ring has two operations (addition and multiplication), while a group has only one.
Hence, it is not immediately obvious how ring generation is related to group generation.
But this is not hard.

\begin{lemma}\label{lem:EL3Generators}
    An elementary matrix group over a free ring with $l$ generators,
    $\EL(3; \ZZ \langle \bww_1,\ldots,\bww_l \rangle )$, is generated by
    the following $4l+6$ elements and their inverses:
    \begin{align}
        E_{a,b}(1) \text{ where }a \neq b, \qquad E_{1,2}(\bww_j),\quad E_{2,3}(\bww_j),
        \quad E_{2,1}(\bww_j),\quad E_{3,2}(\bww_j).
    \end{align}
\end{lemma}
\begin{proof}    
    The following matrix identities show that matrix multiplication encodes ring operations.
    \begin{align}
    \begin{pmatrix}
        1 & 0 & \bxx \byy \\
        0 & 1 & 0 \\
        0 & 0 & 1
    \end{pmatrix}
    &=
    \begin{pmatrix}
        1 & -\bxx & 0 \\
        0 & 1 & 0 \\
        0 & 0 & 1
    \end{pmatrix}
    \begin{pmatrix}
        1 & 0 & 0 \\
        0 & 1 & -\byy \\
        0 & 0 & 1
    \end{pmatrix}
    \begin{pmatrix}
        1 & \bxx & 0 \\
        0 & 1 & 0 \\
        0 & 0 & 1
    \end{pmatrix}
    \begin{pmatrix}
        1 & 0 & 0 \\
        0 & 1 & \byy \\
        0 & 0 & 1
    \end{pmatrix} , \label{eq:RingOperationsInEL}\\
    \begin{pmatrix}
        1 & \bxx + \byy\\
        0 & 1 \\
    \end{pmatrix}
    &=
    \begin{pmatrix}
        1 & \bxx \\
        0 & 1 
    \end{pmatrix}
    \begin{pmatrix}
        1 & \byy \\
        0 & 1
    \end{pmatrix}.
    \nonumber
    \end{align}
    These imply that for any elements $\bxx,\byy$ of a ring
    \begin{align}
        E_{a,c}(\bxx \byy) &= E_{a,b}(-\bxx) E_{b,c}(-\byy) E_{a,b}(\bxx) E_{b,c}(\byy) &\text{with $a,b,c$ all distinct},\\
        E_{a,b}(\bxx + \byy) &= E_{a,b}(\bxx) E_{a,b}(\byy) &\text{with $a,b$ distinct}.\nonumber
    \end{align}
    Therefore, these generate every elementary matrix~$E_{a,b}(r)$ for an arbitrary element~$r$ of the free ring
    and any~$a \neq b$.
\end{proof}

\subsubsection{A small generating set for  $\SL(3s;\F_2)^{\times 2^{15 s}}$}\label{sec:kas_small_set}

We finally give generators for $\SL(3s;\FF_2)^{\times K^5}$ where $K = 2^{3s} - 1$.
We are actually going to give generators for a larger group~$\SL(3s;\FF_2)^{\times 2^{15s}}$.

Combining~\cref{eq:ChainOfEL} and~\cref{lem:EL3Generators} we see that
$\EL(3s;\FF_2)^{\times 2^{k s^2}}$ is generated by $4(k+2)+6$ elements.
For $s \ge 7$ and $k = 3$, surely the $26$ generators suffice.
This is a choice in~\cite[Corollary~3.1]{Kassabov_2007_alt}.
In this section, we take $k = 1$ and $s \ge 15$, in which case $18$ generators suffice and it is easy to explain.
Generalization to any larger value of~$k$ will be obvious.

Let us spell out (some of) the generators for reversible circuit considerations below.
Recall the two matrices $A,B$ from~\cref{eq:matricesAB} with~$m$ there equal to~$s$ here.
\begin{align}
    \varphi_1( E_{1,2}(1) ) &= 
    \begin{pmatrix}
        \one_{s} & \one_s & 0 \\
        0 & \one_s & 0 \\
        0 & 0 & \one_s
    \end{pmatrix}, &
    \varphi_1(E_{1,2}( \byy_1 )) &=
    \begin{pmatrix}
        \one_{s} & (x_{1,a,b})_{a,b=1}^s & 0 \\
        0 & \one_s & 0 \\
        0 & 0 & \one_s
    \end{pmatrix},\\
    \varphi_1(E_{1,2}( \baa )) &=
    \begin{pmatrix}
        \one_{s} & A & 0 \\
        0 & \one_s & 0 \\
        0 & 0 & \one_s
    \end{pmatrix}, &
    \varphi_1(E_{1,2}( \bbb )) &=
    \begin{pmatrix}
        \one_{s} & B & 0 \\
        0 & \one_s & 0 \\
        0 & 0 & \one_s
    \end{pmatrix}. \nonumber
\end{align}
As we discussed in~\cref{rmk:ring_iso_2},
the homomorphism~$\varphi_2$ is nothing but ``evaluating'' these matrices,
viewed as functions of~$x_{1,a,b}$,
at each factor (one of $2^{15s}$ factors) indexed by a $15s$-bit string.
The generators of form~$E_{i,j}(1)$, $E_{i,j}(\baa)$, $E_{i,j}(\bbb)$ 
are ``constant functions'' of~$x_{1,a,b}$.
This means that these generators are uniform across all factors of~$\SL(3s;\FF_2)^{\times 2^{15s}}$:
\begin{align}\label{eq:TheSLLargePowerGenerators}
    (\varphi_2 \circ \varphi_1)( E_{1,2}(1) ) &= 
    \begin{pmatrix}
        \one_{s} & \one_s & 0 \\
        0 & \one_s & 0 \\
        0 & 0 & \one_s
    \end{pmatrix}^{\times 2^{15s}}, &
    (\varphi_2 \circ \varphi_1)(E_{1,2}( \baa )) &=
    \begin{pmatrix}
        \one_{s} & A & 0 \\
        0 & \one_s & 0 \\
        0 & 0 & \one_s
    \end{pmatrix}^{\times 2^{15s}},\\
    (\varphi_2 \circ \varphi_1)(E_{1,2}( \bbb )) &=
    \begin{pmatrix}
        \one_{s} & B & 0 \\
        0 & \one_s & 0 \\
        0 & 0 & \one_s
    \end{pmatrix}^{\times 2^{15s}}. & \nonumber
\end{align}
Here, the superscript $\times 2^{15s}$ on the matrices 
denotes a direct product of $2^{15s}$ copies, not $2^{15s}$-th power.
For $s>15$, there will be more ``variables''~$x_{1,a,b}$ than the bits
that index a factor of~$\SL(3s;\FF_2)^{\times 2^{15s}}$.
This just means that the extra variables, {\it i.e.}, whenever $a > 15$, are always set to zero. 
Since we take $k=1$, let us write $\byy$ instead of $\byy_1$ and $(x_{a,b})$ instead of $(x_{1,a,b})$. 
Then,
\begin{align}
    \label{eq:TheSLLargePowerGenerators_y}
    (\varphi_2 \circ \varphi_1)(E_{1,2}( \byy )) &= \bigtimes_{x \in \{0,1\}^{15s}} 
    \begin{pmatrix}
        \one_{s} & (x_{a,b})_{a,b} & 0 \\
        0 & \one_s & 0 \\
        0 & 0 & \one_s
    \end{pmatrix} 
\end{align}
where $x = (x_{1,1},x_{1,2},\ldots,x_{1,s},x_{2,1},\ldots, x_{15,s}) \in \{0,1\}^{15s}$.

\subsection{Circuits} \label{sec:kas_ckt}

In the opening of this section,
we have explained that we lay $K^6 = (2^{3s}-1)^6$ points in a six-dimensional hypercube $\bKK_s$,
on which $\Alt(\bKK_s)$ acts.
We are in fact interested in the permutation group on~$(K+1)^6 = 2^{18s}$ points,
so each permutation is a bijection from~$\FF_2^{18s}$ onto itself.
Here, we show circuits for specific bijections corresponding to the generators 
of~$\SL(3s;\FF_2)^{\times 2^{15 s}}$ above.

\subsubsection{Depth-$\cO(n)$ circuits}

Recall that we index each factor in $\SL(3s;\FF_2)^{\times 2^{15 s}}$ 
by $x=(x_{1,1},x_{1,2},\ldots,x_{1,s},x_{2,1},\ldots, x_{15,s})\in \{0,1\}^{15s}$. 
It suffices to construct the circuits for one embedding 
(of $\SL(3s;\FF_2)^{\times 2^{15 s}}$ into the permutation group on~$\FF_2^{18s}$),
say $\pi_1(\Sigma_0)$ where $\Sigma_0$ denotes the generating set 
constructed in the preceding subsection.
The other five embeddings are implemented simply by permuting the sets of $3s$ bits.
Given the embedding $\pi_1$, 
all $2^{15s}$ factors act on the same $3s$ bits, denoted as $y_1, y_2, \ldots, y_{3s}$, 
but the specific action of each factor is controlled by the other $15s$ bits $x \in \{0,1\}^{15s}$.
More specifically, if the control bits equal $x$, then we apply the factor indexed by $x$ to $y_1, \cdots, y_{3s}$.
Put differently, a bitstring $(y_1,\ldots,y_{3s})$ is one of the $K+1$ points on a ``line'' of the six-dimensional cube.
The $15s$ control bits~$x$ specify one of $(K+1)^5$ lines.

We first construct the circuits for $(\pi_1\circ\varphi_2\circ\varphi_1)(E_{i,j}(w))$ where $w=1,\baa,\bbb$. 
The controls for these generators are vacuous. 
For example, following~\cref{eq:TheSLLargePowerGenerators}, the action of $(\pi_1\circ\varphi_2\circ\varphi_1)(E_{2,1}(\baa))$ is to apply
\begin{align}
    \begin{pmatrix}
        \one_{s} & 0 & 0 \\
        A & \one_s & 0 \\
        0 & 0 & \one_s
    \end{pmatrix}
\end{align}
directly to bits $y_1, \cdots, y_{3s}$, independent of~$x$.
One CNOT corresponds to a unique elementary matrix of form~$E_{i,j}(1)$,
where, if the action is matrix-vector multiplication on column vectors,
$j$ is the control bit and $i$ is the target bit.
Our specific ``constant function'' generators in~\cref{eq:TheSLLargePowerGenerators}
are each a product of $E_{i,j}(1)$ 
where for each $j$ there is at most one $i$ that enters into the product,
and for each~$i$ there is at most one~$j$ that enters into the product.
This means that all the CNOTs that comprise a ``constant function'' generator
act on disjoint pairs of bits,
and hence all can be implemented parallel (as long as an arbitrary pair of bits may be acted upon by a CNOT).
So, in all-to-all connectivity of bits,
the depth of the CNOT network implementing any ``constant function'' 
generator of~\cref{eq:TheSLLargePowerGenerators} is~$1$. 
See~\cref{fig:circ_b_cnot,fig:circ_one_cnot,fig:circ_a_cnot} for some circuit diagrams.

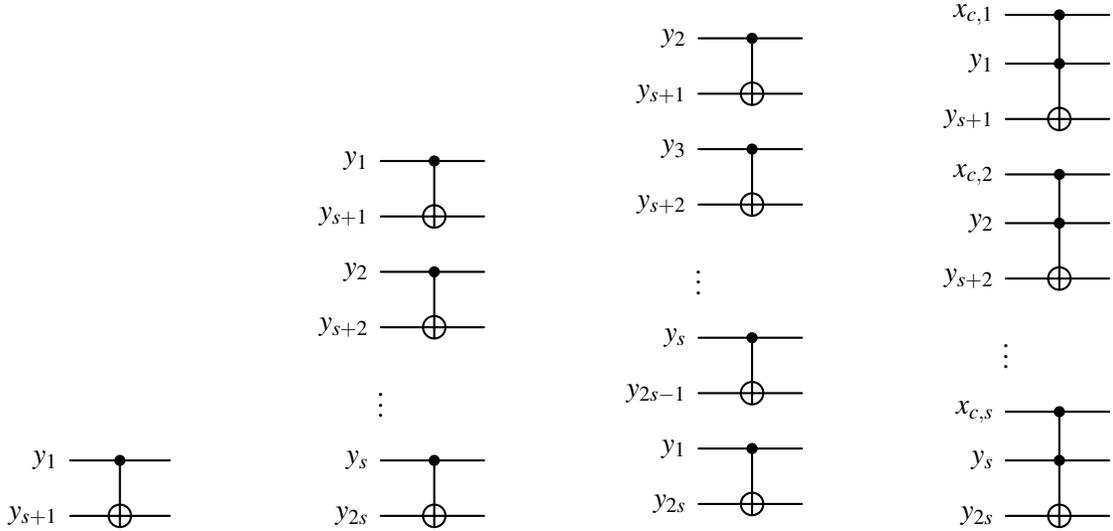
\begin{figure}[!ht]
    \centering
    \begin{subfigure}[b]{0.245\textwidth}
        \begin{equation*}
            \begin{quantikz}[slice style=blue] 
            \lstick{$y_1$} & \ctrl{1} & \\
            \lstick{$y_{s+1}$} & \targ{} &
            \end{quantikz}
        \end{equation*}
        \caption{$(\pi_1\circ\varphi_2\circ\varphi_1)(E_{2,1}(\bbb))$}
        \label{fig:circ_b_cnot}
    \end{subfigure}
    \hfill
    \begin{subfigure}[b]{0.245\textwidth}
        \begin{equation*}
            \begin{quantikz}[slice style=blue] 
            \lstick{$y_1$} & \ctrl{1} & \\
            \lstick{$y_{s+1}$} & \targ{} & \\
            \lstick{$y_2$} & \ctrl{1} & \\
            \lstick{$y_{s+2}$} & \targ{} & \\
            \vdots  \\
            \lstick{$y_{s}$} & \ctrl{1} & \\
            \lstick{$y_{2s}$} & \targ{} & 
            \end{quantikz}
        \end{equation*}
        \caption{$(\pi_1\circ\varphi_2\circ\varphi_1)(E_{2,1}(1))$}
        \label{fig:circ_one_cnot}
    \end{subfigure}
    \hfill
    \begin{subfigure}[b]{0.245\textwidth}
        \begin{equation*}
            \begin{quantikz}[slice style=blue] 
            \lstick{$y_2$} & \ctrl{1} & \\
            \lstick{$y_{s+1}$} & \targ{} & \\
            \lstick{$y_3$} & \ctrl{1} & \\
            \lstick{$y_{s+2}$} & \targ{} & \\
            \vdots  \\
            \lstick{$y_{s}$} & \ctrl{1} & \\
            \lstick{$y_{2s-1}$} & \targ{} & \\
            \lstick{$y_1$} & \ctrl{1} & \\
            \lstick{$y_{2s}$} & \targ{} &
            \end{quantikz}
        \end{equation*}
        \caption{$(\pi_1\circ\varphi_2\circ\varphi_1)(E_{2,1}(\baa))$}
        \label{fig:circ_a_cnot}
    \end{subfigure}
    \hfill
    \begin{subfigure}[b]{0.245\textwidth}
        \begin{equation*}
        \begin{quantikz}[slice style=blue] 
        \lstick{$x_{c,1}$} & \ctrl{1} & \\
        \lstick{$y_1$} & \ctrl{1} & \\
        \lstick{$y_{s+1}$} & \targ{} & \\
        \lstick{$x_{c,2}$} & \ctrl{1} & \\
        \lstick{$y_2$} & \ctrl{1} & \\
        \lstick{$y_{s+2}$} & \targ{} & \\
        \vdots  \\
        \lstick{$x_{c,s}$} & \ctrl{1} & \\
        \lstick{$y_{s}$} & \ctrl{1} & \\
        \lstick{$y_{2s}$} & \targ{} &
        \end{quantikz}
    \end{equation*}
    \caption{$(\pi_1\circ\varphi_2\circ\varphi_1)(E_{2,1}(\bzz_c))$}
    \label{fig:circ_depth_one}
    \end{subfigure}
    \caption{Depth-$1$ circuit implementations}
    \label{fig:kas_circuit_constant}
\end{figure}

Following~\cref{eq:TheSLLargePowerGenerators_y}, the action of $(\pi_1\circ\varphi_2\circ\varphi_1)(E_{2,1}(\byy))$ is to apply
\begin{align}
    \begin{pmatrix}
        \one_{s} & 0 & 0 \\
        (x_{a,b})_{a,b} & \one_s & 0 \\
        0 & 0 & \one_s
    \end{pmatrix}
\end{align}
to bits $y_1, \cdots, y_{3s}$ where $x = (x_{1,1},x_{1,2},\cdots,x_{1,s},x_{2,1},\cdots, x_{15,s})$ equal the values of the control bits. 
If $x_{a,b}=1$, the corresponding matrix demands the presence of a CNOT. 
If $x_{a,b} = 0$, the corresponding matrix demands the absence of a CNOT.
Hence, the implementation is a network of Toffoli gates (control-control-NOT). 
Concretely, for any $a\in \{1,\cdots, 15\}$, $b\in \{1,\cdots, s\}$, there is a Toffoli gate with the two controls at~$x_{a,b}$ and~$y_{b}$, 
and the NOT at~$y_{s+a}$.
In the circuit for each~$(\pi_1\circ\varphi_2\circ\varphi_1)(E_{i,j}(\byy))$, there are exactly~$15s$ Toffoli gates.
See~\cref{fig:kas_circuit_toffoli} for the circuit.

\begin{figure}[!ht]
    \begin{center}
    \begin{equation*}
        \begin{quantikz}[slice style=blue] 
        \lstick{$x_{1,1}$} &\ctrl{12}&& \ \ldots \ &&&&\ \ldots \ &&\\
        \lstick{$x_{1,2}$} &&\ctrl{11}& \ \ldots \ &&&&\ \ldots \ &&\\
        \vdots \\
        \lstick{$x_{1,s}$} &&& \ \ldots \ &\ctrl{9}&&&\ \ldots \ &&\\
        \lstick{$x_{2,1}$} &&& \ \ldots \ &&\ctrl{9}&&\ \ldots \ &&\\
        \lstick{$x_{2,2}$} &&& \ \ldots \ &&&\ctrl{8}&\ \ldots \ &&\\
        \vdots \\
        \lstick{$x_{15,s}$} &&& \ \ldots \ &&&&\ \ldots \ &\ctrl{8}&\\
        \lstick{$y_{1}$} &\control{}&& \ \ldots \ &&\control{}&&\ \ldots \ &&\\
        \lstick{$y_{2}$} &&\control{}& \ \ldots \ &&&\control{}&\ \ldots \ &&\\
        \vdots \\
        \lstick{$y_{s}$} &&& \ \ldots \ &\control{}&&&\ \ldots \ &\control{}&\\
        \lstick{$y_{s+1}$} &\targ{}&\targ{}& \ \ldots \ &\targ{}&&&\ \ldots \ &&\\
        \lstick{$y_{s+2}$} &&& \ \ldots \ &&\targ{}&\targ{}&\ \ldots \ &&\\
        \vdots \\
        \lstick{$y_{s+15}$} &&& \ \ldots \ &&&&\ \ldots \ &\targ{}&
        \end{quantikz}
    \end{equation*}
    \end{center}
    
    \caption{The circuit for $(\pi_1\circ\varphi_2\circ\varphi_1)(E_{2,1}(\byy))$ which consists of $15s$ Toffoli gates. }
    \label{fig:kas_circuit_toffoli}
\end{figure}
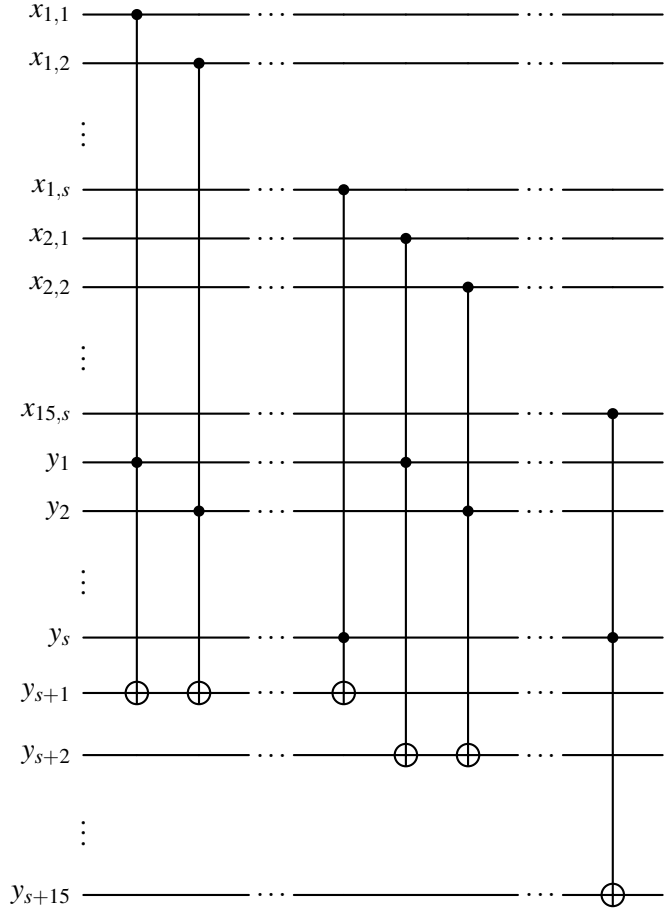

\subsubsection{Depth-$1$ circuits}
\label{sec:kas_depth_one}

Let us consider the generators described in~\cref{rmk:kas_gen_set_diag}. 
If we take $k=15$, then for any integer $s\geq 1$, the number of generators for $\EL(3s;\FF_2)^{\times 2^{15s}}$ is $4(15+2)+6 = 74$. 
Let us overload the notations $\varphi_1$ and $\varphi_2$ for the corresponding group homomorphisms. Then, $\varphi_1(\baa) = A$, $\varphi_1(\bbb) = B$, and $\varphi_1(\bzz_c) = z_c$ for $c=1, \cdots, 15$. 

The circuits for implementing $(\pi_1\circ \varphi_2\circ \varphi_1)(E_{2,1}(w))$ where $w=1, \baa, \bbb$ remain the same as in~\cref{sec:kas_small_set}, so each generator has circuit depth $1$. The action of $(\pi_1\circ \varphi_2\circ \varphi_1)(E_{2,1}(\bzz_c))$ is to apply
\begin{align}
    \begin{pmatrix}
        \one_{s} & 0 & 0 \\
        (x_{c,a,a})_{a=1}^s & \one_s & 0 \\
        0 & 0 & \one_s
    \end{pmatrix}
\end{align}
to bits $y_1, \cdots, y_{3s}$ where $x_{c,1,1},x_{c,2,2},\cdots, x_{c,s,s}$ take the values of the corresponding control bits $x_{c,1},x_{c,2},\cdots, x_{c,s}$.
Concretely, for each $a\in \{1,\cdots, s\}$, there is a Toffoli gate with the two controls at~$x_{c,a}$ and~$y_{a}$, 
and the NOT at~$y_{s+a}$.
Hence, for each $c\in \{1, \cdots, 15\}$, $(\pi_1\circ \varphi_2\circ \varphi_1)(E_{i,j}(\bzz_c))$ can be implemented using $s$ Toffoli gates which act on nonoverlapping triples of bits.
So, the depth of these $s$ Toffoli gates is~$1$; see~\cref{fig:circ_depth_one} for a circuit diagram.

\subsection{The action outside~\texorpdfstring{$\bKK_s$}{Kₛ}}
\label{sec:kas_exterior}

So far, we have focused on the ``interior'' of the hypercube $\bKK_s$, but the ``exterior'' of $\bKK_s$ is also important for our purpose for reversible circuits. Consider the decomposition of $\FF_2^{3s}$ into two components
\begin{align}
\bKK(0) := \{ 0^{3s} \}\quad \text{and}\quad \bKK(1) := \FF_2^{3s} \setminus \{ 0^{3s} \}.   
\end{align}
Then, we decompose $\FF_2^{18s}$ into $2^6$ components labeled by binary $6$-tuples $w = (w_1, w_2,\ldots,w_6) \in \FF_2^6$ as
\begin{align}
 \FF_2^{18s} = \left(\bKK(0)\sqcup \bKK(1)\right)^{6} = \bigsqcup_{w \in \FF_2^6} \bKK(w_1)\sqcup \cdots \sqcup \bKK(w_6)=:\bigsqcup_{w \in \FF_2^6}\bKK(w).
\end{align}
In fact, Kassabov's generating set $S$ have the property that each component $\bKK(w)$ is preserved,
\begin{align}
    S \bKK(w) =\bKK(w)\quad \text{for each}\quad w \in \FF_2^6.
\end{align}
That is, the permutation group $\cP$ generated by the circuits 
is a subgroup of $\prod_{w \in \FF_2^6} \Alt(\bKK(w)) \cong \bigtimes_{w \in \FF_2^6} \Alt(\bKK(w))$.
In this decomposition, indeed, the set $\mathbf K_s$ featured earlier is equal to $\bKK(1^6)$. The goal of this section is to study the effect of the circuit on other components with $w \ne 1^6$ (which was not the main focus of~\cite{Kassabov_2007_alt}). We have shown that the canonical map $\cP \to \Alt(\bKK(1^6))$ is surjective
by using Pyber's result~\cite{Pyber1993}. 

\begin{proposition}[A good Kazhdan constant for $\Alt(\bKK(1^6))$]\label{prop:KassabovCircuit}
For $s \ge 15$,
there exist a generating set $S$ for a subgroup~$\cP$ of $\prod_{w \in \FF_2^6} \Alt(\bKK(w))$,
consisting of $\abs{S} = \cO(1)$ elements, 
each of which can be written as a product of $\cO(n)$ CNOT and Toffoli gates.
The canonical projection $\xi : \cP \to \Alt(\bKK(1^6))$ is surjective.
The Kazhdan constant $\cK(\Alt(\bKK(1^6)),\xi(S))$ is uniformly bounded away from zero.
\end{proposition}

However, even though Kassabov's generating set works well for $\Alt(\bKK(1^6))$, we wanted a good Kazhdan constant for the full set of bit strings; we must introduce extra generators. So far, we have paid little attention to how the circuit acts on the other components $\bKK(w)$. Without repeating a proof of Kazhdan constants from scratch, our strategy is to make use of the good Kazhdan constant of $\Alt(\bKK(1^6))$ by introducing extra generators that mix different components $\bKK(w)$. Technically, to use the short product lemma (\cref{lem:shortproduct}), the nontrivial action on the components $\bKK(w)$ appears to cause some technical issues. Fortunately, we find a way to avoid a detailed foray into the particular nontrivial actions by introducing extra multiply control-NOT to trivialize the action on exterior points.

\subsubsection{Trivializing the action on the exterior points}

We modify the circuit of the preceding subsections and trivialize its action on $\bigcup_{w \neq 1^6} \bKK(w)$,
so that not only do we have a surjection $\cP \to \Alt(\bKK(1^6))$ but also an equality $\cP' = \Alt(\bKK(1^6))$
where $\cP'$ is a permutation group generated by the modified circuits.

It would be instructive to consider two- and three-dimensional cases first,
instead of the six dimensions.
Generalization to the six (or any constant) dimensions will be straightforward.

\begin{figure}[!ht]
    \centering
    \includegraphics[width=\textwidth, trim={0mm 120mm 95mm 0mm}, clip]{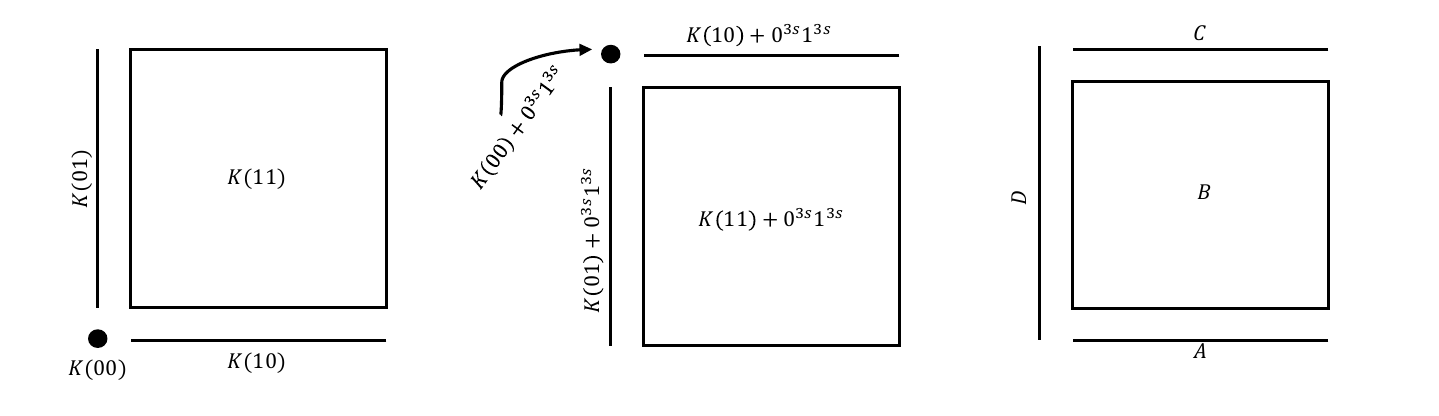}
    \caption{
        For illustration, we regard $\FF_2^{3s}$ as a set of integers in $[0,2^{3s}-1]$.
        Then, $(\FF_2^{3s})^2$ is a square of integral points in a two-dimensional plane.
        This square is divided into four regions $\bKK(00),\bKK(01),\bKK(10),\bKK(11)$ as depicted in the leftmost figure.
        The group $\cP$ generated by the circuits for Kassabov's generators permutes points within each region.
        The group $\cP'$ generated by the modified circuits have the same action on $\bKK(11)$ as $\cP$
        but fixes all the points on $\bKK(00),\bKK(10),\bKK(01)$.
        By conjugating the circuits by $X(01)$ we have a permutation group that preserves each region of the middle figure.
        We merge them to generate a bigger group $\Alt(A\cup B \cup C)$.
    }
    \label{fig:alt-merger}
\end{figure}

In the two-dimensional case 
we have a square of points labeled by~$(u_1,u_2) \in (\FF_2^{3s})^{\times 2}$.
See \cref{fig:alt-merger}.
Under~$\pi_1(\Gamma)$, the first coordinate is acted upon by a product of the matrix groups,
and the undesired action is on the line of points~$(u_1, 0)$.
Each ``line'' of $K+1$ points is horizontal in \cref{fig:alt-merger}.
The ``constant function'' generators 
$\varphi_1(E_{i,j}(1))$, $\varphi_1(E_{i,j}(\baa))$ and $\varphi_1(E_{i,j}(\bbb))$
act nontrivially on each ``line.''
Since every generator has order~$2$,
the trivialization is achieved by an extra layer that applies the element of~$\SL(3s;\FF_2)$ on
the $u_1$ register (consisting of $3s$ bits)
if and only if $u_2$ register is zero.
That is, the goal of our modification is to trivialize the action on $\bKK(10)$ under $\pi_1(\Gamma)$ 
(and on $\bKK(01)$ under $\pi_2(\Gamma)$).
A CNOT conditioned on some set of bits being all zero
is implemented by applying a bit flip~$X$ on all the controlling bits,
applying multiply controlled NOT (that is activated if and only if the controlling bits are all~$1$),
and applying the bit flip again.
The ``nonconstant function'' $\varphi_1(E_{i,j}(\byy_c))$ will be automatically trivial on this line in the boundary
because the control bits in~$u_2$ for Toffoli gates are all zero.
It is clear that the action of this modified circuit on $\bKK(11)$
remains equal to the original.

In the three-dimensional case
we have a cube of points labeled by~$(u_1,u_2,u_3)$.
Under~$\pi_1(\Gamma)$, 
the first coordinate is acted upon by a product of the matrix groups.
The ``exterior'' of concern consists of two faces $(u_1, 0, u_3)$ and $(u_1, u_2, 0)$.
These two faces intersect along~$(u_1,0,0)$, 
which is a new feature we encounter in $3$ dimensions but not in $2$ dimensions.
For the first face~$(u_1,0,u_3)$, 
we amend the circuit by adding another layer~$A$ of gates
that applies a circuit~$C$ controlled by all bits in the second register~$u_2 = 0$ being zero.
(There may be already a Toffoli gate in~$C$ that reads a bit in the second register,
in which case the count of the controlling bits in the gate of~$A$ 
corresponding to this Toffoli gate is one less 
than that for a Toffoli gate in~$C$ that does not touch any bit from the second register.)
We add another layer for the second face~$(u_1,u_2,0)$.
By inclusion-exclusion,
the action on the intersection $(u_1,0,0)$ of the two faces will not have been trivialized,
so we need a third layer that is controlled on all bits in the second and third registers.

It should now be clear how to proceed all the way up to the six dimensions (or any constant $d$ dimensions).
The exterior $\bKK(w \neq 1^6)$ will consist of $(d-1)$-dimensional hypercubes,
which intersect along lower dimensional hypercubes.
There must be layers of amendments for all intersections.
The number of additional layers of multiply controlled NOT gates
will be $\binom{d-1}{1} + \binom{d-1}{2} + \cdots + \binom{d-1}{d-2}$,
a constant.
Each amendment layer has $\CO(s)$ Toffoli, CNOT, or X gates,
all controlled by at most $15 s$ bits.

Our next task is to express the amendment layer by elementary gates ($X$, CNOT, and Toffoli).
To this end, consider the NOT gate controlled by (i.e., activated on) $k$ bits all being~$1$.
\begin{align}
C^k X := \ket{1^k}\bra{1^k} \otimes X +   ( 1- \ket{1^k}\bra{1^k}) \otimes I,
\end{align}
which is a reversible transformation on $k+1$ bits.

\begin{lemma}[\cite{GidneyPost} and~\cref{fig:gidney_combo}]\label{lem:gidney}
    For $m \ge 2$, the reversible transformation $(C^m X)\otimes \one$ on $m+2$ bits
    is a composition of $\cO(m)$ Toffoli gates.
    Every bit is acted upon by $\cO(1)$ different Toffoli gates of the circuit.
\end{lemma}

\begin{figure}[!ht]
    \centering
    \begin{subfigure}{\textwidth}
    \begin{equation*}
        \begin{quantikz}[slice style=blue] 
        \lstick{$A$}&\ctrl{8}&&\ctrl{8}&&\rstick{$A$}\\
        \lstick{$B$}&\control{}&&\control{}&&\rstick{$B$}\\
        \lstick{$C$}&\control{}&&\control{}&&\rstick{$C$}\\
        \lstick{$D$}&\control{}&&\control{}&&\rstick{$D$}\\
        \lstick{$E$}&&\ctrl{4}&&\ctrl{4}&\rstick{$E$}\\
        \lstick{$F$}&&\control{}&&\control{}&\rstick{$F$}\\
        \lstick{$G$}&&\control{}&&\control{}&\rstick{$G$}\\
        \lstick{$T$}&&\targ{}&&\targ{}&\rstick{$T+ABCDEFG$}\\
        \lstick{$a$}&\targ{}&\control{}&\targ{}&\control{}&\rstick{$a$}\\
        \end{quantikz}
    \end{equation*}
    \vspace{-.6cm}
    \caption{$C^m X$ using one ancilla bit~$a$ and $\cO(1)$ $C^{\lceil m/2\rceil} X$ and $C^{\lfloor m/2\rfloor} X$ gates~\cite[Theorem 23]{AaronsonGrierSchaeffer2015}.}
    \vspace{.7cm}
    \label{fig:gidney_multi_control_constant}
    \end{subfigure}
    \hfill
    \begin{subfigure}{\textwidth}
        \begin{equation*}
        \begin{quantikz}[slice style=blue] 
        \lstick{$A$}&&&&\ctrl{2}&&&&&&\ctrl{2}&&&\rstick{$A$}\\
        \lstick{$B$}&&&&\control{}&&&&&&\control{}&&&\rstick{$B$}\\
        \lstick{$a_1$}&&&\ctrl{2}&\targ{}&\ctrl{2}&&&&\ctrl{2}&\targ{}&\ctrl{2}&&\rstick{$a_1$}\\
        \lstick{$C$}&&&\control{&}&&\control{}&&&&\control{}&&\control{}&&\rstick{$C$}\\
        \lstick{$a_2$}&&\ctrl{2}&\targ{}&&\targ{}&\ctrl{2}&&\ctrl{2}&\targ{}&&\targ{}&\ctrl{2}&\rstick{$a_2$}\\
        \lstick{$D$}&&\control{}&&&&\control{}&&\control{}&&&&\control{}&\rstick{$D$}\\
        \lstick{$a_3$}&\ctrl{2}&\targ{}&&&&\targ{}&\ctrl{2}&\targ{}&&&&\targ{}&\rstick{$a_3$}\\
        \lstick{$E$}&\control{}&&&&&&\control{}&&&&&&\rstick{$E$}\\
        \lstick{$T$}&\targ{}&&&&&&\targ{}&&&&&&\rstick{$T+ABCDE$}\\
        \end{quantikz}
    \end{equation*}
    \vspace{-.6cm}
    \caption{$C^m X$ using $m-2$ ancillas and $\cO(m)$ Toffoli gates.}
    \label{fig:gidney_multi_control_linear}
    \end{subfigure}
    \caption{\cite{GidneyPost} shows an implementation of $(C^m X)\otimes \one$ using $\cO(m)$ Toffoli gates by combining~\cref{fig:gidney_multi_control_constant,fig:gidney_multi_control_linear}.}
    \label{fig:gidney_combo}
\end{figure}
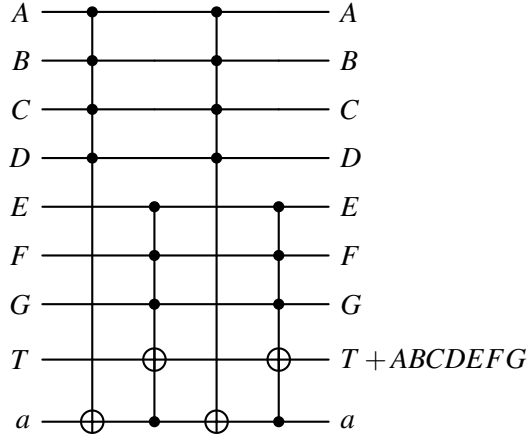
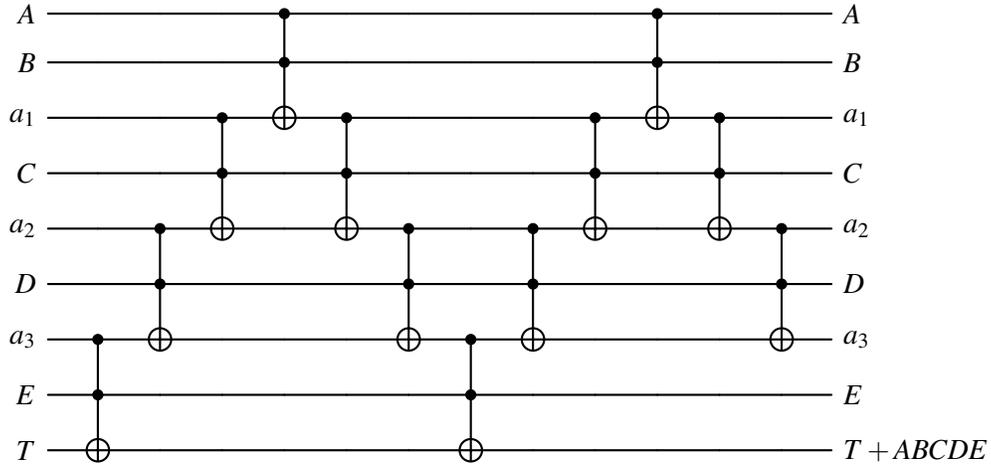

The use of one ancilla is necessary because on $m+1$ bits $C^m X$ is an odd permutation
while any Toffoli gate on $4$ or more bits is an even permutation.
There are other versions of such a construction. 
See~\cite{AaronsonGrierSchaeffer2015} and references therein.
\footnote{Note also that $C^m X$ can be implemented by $\cO(m^2)$ 1- and 2-qubit \textit{unitary} gates 
without any ancillas~\cite[\S8.1.3]{KSV}, but here we are only allowing reversible \textit{classical} circuits.}

For any operator $W$ on $f$ bits, let us write $C^m W$ to denote an operator on $m+f$ bits
which applies $W$ if and only if all the $m$ control bits are $1$.

\begin{lemma}\label{lem:CmXell}
    $C^m(X^{\otimes \ell})$  with $\ell \ge 2$
    is a composition of $\cO(m)$ Toffoli gates and $\cO(\ell)$ CNOT gates,
    without any ancilla bit.
    Every bit in this circuit is acted upon by $\cO(1)$ Toffoli and CNOT gates.
\end{lemma}

\begin{proof}
    Since $\ell \ge 2$, we can write $C^m (X^{\otimes \ell})$ 
    as a product of two operators to spare an ancilla for each,
    so we may assume that we have an ancilla.
    The operator $C^m (X^{\otimes \ell})$ is conjugate to $C^m X$ by a circuit of CNOT
    that toggles between $X$ and $X^{\otimes \ell}$.
    There are many such CNOT circuits, but an easy version has a CNOT in a staircase fashion,
    and has $\ell-1$ CNOTs.
    Hence, in total we need $\cO(m)$ Toffoli gates for $C^m X$ and, in addition,
    $2\ell -2$ CNOTs for $C^m(X^{\otimes \ell})$.
\end{proof}

We arrive at the promised circuit that has trivial action outside $\bKK(1^6)$.
A feature is that the gate complexity is only a constant multiple of the uncleaned version.

\begin{proposition}[A modified generating set with trivial action outside $\bKK(1^6)$]\label{prop:KassabovGenCleanedUp}
    For $s \ge 15$,
    the group~$\Alt(\bKK(1^6))$ is generated by elements of $S'$
    where $\abs{S'} = \cO(1)$
    and each member of $S'$ is a circuit of $\cO(n)$ Toffoli and CNOT gates,
    in which every bit is acted upon by $\cO(1)$ Toffoli and CNOT gates.
\end{proposition}

\begin{proof}
    Recall from \cref{prop:KassabovCircuit} 
    that each element~$W$ of $S$ is a product of $\cO(n)$ Toffoli and CNOT gates.
    As we have explained above, each amendment layer for $W$ is $W$ controlled by at most $15s$ bits,
    {\it i.e.}, $C^m W$ for some $m \le 15 s$.
    Since CNOT and Toffoli gates have order 2, it follows that $W^2 = I$.
    
    We use the identity $(CU) (XI) (CU) (XI) = IU$ where $CU$ is the controlled $U$.
    If the $X$ in this identity is controlled by $m$ control bits and $U^2 = I$,
    then $U$ is implemented if and only if all the $m$ control bits are~$1$.
    This generalizes to many nonoverlapping $U$'s in the following way.
    Let us use subscripts to denote bit labels.
    If $U^{(i)}$ are gates in $W$, we have
    \begin{align}
        \parens*{ \prod_{i}(CU^{(i)})_{2i-1,2i} }
        \parens*{ C^m \prod_{i}(XI)_{2i-1,2i} }
        \parens*{ \prod_{i}(CU^{(i)})_{2i-1,2i} }
        \parens*{ C^m \prod_{i}(XI)_{2i-1,2i} }
        =
        C^m \prod_i U^{(i)}_{2i}  = C^m W \, .
    \end{align}
    \Cref{lem:CmXell} says that the second and fourth factors have gate complexity $\cO(n)$.
    The controlled version of individual gates in $W$ can be further decomposed into Toffoli gates by~\cref{lem:gidney}.
    The number of ancillas in this construction is proportional to the number of gates in $W$.
    It is easy to find a sufficient number of ancillas because we can always work with, 
    say, one-tenth of the gates of $W$ at a time.
\end{proof}

Note that the depth of each member of $S'$ is dominated by the depth of $C^m X$.
The circuit of CNOTs that turns~$X$ to $X^{\otimes \ell}$ can be made to have depth $\cO(\log \ell)$.

\subsection{Bounded generation of~\texorpdfstring{$\Alt(2^n)$}{Alt(2ⁿ)} from~\texorpdfstring{$\Alt(\mathbf K_s)$}{Alt(Kₛ)}}

Here, we show how to obtain a constant-sized generating set for $\Alt(2^{18s})$
from that on the subset $\Alt(\mathbf K_s)$ (with trivial action on other components)
by enlarging the generating set with bit flips. Another use of short product lemma (\cref{lem:shortproduct}) then give general $2^n$ from $2^{18s}$.

We will be handling many subgroups below, 
and we will have to distinguish an abstract group (an isomorphism class) from a concrete group.
If $N$ is a natural number, $\Alt(N)$ denotes an abstract alternating group of order $N!/2$.
If $A$ is a finite set of cardinality $\abs A = N$,
then $\Sym(A)$ is the group of all bijective functions from~$A$ to~$A$,
which is isomorphic to, but not the same as, $\Sym(N)$. 
$\Alt(A)$ is then a subgroup of $\Sym(A)$, defined to be the image of $\Alt(N)$ under $\Sym(N) \xrightarrow{\cong} \Sym(A)$.
If $A \subseteq B$ are finite sets, then $\Alt(A)$ is a subgroup of $\Alt(B)$.
If two subgroups $H_1, H_2$ of a group $G$ are element-wise commuting, 
then the product $H_1 H_2 := \{h_1 h_2 \in G \,|\, h_i \in H_i \}$ is a subgroup of~$G$,
and $H_1 H_2$ is isomorphic to, but not the same as, the Cartesian product group $H_1 \times H_2$.
Whenever we speak of disjoint finite sets, we regard that they are subsets of the union of them.

\subsubsection{Conjugations and commutators in permutation groups}
This section introduces ways to generate permutation groups by its subgroups that will be useful for short product lemma (\cref{lem:shortproduct}).
\begin{lemma}[Generating $\Alt(N)$ by commutators]\label{lem:threecommutators}
    If $N \ge 5$, every member of $\Alt(N)$ is a product of $3$ commutators.
\end{lemma}

We use the standard cycle notation for permutations. 
For example, $(1,2,3)$ means $1 \mapsto 2 \mapsto 3 \mapsto 1$.
The product of permutation follows the rule of function composition.
So, $(1,2)(2,3) = (1,2,3)$.
By convention, every cycle has length $\ge 2$ and we do not call the identity permutation a cycle.

\begin{proof}
    We begin with the case where $\pi\in\Alt(N)$ is a product of disjoint odd-length cycles whose lengths' sum is $\le N - 2$.
    Recall that the conjugacy class of $\Sym(N)$ is determined by the \emph{cycle shape} of $\pi$.
    (The cycle shape is the lengths of cycles in its disjoint cycle form.)
    Conjugation is nothing but relabeling the ``letters'' that $\Sym(N)$ permutes.
    Hence, there exists $b \in \Sym(N)$ such that $b \pi b^{-1} = \pi^{-1}$.
    Since $\pi$ has total length $\le N-2$, we can always ensure $b$ is even.
    Then, $b \pi b^{-1} \pi^{-1} = \pi^{-2}$ is a commutator of~$\Alt(N)$.
    Since a square root exists for any odd-length cycle, 
    we conclude that $\pi$ is a commutator.
    
    An arbitrary member~$\pi \in \Alt(N)$ is a product $\pi_1 \pi_2 = \pi_2 \pi_1$ 
    where $\pi_1$ consists of disjoint odd-length cycles 
    and $\pi_2$ of disjoint even-length cycles. 

    (i) Suppose that $\pi_2$ is nonidentity.
    Then the length requirement ``$\le N-2$'' is satisfied for~$\pi_1$,
    and hence $\pi_1$ is a commutator.
    A product $(1,2,\ldots,k)(k+1,k+2,\ldots,m)$ of two disjoint even-length cycles (both $k,m$ are even),
    is equal to $(1,2,\ldots,k,k+1)(k,k+1,\ldots,m-1,m)$, which is a product of two odd-length cycles.
    Generally, $\pi_2$ is a product of $2$ nonidentity permutations, 
    each of which consists of disjoint odd-length cycles.
    These odd-length cycles satisfy the length requirement ``$\le N -2$.''
    Therefore, $\pi$ is a product of $3$ commutators.

    (ii) Suppose that $\pi_2$ is identity and $\pi_1$ is not a cycle.
    Then, $\pi = \pi_1$ must be a product of two nonidentity permutations,
    each consisting of disjoint odd-length cycles and satisfies the length requirement.
    Therefore, $\pi_1$ is a product of $2$ commutators.
    
    (iii) Suppose $\pi_2$ is identity and $\pi_1$ is a cycle, which must be odd-length.
    If its length is $\le N -2$, then $\pi = \pi_1$ is a commutator.
    If not, we can break it into $2$ or $3$ odd-length cycles (not disjoint),
    each of which has length $\le N-2$,
    so $\pi$ is a product of $3$ commutators.
\end{proof}

Next, we also show that $\Alt(N)$ can be generated by conjugating a smaller alternating group with the full group.

\begin{lemma}[Conjugating a smaller alternating group]\label{lem:fourpqpinv}
    Let $A$ and $B$ be disjoint finite sets such that $\abs{B} \ge \abs{A} + 6 \ge 8$.
    Every element of~$\Alt(A\cup B)$ is a product of $4$ elements, 
    each of which takes the form $PQP^{-1}$ for some $P \in \Alt(A \cup B)$ and $Q \in \Alt(B)$.
\end{lemma}

The upper bound $4$ depends on the size of $B$ relative to $A$.
A similar statement is true if $\abs B$ is at least a constant fraction of $\abs A$ in which case the bound will be larger than $4$.
For an extreme case, if $\abs{B} = 3 \le \abs{A} - 2$,
then an element of form $PQP^{-1}$ can be any cycle of length~$3$,
and $\Alt(A \cup B)$ is generated by $3$-cycles,
but we need $\Theta(\abs{A})$ $3$-cycles to express a cycle of length~$\abs{A}$.

\begin{proof}[Proof of~\cref{lem:fourpqpinv}]
    A nontrivial element of~$\Alt(A \cup B)$
    is a product of disjoint cycles of lengths $\ell_1 \le \ell_2 \le \cdots \le \ell_m$,
    where $1 \le m \le \frac 1 2 \abs{A \cup B}$ and $2 \le \ell_1 \le \ell_m \le \abs{A \cup B}$.
    A permutation $(1,2)(3,4,5)(7,8,9,10)(11,12,13,14,15) \in \Alt(15)$, for example,
    has $m = 4$ and the lengths are $2,3,4,5$.
    Let $q'_3$ be the last cycle of length $\ell_m$,
    and let $q_1$ and $q_2$ be the products of all the odd- and even-indexed cycles except the last, respectively.
    In the example, $q_1 = (1,2)(7,8,9,10)$, $q_2=(3,4,5)$, and $q'_3 = (11,12,13,14,15)$.
    The product $q_1 q_2 q'_3$ is the original.
    The sum of the lengths of the cycles in $q_1$ is at most $\frac 1 2 \abs{A \cup B}$,
    and so is the sum of those in~$q_2$.
    By assumption, these sums are each at most $\abs{B} - 3$.
    For the last cycle, we write it as a product of two cycles $q_3$ and $q_4$ whose lengths differ by at most~$1$.
    In the example, $q_3 = (11,12,13)$ and $q_4 = (13,14,15)$.
    Then, both $q_3$ and $q_4$ have length $ \le \frac 1 2 \abs{A \cup B} +1 \le \abs{B}-2$.
    Zero, two, or four of $q_1,q_2,q_3,q_4$ may be odd permutations,
    in which case we multiply them by transpositions so that 
    each of them is an even permutation of total length at most $\abs{B}$,
    while the product of the four remains unchanged.
    Therefore, for $i = 1,2,3,4$, 
    there exists $P_i \in \Sym(A\cup B)$ and $Q_i \in \Alt(B)$ such that $P_i Q_i P_i^{-1} = q_i$.
    Since $\abs{A} \ge 2$, we may multiply $P_i$ by a transposition disjoint from $Q_i$ to ensure $P_i \in \Alt(A\cup B)$
    while not changing $P_i Q_i P_i^{-1}$.
\end{proof}

Finally, here is a way to combine two subgroups to generate a larger group.

\begin{lemma}[Two subgroups]\label{lem:GeneratingLargerSymmetricGroupFromSmallerOnes}
    Let $A,B,C$ be disjoint finite sets such that $\abs{B} \ge \abs{C}$.
    The permutation groups~$\Sym(A \cup B)$ and~$\Sym(B \cup C)$ are subgroups of~$\Sym(A \cup B \cup C)$.
    Then, every element~$p \in \Sym(A \cup B \cup C)$ is a product $p = \ell \cdot r \cdot \ell'$
    for some $\ell,\ell \in \Sym(A \cup B)$ and $r \in \Sym(B \cup C)$.
    Similarly, if $\abs{B} \ge 2$, every element of~$\Alt(A \cup B \cup C)$ 
    is a product of three elements from~$\Alt(A \cup B) \cup \Alt(B \cup C)$.
\end{lemma}

For example, if $0 < \abs A < \abs B \le 3\abs A/2$,
there exist two embeddings of~$\Sym(A)$ into~$\Sym(B)$ 
such that every element of~$\Sym(B)$ 
is a product of at most three elements of the union of the two embedded subgroups.

An argument for a similar purpose is found in~\cite[\S5]{Kassabov_2007_alt},
but \Cref{lem:GeneratingLargerSymmetricGroupFromSmallerOnes}
is our own alternative.

\begin{proof}
    Since $\abs{C} \le \abs{B}$,
    there exists a permutation~$\ell^{-1} \in \Sym(A \cup B)$ 
    such that the images~$(\ell^{-1} \cdot p)(c)$ for all~$c \in C$ land in~$B \cup C$.
    Then, there exists a permutation~$r^{-1} \in \Sym(B \cup C)$ such that 
    $(r^{-1} \cdot \ell^{-1} \cdot p)(c) = c$ for all~$c \in C$.
    Hence, the permutation $r^{-1} \cdot \ell^{-1} \cdot p$ is equal to some~$\ell' \in \Sym(A \cup B)$.

    The proof for alternating groups is similar.
    The only modification is to make sure that $\ell,r$ are even permutations.
    Since $\abs{B} \ge 2$, we can multiply $\ell,r$ by a transposition of~$\Sym(B)$ if needed.
\end{proof}

\subsubsection{From $(2^{3s}-1)^6$ to $2^{18s}$}

Observe that $\Alt(\bKK_s) = \Alt(\bKK(1^6))$ can be embedded in $2^6$ ways into $\Alt(\FF_2^{18s})$
by toggling any of six registers (each of $3s$ bits)
by flipping $3s$ bits (X gates) there.
More formally,
for each $u = (u_1,\ldots,u_6) \in \FF_2^6$, 
we define a bit flip operator $X(u)$ acting on all $3s$ bits in each register $i$ if $u_i = 1$:
\begin{align}
    X(u) := (X^{\otimes 3s})^{ u_1}\otimes \cdots \otimes(X^{\otimes 3s})^{ u_6}.
\end{align}
For example, $X(010100)$ is an operator flipping $6s$ bits in the second and fourth registers
and $X(111111)$ flips all $18s$ bits.
If we conjugate the circuits for $S'$ of \cref{prop:KassabovGenCleanedUp} by $X(u)$,
then we obtain a group~$\cP'(u) = X(u) \cP' X(u) \subset \Alt(\FF_2^{18s})$ where $\cP' = \Alt(\bKK(1^6))$.

\begin{proposition}\label{prop:AltTwoToTheEighteenEssByCircuit}
    For $s \ge 15$,
    every element of $\Alt(\FF_2^{18s})$ is a product of $3^6$ elements of $\bigcup_{u \in \FF_2^6} \cP'(u)$.
\end{proposition}

\begin{proof}
    As the six dimensions might obscure the idea,
    we work out one- and two-dimensional analogs first,
    after which the proof will be clear. 
    Consider $A \sqcup B \sqcup C = \FF_2^{3s}$ where
    \begin{align}
            A &= \bKK(0) = \{0^{3s}\}\nonumber\\
            B &= \bKK(1) \cap (\bKK(1) + 1^{3s}) = \FF_2^{3s} \setminus \{0^{3s},1^{3s}\} \\
            C &= A + 1^{3s} = \{1^{3s}\}.\nonumber
    \end{align}
        Observe that $B + 1^{3s} = B$.
        We have $\cP = \cP(0) = \Alt(B \cup C)$
        and $\cP(1) = \Alt(A \cup B)$.
        \cref{lem:GeneratingLargerSymmetricGroupFromSmallerOnes}
        says that every element of $\Alt(A \cup B \cup C)$ is a product of $3$ elements of $\Alt(A \cup B) \cup \Alt(B \cup C)$.
    
    In the two-dimensional case, 
    we have $\bKK(00), \bKK(01), \bKK(10), \bKK(11)$ that are each preserved by $\cP(00)$.
    Consider $(\FF_2^{3s})^2 = A \sqcup B \sqcup C \sqcup D$ where
    \begin{align}
            A &= \bKK(10), \nonumber \\
        B &= \bKK(11) \cap (\bKK(11) + 0^{3s} 1^{3s}),\\
        C &= A + 0^{3s}1^{3s},\nonumber\\
        D &= \bKK(00) \cup \bKK(01). \nonumber
    \end{align}
    See \cref{fig:alt-merger}.
    It is routine to check that $\cP(01) = \Alt(AB)$ and preserves each of $C$ and $D$.
    (We omitted $\cup$ to write $AB$ instead of $A \cup B$.)
    The union $\cP(00) \cup \cP(01)$ generates a group $\cP'(0) = \Alt(ABC)$.
    Applying \cref{lem:GeneratingLargerSymmetricGroupFromSmallerOnes} while ignoring~$D$,
    we see that every element of $\Alt(ABC) = \Alt(\bKK(10) \cup \bKK(11))$
    is a product of $3$ elements of $\cP(00) \cup \cP(01)$.
    By a parallel argument, we see that 
    $\langle \cP(10) \cup \cP(11) \rangle = \Alt(ABC + 1^{3s}0^{3s})$,
    and every element of the latter is a product of $3$ elements of $\cP(10) \cup \cP(11)$.
    Now, we are back to the one-dimensional case with $\cP'(0)$ and $\cP'(1)$.
    We conclude that every element of $\Alt(ABCD)$ is a product of $3^2$ elements 
    of $\cP(00) \cup \cP(10) \cup \cP(01) \cup \cP(11)$.

    To get to the six-dimensional case, 
    we proceed inductively by including one $\bKK(0)=\{0^{3s}\}$ at a time, 
    proceeding from the right most register.
    Schematically,
\begin{align}
    \bKK(1)^6 
        &\to \bKK(1)^5 \times (\bKK(1)\sqcup \bKK(0)) &&=  \bKK(1)^5 \times \FF_2^{3s} \nonumber\\
        &\to \bKK(1)^4 \times (\bKK(1)\sqcup \bKK(0)) \times \FF_2^{3s} &&= \bKK(1)^4\times \FF_2^{6s}\\
        &\cdots \nonumber\\
        &\to (\bKK(1)\sqcup \bKK(0))\times \FF_2^{15s} &&= \FF_2^{18s}.\nonumber
\end{align}
    Since we invoke \cref{lem:GeneratingLargerSymmetricGroupFromSmallerOnes} six times recursively,
    every element of~$\Alt(\FF_2^{18s})$ is a product of $3^6 = \cO(1)$ elements 
    of $\bigcup_{u \in \FF_2^6} X(u) \cP X(u)$.
\end{proof}

\subsubsection{From $2^{18s}$ to $2^n$}

Now, we show how to relax the number of qubits to be a nonmultiple of 18.
In particular, we show how to generate $\Alt(\FF_2^n)$ from two alternating subgroups that act on $n-1$ bits, 
leaving the remaining $1$ bit intact.
To this end, we need to clarify three embeddings of an abstract group $\Alt(2^{n-1})$ 
into a concrete group $\Alt(\FF_2^n)$.
The first is by an obvious bijection $\FF_2^{n-1} \to \FF_2^{n-1} \times \{0\}$
giving a group isomorphism $\phi_0 : \Alt(2^{n-1}) \xrightarrow{~\cong~} \Alt(\FF_2^{n-1} \times \{0\})$.
Of course, we also have the second embedding $\phi_1 : \Alt(2^{n-1}) \xrightarrow{~\cong~} \Alt(\FF_2^{n-1} \times \{1\})$.
The third way is to let $\Alt(2^{n-1})$ act on $n$-bit strings such that it leaves the last bit intact.
This third way amounts to the ``diagonal'' subgroup 
\begin{align}
 \Alt(\FF_2^{n-1}) \otimes \one 
 = \big\{ \phi_0(x) \phi_1(x) \in \Alt(\FF_2^{n-1} \times \{0\}) \Alt(\FF_2^{n-1} \times \{1\}) ~\big|~ x \in \Alt(2^{n-1}) \big\}
\end{align}
where we have used the tensor product symbol\footnote{This is just a notation; we do not define a tensor operation for groups.}
since we regard $\FF_2^n$ as a basis of $(\CC^2)^{\otimes n}$.
Similarly, we also think of an embedding of $\Alt(2^{n-1})$ which leaves the first bit intact:
$\one \otimes \Alt(\FF_2^{n-1})$.
In a similar manner, we will also consider $\one \otimes \Alt(\FF_2^{n-1}) \otimes \one \subset \Alt(\FF_2^{n+1})$.

\begin{proposition}\label{prop:increasingNinAltTwoToTheN}
    For $n \ge 5$,
    every element of $\Alt(\FF_2^{n+1})$ 
    is a product of $2592$ elements of $(\Alt(\FF_2^{n}) \otimes \one) \cup (\one \otimes \Alt(\FF_2^{n}) )$.
\end{proposition}

\begin{proof}
    In this proof, we call an element of form $g h g^{-1} \in \Alt(\FF_2^n)$ a \emph{primitive}
    if $g \in \Alt(\FF_2^n)$ and $h \in \Alt(\FF_2^{n-1}) \otimes \one$.
    A conjugate of a primitive is a primitive.
    
    For $x,y \in \Alt(2^{n-1})$,
    we have $\phi_0(y) \in \Alt(\FF_2^{n-1} \times \{ 0 \}) \Alt(\FF_2^{n-1} \times \{ 1 \})$
    and $\phi_0(x)\phi_1(x) \in \Alt(\FF_2^{n-1}) \otimes \one$.
    Then,
    \begin{align}
        \phi_0(y)\cdot (\phi_0(x)\phi_1(x))\cdot \phi_0(y)^{-1}\cdot (\phi_0(x)\phi_1(x))^{-1} = \phi_0(yxy^{-1}x^{-1}) \in \Alt(\FF_2^{n-1} \times \{0\} )
    \end{align}
    is a commutator and is a product of $2$ primitives.
    By \cref{lem:threecommutators}, every element of $\Alt(\FF_2^{n-1} \times \{0\} )$ is a product of $3$ 
    such commutators or $6$ primitives.

    Let $B = \FF_2^{n-1} \times \{0\}$ and $A = \FF_2^{n-2} \times \{(0,1)\}$.
    They are disjoint because of the different last bit,
    and $\abs B = 2 \abs A \ge \abs{A} + 6 \ge 8$.
    \cref{lem:fourpqpinv} implies that every element of~$\Alt(A \cup B)$ is a product of $4 \cdot 6 = 24$ primitives.
    Let $B' = A \cup B$ and $A' = \FF_2^n \setminus B'$.
    We have $\abs{B'} = 2^{n-1} + 2^{n-2}$ and $\abs{A'} = 2^{n-1} - 2^{n-2}$.
    \cref{lem:fourpqpinv} implies that any element $g$ of $\Alt(\FF_2^n)$ is a product of
    $4$ conjugates of elements of $B'$,
    each of which is a product of $24$ primitives.
    Hence, $g$ is a product of $4 \cdot 24 = 96$ primitives.

    We have shown that $r \in \one \otimes \Alt(\FF_2^n)$ is a product $\prod_{i=1}^{96} (g_i \phi(h_i) g_i^{-1})$
    where $g_i \in \one \otimes \Alt(\FF_2^n)$,
    $h_i \in \Alt(\FF_2^{n-1})$,
    and $\phi : \Alt(\FF_2^{n-1}) \xrightarrow{~\cong~} \one \otimes \Alt(\FF_2^{n-1}) \otimes \one$.
    Let $h'_i \in \Alt(\FF_2^n) \otimes \one$ 
    be the image of $h_i$ under the isomorphism $\Alt(\FF_2^{n-1}) \to \Alt(\{1\} \times \FF_2^{n-1}) \otimes \one$.
    This $h'_i$ amounts to $h_i$ controlled on the first bit.
    Replacing every $\phi(h_i)$ with $h'_i$,
    we have $r' = \prod_{i=1}^{96} (g_i h'_i g_i^{-1}) \in \Alt(\{1\} \times \FF_2^n)$
    which ranges over the entire $\Alt(\{1\} \times \FF_2^n)$.
    Therefore, we conclude that every element of $\Alt(\{1\} \times \FF_2^n)$ is a product of $3 \cdot 96$ elements
    of $(\Alt(\FF_2^{n}) \otimes \one) \cup (\one \otimes \Alt(\FF_2^{n}) )$.
    The same is true for $\Alt(\{0\} \times \FF_2^n)$, 
    $\Alt(\FF_2^n \times \{0\})$, and $\Alt(\FF_2^n \times \{1\})$
    by symmetry.
    
    Let $A = \{0\} \times \FF_2^n$, $B = \{1\} \times \FF_2^n$, $C = \FF_2^n \times \{0\}$, and $D = \FF_2^n \times \{1\}$.
    So, we have inclusions of $\Alt(A)$, $\Alt(B)$, $\Alt(C)$, and $\Alt(D)$ into $\Alt(\FF_2^{n+1})$.
    Note that $A \cap B = \emptyset = C \cap D$,
    but $\abs{A \cap C} = \abs{B \cap C} = \abs{A \cap D} = \abs{B \cap D} = 2^{n-1}$.
    Applying \cref{lem:GeneratingLargerSymmetricGroupFromSmallerOnes}
    to disjoint sets $A \setminus C, A \cap C, C \setminus A$,
    we obtain $\Alt(A \cup C) \cong \Alt(3 \cdot 2^{n-1})$,
    each element of which is a product of $3^2 \cdot 96$ elements from $(\Alt(\FF_2^{n}) \otimes \one) \cup (\one \otimes \Alt(\FF_2^{n}) )$.
    Similarly, we have $\Alt(B \cup D) \cong \Alt(3 \cdot 2^{n-1})$.
    Since $\abs{(A \cup C) \cap (B \cup D)} = 2^n$,
    we have $\Alt((A \cup C) \cup (B \cup D)) = \Alt(\FF_2^{n+1})$,
    each element of which is a product of $3^3 \cdot 96 = 2592$ elements 
    from $(\Alt(\FF_2^{n}) \otimes \one) \cup (\one \otimes \Alt(\FF_2^{n}) )$.
\end{proof}

\subsection{Proof of \texorpdfstring{\cref{thm:UltimatePermutationDesign}}{Theorem 5.2}}
We combine the discussion to give a full proof of~\cref{thm:UltimatePermutationDesign}.
\begin{proof}
    If $n = 18s$ for some $s \ge 15$,
    then \cref{prop:AltTwoToTheEighteenEssByCircuit} 
    says that every element of $\Alt(2^n)$ is a product of $\cO(1)$ elements of the union of $\cO(1)$ subgroups,
    each of which is generated by some bit flip conjugation of the generating set $S'$ of \cref{prop:KassabovGenCleanedUp}. 
    Hence, the Kazhdan constant is lower bounded by a constant by \cref{thm:Kassabov2007Alt} and \cref{cor:GoodKazhdanSubgroupToBigger}.
    If $n$ is not a multiple of $18$ or $n < 18 \cdot 15 = 270$,
    then we use \cref{prop:increasingNinAltTwoToTheN} recursively $17$ or fewer times
    to bound the Kazhdan constant from below by a constant.
    This complete the proof of \cref{thm:UltimatePermutationDesign}.
\end{proof}

\section{Circuit complexity and generalizations of Shannon's argument}\label{section:circuitcomplexity}

In this section, we prove that the circuit complexity of random classical and quantum circuits 
grows linearly for an exponentially long time.
This follows from our results that $\cO(t\cdot\poly(n))$-depth random (classical or quantum) circuits 
form $t$-designs via a similar counting argument as Shannon's proof 
that most Boolean functions have exponential circuit complexity. 
We have already stated these results as \cref{cor:lineargrowthreversible,cor:lineargrowth} in the introduction.
Here, we restate these corollaries for convenience and give a formal proof. 
Note that similar arguments already appeared multiple times in the literature, 
especially in the quantum setting ({\it e.g.},~\cite{brandao2016local,brandao2021models}), 
although the exact definitions and arguments may differ from ours.

Let us first briefly recall Shannon's argument~\cite{Shannon1949synthesis}: 
let $f:\{0,1\}^n\to \{0,1\}$ be a Boolean function. 
There are $2^{2^n}$ such functions in total. 
Let $M_R$ be the set of Boolean circuits that use $R$ AND/OR/NOT gates. 
For each gate, there are (at most) 2 input wires, 
which can be chosen from the $n$ input wires or some output wire of some gate. 
Therefore, $\left|M_R\right|\leq \left(\cO(n+R)^2\right)^R$. 
This implies that when $R=c\cdot 2^n / n$ for some small constant $c>0$, 
the probability of a function $f$ chosen uniformly at random having circuit complexity less than $R$ 
is at most $\left|M_R\right| / 2^{2^n}\leq 2^{-\Omega (2^n)}$.


In the following, we generalize this argument in the following sense:
we are going to compare the number of all functions in the ensemble generated by our random circuits of $L$ gates
to the number of all functions that can be expressed by $R$ gates.
The latter is exponential in $R$ and we already know an upper bound that is exponential in~$R$.
We have to show a lower bound for the former that is exponential in $L$.
Since our ensemble is generated by a specific mechanism, the random circuit,
we have to rule out the possibility that some small number of functions 
are produced in a very many different ways in our mechanism.
That is, we have to show that our random walk does not visit a small number of places too often.
To this end, we invoke our result that a random circuit of $L$ gates is $t$-design with $t = L / \poly(n)$.
A distribution~$\nu$ on the space of all permutations that is a $t$-design with a small multiplicative error,
is so uniform that the distribution has significant probability mass near every point in a set of $\exp(t)$ permutations.
This gives a desired lower bound on the number of elements in our ensemble generated by $L$ gates.
In the unitary case, essentially the same counting is performed
by discretizing the space of unitaries by $\delta$-nets.
All this is expressed in terms of probabilities in the proofs below.

\begin{definition}[Reversible circuit complexity]
    Let $f:\{0,1\}^n\to\{0,1\}^n \in \Alt(2^n)$ be a one-to-one function. 
    The reversible circuit complexity $C_R(f)$ 
    is defined as the size of the smallest reversible circuit on $n$ bits (using 3-bit reversible gates) 
    that computes $f$, 
    {\it i.e.},
    \begin{equation}
        C_R(f) = \min\{m:\text{there is a reversible circuit with }m\text{ gates that computes }f\}.
    \end{equation}
\end{definition}

\noindent
The condition that $f$ be an even permutation 
is necessary (and sufficient) because every $3$-bit reversible gate on $n \ge 4$
bits is an even permutation.

\revlineargrowth*

\begin{proof}
Let $N = 2^n$. Let $\pi\sim\nu$ be a permutation drawn from an approximate permutation $t$-design $\nu$ with multiplicative error $\eps$, where $t\leq c_0 N$ for some small constant $c_0>0$. For any fixed permutation $\sigma\in \Sym(N)$ and any fixed distinct $t$-tuple $x_1,x_2,\dots,x_t\in\{0,1\}^n$, we have
\begin{align}
\begin{split}
    \Pr_{\pi\sim\nu}[\pi(x)=\sigma(x),\,\forall x\in\{0,1\}^n]&\leq  \Pr_{\pi\sim\nu}[\pi(x_1)=\sigma(x_1),\dots, \pi(x_t)=\sigma(x_t)]\\
    &\leq\frac{1+\eps}{N(N-1)\cdots(N-t+1)}\\
    &\leq \Omega(N)^{-t}.
    \end{split}
\end{align}
Denote by $M_R$ the set of all permutations that can be represented by a reversible circuit acting on $n$ bits using $R$ reversible $3$-bit gates, which has size at most
\begin{equation}
    \left|M_R\right|\leq \left(\binom{n}{3}8!\right)^R = \left(\mathcal O(n^3)\right)^R.
\end{equation}
By union bound,
\begin{equation}
    \Pr_{\pi\sim\nu}[\exists \sigma\in M_R\,\,\text{s.t.}\,\, \pi=\sigma]\leq \sum_{\sigma\in M_R}\Pr_{\pi\sim\nu}[\pi=\sigma]\leq 2^{\cO(\log n)\cdot R} \cdot 2^{-\Omega(n)t}.
\end{equation}
Choose $R=c_1\cdot nt/\log n$ for some sufficiently small constant $c_1>0$. Then the probability of a random $\pi\sim\nu$ having circuit complexity less than $R$ is at most $2^{-\Omega(n)t}$.

Recall that $\nurevalltoall$ denotes the distribution of choosing a random set of 3 bits, and then applying a random gate from $\Sym(2^3)$ on those bits. For any $t\leq N-2$, \cref{thm:main_rev} together with \cref{lemma:equivalenceSandAlt} imply that $g(\nurevalltoall, \ t, \ \Sym(2^n)) \leq 1-\Omega(n^{-3})$. By \cref{lemma:fromgtopermutationdesign}, this implies that $\nurevalltoall^{*L}$ is an approximate permutation $t$-design with constant multiplicative error when $L\geq c_2 n^4 t$ for some constant $c_2>0$. In other words, when $L=\cO(2^n)$, a random reversible circuit with $L$ gates is an approximate permutation $t$-design with constant multiplicative error for $t=\Omega (L/n^4)$, and therefore has circuit complexity $R=\Omega (L/n^3\log n)$ with probability at least $1-2^{-\Omega(L/n^3)}$.
\end{proof}

Next, we define the quantum circuit complexity of a unitary matrix. Note that, unlike Boolean functions, the set of unitary matrices is continuous, so a reasonable notion of quantum circuit complexity should be ``robust,'' 
{\it i.e.}, allow some approximation error. Here we think of this error as some small constant.

\begin{definition}[Quantum circuit complexity]\label{def:quantumcircuitcomplexity}
Let $U\in\SU(2^n)$ be a unitary on $n$ qubits. Fix a constant $\delta>0$. The quantum circuit complexity $C_{Q,\delta}(U)$ is defined as the size of the smallest quantum circuit on $n$ qubits (using 2-qubit unitary gates with all-to-all connectivity) that approximates $U$ within $\delta$ error, {\it i.e.},
\begin{equation}
    C_{Q,\delta}(U)=\min\{m:\text{there is a quantum circuit }V\text{ with }m\text{ gates such that }\norm{U-V}_\infty\leq\delta\}.
\end{equation}
We also consider a stronger notion of circuit complexity: the smallest quantum circuit $U$ that prepares the state $\ket{\psi}=U\ket{0^n}$, {\it i.e.},
\begin{equation}
    C_{Q,\delta}(\ket{\psi})=\min\left\{m:\text{there is a quantum circuit }V\text{ with }m\text{ gates such that }\left|\bra{0^n}V^\dag \ket{\psi}\right|^2\geq 1-\delta^2\right\}.
\end{equation}
\end{definition}
By definition, $C_{Q,\delta}(U)\geq C_{Q,\delta}(U\ket{0^n})$. The second notion is much stronger as it says that no small circuit can approximate $U$ even on a single input. Our circuit lower bound works for this stronger notion.

\quantumlineargrowth*

This follows in a similar spirit to~\cref{cor:lineargrowthreversible} but is more technically involved. 

\begin{proof}
    Let $\nu$ be an approximate $t$-design with constant multiplicative error $\eps$. For any fixed vector $\ket{\phi}$ and $\delta_1\in(0,1)$ we have
    \begin{align}
    \begin{split}
        \Pr_{U\sim\nu}[\left|\bra{\phi}U\ket{0^n}\right|^2\geq 1-\delta_1]&= \Pr_{U\sim\nu}[\left|\bra{\phi}U\ket{0^n}\right|^{2t}\geq (1-\delta_1)^t]\\
        &\leq\frac{1}{(1-\delta_1)^t} \avg_{U\sim\nu}\left|\bra{\phi}U\ket{0^n}\right|^{2t}\\
        &\leq\frac{1+\eps}{(1-\delta_1)^t} \avg_{U\sim\mu(\SU(2^n))}\left|\bra{\phi}U\ket{0^n}\right|^{2t}\\
        &=\frac{1+\eps}{(1-\delta_1)^t} \binom{2^n+t-1}{t}^{-1}\\
        &\leq \frac{1+\eps}{(1-\delta_1)^t} \cdot\frac{t!}{2^{nt}}.
        \end{split}
    \end{align}
Here, the second line is by Markov's inequality; the third line follows from \cref{def:unitary_designs}; the fourth line uses properties of the symmetric subspace ({\it e.g.},~\cite[Proposition 6]{harrow2013church}).

Consider a $\delta_2$-net over $\SU(4)$ (a finite set $G_{\delta_2}\subseteq\SU(4)$ such that $\forall V\in\SU(4)$, $\exists V'\in G_{\delta_2}$ such that $\norm{V-V'}_\infty\leq \delta_2$). By standard arguments (see {\it e.g.},~\cite{milman1986asymptotic}) such a net can be constructed using $|G_{\delta_2}|=(\cO(1)/\delta_2)^{16}$ elements. Let $M_{G_{\delta_2},R}$ be the set of quantum circuits constructed using $R$ gates from $G_{\delta_2}$, which has size at most
\begin{equation}
\left|M_{G_{\delta_2},R}\right| \leq \left(\binom{n}{2}|G_{\delta_2}|\right)^R.
\end{equation}
By triangle inequality, for any quantum circuit $V$ with $R$ gates, there exists $V'\in M_{G_{\delta_2},R}$ such that $\norm{V-V'}_\infty\leq R \delta_2$. This implies that for any vector $\ket{\psi}$,
\begin{equation}
\begin{aligned}
    &\exists V\text{ with }R\text{ gates such that }\left|\bra{0^n}V^\dag \ket{\psi}\right|^2\geq 1-\delta^2\\
    \Rightarrow\,\,&\exists V'\in M_{G_{\delta_2},R}\text{ such that }\left|\bra{0^n}V'^\dag \ket{\psi}\right|^2\geq 1-\delta^2 - 2R\delta_2.
\end{aligned}
\end{equation}
Let $\delta_2 = \frac{\delta^2}{2R}$. Then,
\begin{align}
\begin{split}
&\Pr_{U\sim\nu}\left[\exists V\text{ with }R\text{ gates s.t. }\left|\bra{0^n}V^\dag U\ket{0^n}\right|^2\geq 1-\delta^2\right]\\
\leq &\Pr_{U\sim\nu}\left[\exists V'\in M_{G_{\delta_2},R}\text{ s.t. }\left|\bra{0^n}V'^\dag U\ket{0^n}\right|^2\geq 1-2\delta^2\right]\\
\leq & \left|M_{G_{\delta_2},R}\right|\cdot\Pr_{U\sim\nu}\left[\left|\bra{\phi}U\ket{0^n}\right|^2\geq 1-2\delta^2\right]\\
\leq & \left|M_{G_{\delta_2},R}\right|\cdot\frac{1+\eps}{(1-2\delta^2)^t} \cdot\frac{t!}{2^{nt}}\\
\leq &\left(\cO(n^2)\left(\frac{\cO(1)R}{\delta^2}\right)^{16}\right)^R \cdot\frac{\cO(1)}{(1-2\delta^2)^t} \cdot\frac{t!}{2^{nt}}.
\end{split}
\end{align}
This implies that there is a small constant $c_1>0$ such that when $R=c_1\cdot nt/\log(nt)$, we have
\begin{equation}
\Pr_{U\sim\nu}\left[\exists V\text{ with }R\text{ gates s.t. }\left|\bra{0^n}V^\dag U\ket{0^n}\right|^2\geq 1-\delta^2\right]\leq 2^{-\Omega(nt)}.
\end{equation}

Recall that $\nu_{\mathrm{2,\mathrm{All}\to\mathrm{All}},n}$ denotes the distribution of choosing a random set of 2 qubits, and then applying a Haar random 2-qubit gate on those qubits. For any $t\leq \cO(2^{n/2})$, \cref{thm:main_quantum} says that $$g(\nu_{\mathrm{2,\mathrm{All}\to\mathrm{All}},n}, \ t, \ \SU(2^n)) \leq 1-\Omega(n^{-3}).$$ By \cref{lemma:fromgtounitarydesign}, this implies that $\nu_{\mathrm{2,\mathrm{All}\to\mathrm{All}},n}^{*L}$ is an approximate unitary $t$-design with constant multiplicative error when $L\geq c_2 n^4 t$ for some constant $c_2>0$. In other words, when $L=\cO(2^{n/2})$, a random quantum circuit with $L$ gates is an approximate unitary $t$-design with constant multiplicative error for $t=\Omega (L/n^4)$, and therefore has circuit complexity $R=\Omega (L/n^4)$ with probability at least $1-2^{-\Omega(L/n^3)}$.
\end{proof}

\begin{remark}\label{rem:nonunitaryGates}
    In the two proofs above on the circuit complexity lower bounds,
    we have counted the number $\abs{M_R}$ of (sufficiently) different global operations 
    given a limited number $R$ of elementary gates.
    What matters is the result that $\abs{M_R}$ is only exponential in $R$.
    So, essentially the same counting holds even if we considered more general gate sets
    as long as the number of local choices of gates was modest, say, $\poly(n,R)$.
    This gives enough room for nonunitary gates that can even be local quantum channels or nonreversible classical gates that use ancillas,
    and still we would have the complexity lower bound on a richer class of implementation of a target unitary or permutation.
    For example, in the quantum setting we conclude that a typical unitary generated by the random unitary circuit of depth~$L$
    cannot be approximated in operator norm by any network of local tensors 
    unless the depth is at least linear in~$L$,
    where the local tensors are all arbitrary, with only restriction that its Frobenius norm is at most a constant.
    The last condition is to ensure that the local gate set is compact.
\end{remark}

\bibliographystyle{alphaurl}

\newcommand{\etalchar}[1]{$^{#1}$}

\appendix

\section{Overlap Theorem for Permutations}\label{section:permutationoverlap}

This section proves the following for the product of overlapping chunks of random permutations.
As used in the main text, we write $\Sym(A)$ for a finite set $A$ the full symmetric group of order $\abs{A}!$ 
which permutes elements of~$A$.
If $A = \{1,2,\ldots,N\}$, we often write $\Sym(N) = \Sym(A)$.

\begin{theorem}[Overlapping permutations]\label{thm:permutationoverlap}
    For finite sets $A$, $B$, and $C$,
    let $\{\ket a \otimes \ket b \otimes \ket c ~|~ a \in A, b \in B, c \in C\}$
    be an orthonormal basis for $\CC^{\abs{A}} \otimes \CC^{\abs{B}} \otimes \CC^{\abs{C}}$.
    Then, 
    \begin{align}
            \lnorm{ 
                \avg_{\substack{\pi_{AB} \sim \mu(\Sym(A\times B)) \\ \pi_{BC} \sim \mu(\Sym(B \times C))}} P(\pi_{AB})^{\otimes t} P(\pi_{BC})^{\otimes t} 
                -  
                \avg_{\pi_{ABC} \sim \mu(\Sym(A \times B \times C))} P(\pi_{ABC})^{\otimes t}
            }_{\infty} 
            \le  
            \CO\left( 
                \frac{\left( t\log t + \log \abs{B} \right)^{3}}{\sqrt{\abs{B}}}
            \right)
    \end{align}
    where we implicitly meant $P(\pi_{AB}) = P(\pi_{AB}) \otimes \one_C$  and $P(\pi_{BC}) = \one_A \otimes P(\pi_{BC})$.
\end{theorem}

That is, whenever the overlap region is reasonably large ($\labs{B} \ge \poly(t)$), 
we can effectively emulate a larger permutation acting jointly on all $A\times B \times C$, 
by a product of individually smaller permutations on subsystems $A \times B$ and $B \times C$. 
This phenomenon is at the heart of attaining linear $n$-dependence for the gap of random quantum circuits 
where $n$ is the number of qubits.

Recall that for the overlap lemma for unitaries~\cite{brandao2016local}, 
the idea is to analyze the $+1$-eigenvectors of $\avg_U  (U \otimes \overline U)^{\otimes t}$, 
which are labeled by permutations on $t$ elements due to Schur--Weyl duality.
The permutation group on $N$ letters is much smaller than the $N$-dimensional unitary group,
so the $+1$-eigenspace of $\avg_\pi P(\pi)^{\otimes t}$ is much larger.
We will find a basis for this eigenspace by \emph{partitions} of the set of $t$ elements.
We include minimal exposition regarding spanning sets of this eigenspace in \cref{sec:paritions_permutations}. 
Further details can be found in, {\it e.g.},~\cite{halverson2020set} in the algebra literature, 
but we quote some results from~\cite{low2010pseudo} and~\cite{chen2024efficient}.

In \cref{sec:asym_limit} we first derive a \textit{weaker} version (\cref{lemma:permutationoverlap_LargeN}) 
of the overlap theorem in the easy limit $\labs{B} \to \infty$ 
where $\avg P(\pi_{AB})^{\otimes t}$, $\avg P(\pi_{BC})^{\otimes t}$, and $\avg P(\pi_{ABC})^{\otimes t}$ 
are ``factorized'' and the norm in \cref{thm:permutationoverlap} tends to zero.
Our interest is of course in finite $B$,
but, a priori, the limit does not tell us anything to this end.
Hence, in \cref{sec:interpolate},
we exploit an additional structure 
that the expression is a \emph{low degree polynomial} in the inverse dimension $1/\abs{B}$ (\cref{lem:partition_props}~\cref{item:low-degree}), 
and we tame the finite-$\labs{B}$ behavior~\cite{chen2024efficient} from the limit,
circumventing the need for fine-grained combinatorics of set partitions.
The fundamental reason that such a ``soft'' interpolation argument is conceptually plausible 
is the Markov's other inequality,
saying that a bounded, low-degree polynomial cannot change too quickly. 
This inequality is also used in the \textit{polynomial method} in query complexity theory.

\subsection{Partitions and eigenvectors}
\label{sec:paritions_permutations}

We have to describe the eigenvectors of 
\begin{align}
    \avg_{\pi \sim \mu(\Sym(N))} P(\pi)^{\otimes t}\quad \text{where}\quad P(\pi)\ket{i} = \ket{\pi(i)} \quad \text{for each}\quad i\in [N].
\end{align}
For the eigenvalue~$+1$, the eigenspace is spanned by vectors 
that are indexed by \emph{partitions} $\Pi$ of $[t] = \{1,2,\ldots, t\}$.
Detailed exposition can be found in~\cite[Section 4.2.2]{low2010pseudo} and~\cite[Section III.D]{chen2024efficient},
but here we quote what we will use below.

A partition~$\Pi = \{S_1, \ldots \}$ of a set $S$ 
is a collection of disjoint subsets~$S_i$ whose union is~$S = \bigsqcup S_i$.
We will only consider partitions of~$[t]$.
Some definitions will be handy later:
\begin{itemize}
    \item $\Pi \vdash [t]$ means that $\Pi$ is a partition of~$[t]$. 
    An element of $\Pi$ is a subset of~$[t]$ and is called a \emph{block}.
    \item If $\Pi \vdash [t]$, we write $i \sim_\Pi j$ to mean that there is a block of $\Pi$ that contains both $i$ and $j$.
    It is an equivalence relation on $[t]$ defined by $\Pi$.
    \item The number of blocks of $\Pi$ is $\abs{\Pi}$.
    \item For two partitions $\Pi_1$ and $\Pi_2$ of a common set, 
    we say that $\Pi_1$ refines $\Pi_2$ and write $\Pi_1 \preceq \Pi_2$ 
    if $ i \sim_{\Pi_1} j$ always implies $i \sim_{\Pi_2} j$.
    That is, a refinement is obtained by splitting some blocks, 
    and therefore $\abs{\Pi_1} \ge \abs{\Pi_2}$ whenever $\Pi_1 \preceq \Pi_2$.
\end{itemize}
For example, 
\begin{align}
    \{\{1\},\{2\},\{3\}, \{4\} \} \preceq \{\{1,2\},\{3\}, \{4\} \} \preceq \{\{1,2,4\}, \{3\} \}
    = \Pi \vdash [4], \quad \text{and}\quad \labs{\Pi} =2.
\end{align}

\newcommand{\unn}[1]{{\color{red}{\tilde{#1}}}} 
\newcommand{\nor}[1]{#1} 

Given a partition~$\Pi$ of $[t]$, we consider two sets of tuples:
\begin{align}
    M_{\Pi} &\deq  \{ (m_1, \dots, m_t) \in [N]^t ~|~ m_i = m_j \text{ if } i \sim_\Pi j \}\,. \\
    M'_{\Pi} &\deq  \{ (m_1, \dots, m_t) \in [N]^t ~|~ m_i = m_j \text{ if and only if } i \sim_\Pi j \}\subseteq M_\Pi \,.\nonumber
\end{align}
The set $M_\Pi$ contains all tuples that have the same entry on indices that are in the same block of the partition.
The set $M'_{\Pi}$ comes with the additional restriction that the entries on indices in different blocks be \emph{distinct}.
We then define the following unnormalized states,
which will span the $+1$-eigenspace of $\avg_{\pi\in \Sym(N)} P(\pi)^{\otimes t}$ (see \cref{lem:partition_props}):
\begin{align}
    \ket{\unn O_{\Pi}} \deq   \sum_{\vec{m} \in M_{\Pi}} \ket{\vec{m}} 
    \qquad \text{and} \qquad 
    \ket{\unn O'_{\Pi}} \deq   \sum_{\vec{m} \in M'_{\Pi}} \ket{\vec{m}} \,. \label{eq:unnOOprimeDefinition}
\end{align}
As an example, for $\Pi = \{\{1,2,4\}, \{3\} \}$
\begin{align}
    \ket{\unn O_{\Pi}} = \sum_{i,j=1}^{N} \ket{i}\ket{i} \ket{j} \ket{i}
    \quad \text{and}\quad 
    \ket{\unn O'_{\Pi}} = \sum_{i,j=1}^{N} \indicator(i\ne j)\ket{i}\ket{i} \ket{j} \ket{i}
\end{align}
where $\indicator$ assumes~$1$ if the clause is true and $0$ otherwise.
Due to the distinctness constraint, 
the states $\ket{\unn O'_{\Pi}}$ are orthogonal for different partitions $\Pi$.
We will also use their normalized versions.
\begin{align}
    \ket{\nor O_{\Pi}} \deq \frac{\ket{\unn O_{\Pi}}}{\norm{\ket{\unn O_{\Pi}}}} 
    \qquad \text{and} \qquad 
    \ket{\nor O'_{\Pi}} \deq \frac{\ket{\unn O'_{\Pi}}}{\norm{\ket{\unn O_{\Pi}}}} \,. \label{eqn:normalized_part_states}
\end{align}

We collect a few basic properties of these states.
See ~\cite[Section 4.2.2]{low2010pseudo} and~\cite[Section III.D]{chen2024efficient} and the references therein for details and proofs. For our purposes, we will only work in the \emph{stable} range $N \ge t$ where the basis are better behaved. 
\begin{lemma}[Eigenspace basis by partitions] \label{lem:partition_props}
    ~
    \begin{enumerate}
        \item The set $\{ \ket{\unn O_\Pi}\}_{\Pi \vdash [t]}$ is a (nonorthogonal) basis
        for the $+1$-eigenspace of $\avg_{\pi \in \Sym(N)} P(\pi)^{\ot t}$. \label{item:span}
        \item If $N \geq t$, the states $\{ \ket{\nor O'_\Pi}\}_{\Pi \vdash [t]}$ 
        form an orthonormal basis for the $+1$-eigenspace of $\avg_{\pi \in \Sym(N)} P(\pi)^{\ot t}$. \label{item:ortho_basis}
        \item $\norm{\ket{\unn O_{\Pi}}}^2 = \abs{M_{\Pi}} = N^{\labs{\Pi}}$ and
        \footnote{If $N-\labs{\Pi}+1 \le 0$, then $M'_{\Pi}$ vanishes. 
            We do not use this regime.
        }
        $\norm{\ket{\unn O'_{\Pi}}}^2 = \abs{M'_\Pi}= N(N-1)\cdots (N-\abs{\Pi}+1)$.
        \label{item:low-degree}
    \end{enumerate}
\end{lemma}

\begin{remark}[Factorization]
    For a finite set $S$, we define the states $\ket{\unn O_\Pi^{(S)}}$ and $\ket{\unn O'^{(S)}_\Pi}$ as above, 
    but with permutations on~$S$ rather than $[N]$.
    If $S = A \times B$ is a Cartesian product,
    we have
        \begin{align}
            \ket{\unn O^{(A \times B)}_\Pi} = \ket{\unn O^{(A)}_\Pi}\otimes \ket{\unn O^{(B)}_\Pi}. \label{eqn:factorisation}
        \end{align}  
        We will write $\ket{\unn O^{(AB)}_\Pi} = \ket{\unn O^{(A \times B)}_\Pi}$ for short, 
        and similarly $\ket{\unn O'^{(AB)}_\Pi} = \ket{\unn O'^{(A \times B)}_\Pi}$.
        Note that the property~\cref{eqn:factorisation} only holds for nonorthogonal basis states
        $\ket{\unn O^{(A \times B)}_\Pi}$, 
        not for the orthogonal basis states $\ket{\unn O'^{(A \times B)}_\Pi}$.
\end{remark}

\noindent
One can switch between the orthogonal states $\ket{\unn O'_{\Pi}}$ and nonorthogonal states $\ket{\unn O_{\Pi}}$ 
using a generalized inclusion-exclusion principle for partial orders.

\begin{lemma}[Inclusion-exclusion] \label{lem:incl_excl} 
    If $N \ge t$, then for any partition $\Pi \vdash [t]$ 
    \begin{align}
        \ket{\unn O_{\Pi}}
        &= \sum_{\Sigma \vdash [t] } K_{\Pi\Sigma} \ket{\unn O'_{\Sigma}}\quad \text{where}\quad K_{\Pi\Sigma} = \indicator(\Pi \preceq \Sigma) = \begin{cases} 1 & \text{ if } \Pi \preceq \Sigma \\ 0 & \text{otherwise} \end{cases} \notag \,,\\
        \ket{\unn O'_{\Sigma}} &=  \sum_{\Pi \vdash [t] ,\, \Pi \succeq \Sigma} (K^{-1})_{\Sigma \Pi} \ket{\unn O_{\Pi}}\,.\label{eq:Kinv} 
    \end{align}
\end{lemma}

\noindent 
For example,
\begin{align}
    \sum_{i,j=1}^{N}\ket{i}\ket{j} = \ket{\unn O_{\{\{1\},\{2\}\} }} &= \ket{\unn O'_{\{\{1\},\{2\}\} }} +\ket{\unn O'_{\{\{1,2\}\} }}
    =\sum_{i,j=1}^{N}\indicator(i\ne j)\ket{i}\ket{j} +\sum_{i=1}^{N}\ket{i}\ket{i}.
\end{align}
The matrix $K$ is invertible 
because if we sort all partitions according to the number of blocks,
then $K$ is upper-triangular with $1$ on the diagonal.
Since $K$ is an integral matrix with determinant~$1$,
the inverse~$K^{-1}$ is also integral and upper-triangular.
A reader may wish to note that $K^{-1}$ is the M\"obius inversion of~$K$ 
associated with the partial order~$\preceq$ on all partitions of~$[t]$.
Trivially but importantly, the matrix~$K$, with entries labeled by partitions $\Pi,\Sigma$, depends only on~$t$, not $N$.
The entries of~$K^{-1}$ can be bounded as follows.

\begin{lemma}[{Combinatorial bounds on partitions~\cite[Corollary 4.2.6]{low2010pseudo}}]\label{lemma:comb_bound_partitions}
    For each partition $\Pi\vdash [t]$,
    \begin{align}
        \sum_{\Sigma \succeq \Pi} 1
        \le 
        \sum_{\Sigma \succeq \Pi} \labs{K^{-1}_{\Pi\Sigma}} 
        = 
        \labs{\Pi}! 
        \le 
        t! \; .
    \end{align}
\end{lemma}

\subsection{Infinite \texorpdfstring{$B$}{B} limit} 
\label{sec:asym_limit}

We prove \cref{thm:permutationoverlap} in the limit where the dimension $|B|$ of the overlap goes to infinity.
This will follow from the fact that in this limit, the states $\ket{\nor O_\Pi}$ become orthogonal.

\begin{lemma}[Asymptotic orthogonality of partitions]\label{lemma:asym_ortho_partitions} 
    For any fixed $t$,
    \begin{align}
        \lim_{N \to \infty}
        \lnorm{
            \avg_{\pi \in \Sym(N)} [P(\pi)^{\otimes t} ] - \sum_{\Pi \vdash [t]} \proj{\nor O_\Pi}
        }_{\infty} 
        = 0 .
    \end{align}
\end{lemma}

\begin{proof}
    We can freely interchange $\ket{\nor O'_\Pi}$ and $\ket{\nor O_\Pi}$ 
    since they are equal in the large-$N$ limit:
    \begin{align}
        \lim_{N\to\infty} 
        \norm{
            \ket{\nor O_\Pi} - \ket{ \nor O'_\Pi}
        } = 0 \label{eq:OO'_largeN}.
    \end{align}
    The lemma then follows from \cref{lem:partition_props} \cref{item:ortho_basis}, 
    which implies $\avg_{\pi \in \Sym(N)} [P(\pi)^{\otimes t} ]= \sum_{\Pi \vdash [t]} \proj{\nor O'_\Pi}$ for $N \geq t$.
\end{proof}

\begin{lemma}[Asymptotic overlap lemma]\label{lemma:permutationoverlap_LargeN}
    \begin{align}
        \lim_{\abs{B} \to \infty}
        \norm*{ 
                \avg_{\substack{\pi_{AB} \sim \mu(\Sym(A\times B)) \\ \pi_{BC} \sim \mu(\Sym(B \times C))}} P(\pi_{AB})^{\otimes t} P(\pi_{BC})^{\otimes t} 
                -  
                \avg_{\pi_{ABC} \sim \mu(\Sym(A \times B \times C))} P(\pi_{ABC})^{\otimes t}
            }_{\infty} 
        = 0 .
    \end{align}
\end{lemma}
\begin{proof}
From \cref{lemma:asym_ortho_partitions} and~\cref{eqn:factorisation} we get that 
\begin{align*}
    0 &= \lim_{|B| \to \infty} 
    \norm*{
        \avg [P(\pi_{AB})^{\otimes t} P(\pi_{BC})^{\otimes t}] - 
        \left(
            \sum_{\Pi \vdash [t]}\proj{\nor O ^{(A)}_\Pi} \otimes \proj{\nor O^{(B)}_\Pi}
        \right)
        \left(
            \sum_{\Pi \vdash [t]} \proj{\nor O^{(B)}_\Pi} \otimes \proj{\nor O^{(C)}_\Pi}
        \right)
    }_\infty \\
\intertext{
    From \cref{eq:OO'_largeN} we have 
    $\braket{\nor O^{(B)}_\Pi | \nor O^{(B)}_{\Sigma}} \to \one[\Pi = \Sigma]$ as $\abs B \to \infty$, 
    so we can simplify:
}
    &= \lim_{|B| \to \infty} 
    \lnorm{
        \avg [P(\pi_{AB})^{\otimes t} P(\pi_{BC})^{\otimes t}] 
        - 
        \sum_{\Pi \vdash [t]}
        \proj{\nor O^{(A)}_\Pi} \otimes \proj{\nor O^{(B)}_\Pi} \otimes \proj{\nor O^{(C)}_\Pi}
    }_\infty \\
    &=\lim_{|B| \to \infty} 
    \lnorm{
        \avg [P(\pi_{AB})^{\otimes t} P(\pi_{BC})^{\otimes t}] 
        - 
        \sum_{\Pi \vdash [t]} \proj{\nor O^{(ABC)}_\Pi}
    }_{\infty}
    \tag{by~\cref{eqn:factorisation}}\\
    &= \lim_{|B| \to \infty} 
    \lnorm{
        \avg [P(\pi_{AB})^{\otimes t} P(\pi_{BC})^{\otimes t}] - \avg [P(\pi_{ABC})^{\otimes t}]
    }_{\infty},\tag{by~\cref{lemma:asym_ortho_partitions}}
\end{align*}
as advertised.
\end{proof}

\begin{remark}[Partitions are far from being orthogonal if $t > \log N$]\label{rem:partition_non_orthogonal}
    One should be cautious in interpreting the $N\rightarrow \infty$ calculation in \cref{lemma:asym_ortho_partitions} 
    since the states $\ket{\nor O_{\Pi}}$, 
    are far from being orthogonal even for small~$t$
    in the sense that $\norm*{\sum_{\Pi \vdash [t]} \proj{\nor O_\Pi}}_\infty   \ge  (2^{t-1}-1)/N$.
    Since $\ket{\nor O_\Pi}$ form a basis, this norm would be close to~$1$ if they were nearly orthogonal.
    (This contrasts to the unitary case~\cite{brandao2016local}.)
    To understand the lower bound on the norm, 
    consider a partition $\Sigma = \{ \{1,\cdots, t\} \} \vdash [t]$ with a single block
    and another partition~$\chi = \{ S_1, S_2 \}$ with two blocks.
    The associated states are $\ket{\unn O_\Sigma} = \sum_{i=1}^N \ket{i}^{\otimes t}$
    and $\ket{\unn O_\chi} = \sum_{i,j=1}^N \ket{i}^{\otimes \abs{S_1}}\ket{j}^{\otimes \abs{S_2}}$.
    Hence,
    \begin{align}
        \braket{\nor O_\Sigma | \nor O_\chi}\!\braket{\nor O_\chi | \nor O_\Sigma} = \frac{1}{N} \, .
    \end{align}
    However, there are $(2^{t} - 2)/2$ choices of such~$\chi$. Therefore,
    \begin{align}
        \norm*{\sum_{\Pi \vdash [t]} \proj{\nor O_\Pi} }
        \ge
        \bra{\nor O_\Sigma} \left( \sum_{\Pi \vdash [t]} \proj{\nor O_\Pi} \right)\ket{\nor O_\Sigma} 
        \ge 
        \frac{2^{t-1}-1}{N} \, .
    \end{align}
\end{remark}

\subsection{Polynomial interpolation}
\label{sec:interpolate}

The operator of our interest is
\newcommand{\DD}{{\mathds{D}}}
\begin{align}
    \DD_{ABC} = \avg_{\substack{\pi_{AB} \sim \mu(\Sym(A\times B)) \\ \pi_{BC} \sim \mu(\Sym(B \times C))}} 
    P(\pi_{AB})^{\otimes t} P(\pi_{BC})^{\otimes t} 
    -  
    \avg_{\pi_{ABC} \sim \mu(\Sym(A \times B \times C))} P(\pi_{ABC})^{\otimes t}
\end{align}
The norm $\norm{\DD_{ABC}}$, if regarded as a function of $x = 1 / \abs B$, assumes zero at $x=0$ 
by~\cref{lemma:permutationoverlap_LargeN}.
The expression~$\DD_{ABC}$ contains multiple parts involving $1/\abs B$,
some of which we will regard as fixed values and some other as an infinitesimal variable~$x$.
This is Step~1 in the proof of~\cref{thm:permutationoverlap} below.
We will then investigate a series expansion of a function that is closely related to~$\norm \DD$ in the infinitesimal~$x$,
truncate it to obtain a polynomial approximation of a certain low degree (Step~2 and~3),
and bound the error term using Markov's other inequality (Step~4),
which relates the derivative of a polynomial by its values and degree.
The series expansion in~$x$ can be thought of as an interpolation 
between the limit value $0 = \lim_{\abs{B} \to \infty} \norm \DD$ at~$x = 0$ 
and the value $\norm{\DD_{ABC}}$ of interest at~$x = 1/\sqrt{\abs B}$.

\begin{proof}[Proof of~\cref{thm:permutationoverlap}]
Since the goal is to estimate $\norm{\DD_{ABC}}$,
we are going to bound the magnitude of~$\bra{\psi_\text{left}^N} \DD_{ABC} \ket{\psi_\text{right}^N}$
for arbitrary $\ket{\psi_\text{left}^N}$ and $\ket{\psi_\text{right}^N}$.
Observe that the averages, $\avg_{\pi_{AB}} P(\pi_{AB})^{\otimes t}$ and two others that appear in $\DD_{ABC}$, 
are each a projector,
and the image of $\avg_{\pi_{ABC}} P(\pi_{ABC})^{\otimes t}$ is contained in the intersection of the other two.
Hence, we may assume without loss of generality that $\ket{\psi_\text{left}^N}$ 
is in the image of $\avg_{\pi_{AB}} P(\pi_{AB})^{\otimes t}$
and that $\ket{\psi_\text{right}^N}$ is in the image of $\avg_{\pi_{BC}} P(\pi_{BC})^{\otimes t}$.

\paragraph{Step~1: defining an interpolating function~$f$.}

Using the assumption that $\ket{\psi_\text{left, right}^N}$ are in the image of the projectors,
we expand them in the basis~$\{\ket{\nor O'_\Pi}\}_\Pi$:
\begin{align}
        \ket{\psi_\text{left}^N} &= \sum_{\Pi\vdash [t]} \ket{\nor{O'}^{(AB)}_\Pi} \otimes \ket{\unn u^{(C)}_{\Pi}}
        \quad \text{such that}\quad 
        \norm{\ket{\psi_\text{left}^N}}^2 = \sum_{\Pi\vdash [t]} \braket{\unn u^{(C)}_{\Pi} | \unn u^{(C)}_{\Pi}} = 1,\label{eq:uCandvA}\\
        \ket{\psi_\text{right}^N} &=\sum_{\Pi\vdash [t]} \ket{\unn v^{(A)}_{\Pi}} \otimes \ket{\nor{O'}^{(BC)}_{\Pi}}
        \quad \text{such that}\quad 
        \norm{\ket{\psi_\text{right}^N}}^2 = \sum_{\Pi\vdash [t]} \braket{\unn v^{(A)}_{\Pi} | \unn v^{(A)}_{\Pi}} = 1 \nonumber
\end{align}
Now we choose which part to regard as fixed and which part as varying.
Our tripartite system has been $ABC$, 
but we will consider $X$ in place of~$B$ where $\abs X = N'$ may differ from $\abs B = N$.
This notation is a reminder that $x = 1/\sqrt{N'}$ will be regarded as an infinitesimal variable.
We define
\begin{align}
    \ket{\phi_\text{left}^{N'}} 
    =
    \sum_{\Pi\vdash [t]} \ket{\nor O'^{(AX)}_{\Pi}}\otimes \ket{ \unn u^{(C)}_{\Pi}} 
    \quad \text{and} \quad 
    \ket{\phi_\text{right}^{N'}} 
    =
    \sum_{\Pi\vdash [t]}  \ket{\unn v^{(A)}_{\Pi}} \otimes \ket{\nor O'^{(X C)}_{\Pi}} \, .
\end{align}
That is, we retain $\ket{\unn u^{(C)}_\Pi}$ and $\ket{\unn v^{(A)}_\Pi}$
but replace $B$ with $X$ from $\ket{\psi_\text{left, right}^N}$ using 
\cref{eq:unnOOprimeDefinition,eqn:normalized_part_states}
with $A$ and $C$ unchanged.
This makes sense because $\Pi$ are partitions of~$[t]$, having nothing to do with~$B$ or~$X$,
and consequently the number of basis states $\ket{\nor{O'}^{(XC)}_{\Pi}}$ depends only on~$t$, not~$\abs X$.

We now define a complex-valued function of $x = 1/ \sqrt{N'}$ by
\begin{align}
    f(x) 
    = 
    \Bra{\phi_\text{left}^{N'}} 
        \DD_{AXC}
    \Ket{\phi_\text{right}^{N'}}
    =
    \Bra{\phi_\text{left}^{N'}} 
        \avg_{\pi_{AX}, \pi_{XC}} P(\pi_{AX})^{\otimes t} P( \pi_{XC})^{\otimes t} -  \avg_{\pi_{AXC}} P(\pi_{AXC})^{\otimes t} 
    \Ket{\phi_\text{right}^{N'}}.\label{eqn:fwrittenup}
\end{align}
The function~$f$ depends on the subsystems $A,C$,
and the normalized vectors $\ket{\psi_\text{left}^{N}}$ and $\ket{\psi_\text{right}^{N}}$.
By the construction and the limit result of~\cref{lemma:permutationoverlap_LargeN},
we have
\begin{align}
    f\parens{\frac 1 {\sqrt N}} &= \bra{\psi_\text{left}^N} \DD_{ABC} \ket{\psi_\text{right}^N} \, , \\
    \abs{ f(x) } &\leq 2 \, \quad\text{if $x = \frac{1}{\sqrt{N'}}$ for any integer~$N' \ge 1$},\label{eqn:fbounded}\\ 
    \lim_{N' \to \infty} f\parens{\frac{1}{\sqrt{N'}}} &= 0 \,. \label{eqn:ftozero}
\end{align}

\paragraph{Step $2$: the $N'$-dependence in $f$.}

The domain of the function~$f$ is \emph{not} a continuous real interval,
but a discrete set~$\{ \frac 1 {\sqrt n} \}_n \subset (0,1]$.
We will extend the domain of~$f$ to a continuous interval $[0, \frac 1 {\sqrt t}]$ by an obvious formula,
which we will later use to find a polynomial approximation.

Let us examine $f(x)$ term by term.
\begin{align}
f(\frac{1}{\sqrt{N'}}) 
&= 
\bra{\phi_\text{left}^{N'}} \DD_{AXC} \ket{\phi_\text{right}^{N'}} 
= 
\braket{\phi_\text{left}^{N'} | \phi_\text{right}^{N'}} 
- 
\bra{\phi_\text{left}^{N'}} \avg_{\pi_{AX}} P(\pi_{AX C})^{\otimes t} \ket{\phi_\text{right}^{N'}} \label{eqn:fwrittenout}
\end{align}
For the first term, we use the factorization of basis vectors: 
\begin{align}
    \ket{\unn O'^{(BC)}_\Pi} &= \sum_{\Sigma \succeq \Pi }K^{-1}_{\Pi\Sigma}\ket{\unn O^{(BC)}_{\Sigma}}\tag{by inclusion-exclusion:~\cref{lem:incl_excl}}\\
    &= \sum_{\Sigma\vdash [t] }K^{-1}_{\Pi\Sigma}\ket{\unn O^{(B)}_{\Sigma}}\otimes \ket{\unn O^{(C)}_{\Sigma}}\tag{by factorization:~\cref{eqn:factorisation}}\\
    &= \sum_{\Sigma\vdash [t] }K^{-1}_{\Pi\Sigma} \sum_{\Pi_{B} \succeq \Sigma} \ket{\unn O'^{(B)}_{\Pi_{B}}}\otimes \ket{\unn O^{(C)}_{\Sigma}}.\tag{by inclusion-exclusion:~\cref{lem:incl_excl}}
\end{align}
Therefore, the state $\ket{\phi_\text{right}^{N'}}$ can be written in terms of $\ket{\unn O'^{(X)}_{\chi}}$:
\begin{align}
    \ket{\phi_\text{right}^{N'}} &=\sum_{\Pi \vdash [t]} \ket{ \unn v^{(A)}_{\Pi}} \otimes \ket{\nor O'^{(XC)}_{\Pi}}
    =\sum_{\Pi\vdash [t]} \frac{\ket{ \unn v^{(A)}_{\Pi}} }{\norm*{\ket{ \unn O'^{(XC)}_{\Pi}}}} \otimes \sum_{\Sigma \succeq \Pi }K^{-1}_{\Pi\Sigma} \sum_{\chi \succeq \Sigma} \ket{\unn O'^{(X)}_{\chi}} \otimes \ket{\unn O^{(C)}_{\Sigma}}. \label{eq:phirightexpansion}
\end{align}
Similarly,
\begin{align}
    \bra{\phi_\text{left}^{N'}} 
    &= \sum_{\Pi \vdash [t]} \sum_{\Sigma \succeq \Pi} \sum_{\chi \succeq \Sigma} 
    \frac{K^{-1}_{\Pi\Sigma}}{\norm*{\ket{\unn O'^{(AX)}_\Pi}}} \bra{\unn O^{(A)}_\Sigma} \bra{\unn O'^{(X)}_\chi} \bra{\unn u^{(C)}_\Pi} \, .
\end{align}
Hence,
\begin{align}
    \braket{ \phi_\text{left}^{N'} | \phi_\text{right}^{N'} } 
    &=
    \sum_{\Pi_{1}, \Pi_2, \Pi_{3} \vdash [t]} c(\Pi_1, \Pi_2, \Pi_3)\cdot 
    \frac{
        \sqrt{\abs A}^{\abs{\Pi_1}} \sqrt{\abs C}^{\abs{\Pi_2}}\cdot\norm*{\ket{\unn O'^{(X)}_{\Pi_3}}}^2
    }{
        \norm*{\ket{\unn O'^{(AX)}_{\Pi_1}}} \cdot \norm*{\ket{\unn O'^{(XC)}_{\Pi_2}}}
    }
\end{align}
where the coefficients are
\begin{align}
c(\Pi_1,\Pi_2,\Pi_3) 
    &=
    \sum_{\substack{\Sigma_1,\Sigma_2 \, : \\
    \Pi_1 \preceq \Sigma_1 \preceq \Pi_3,\\
    \Pi_2 \preceq \Sigma_2 \preceq \Pi_3}} 
        K^{-1}_{\Pi_1,\Sigma_1}  K^{-1}_{\Pi_2,\Sigma_2} 
        \frac{\braket{\unn O^{(A)}_{\Sigma_1} | \unn v^{(A)}_{\Pi_1} }}{\sqrt{\abs A}^{\abs{\Pi_1}}}
        \frac{\braket{\unn u^{(C)}_{\Pi_2} | \unn O^{(C)}_{\Sigma_2}}}{\sqrt{\abs C}^{\abs{\Pi_2}}} \, ,\label{eq:denominators_ABBC}\\
    \abs{c(\Pi_1,\Pi_2,\Pi_3)}
    &\le
    \sum_{\Sigma_1 \succeq \Pi_1, \Sigma_2 \succeq \Pi_2 } 
        \abs{K^{-1}_{\Pi_1,\Sigma_1}} \abs{ K^{-1}_{\Pi_2,\Sigma_2} } \qquad \text{by \cref{eq:uCandvA} and \cref{lem:partition_props} \cref{item:low-degree}} \nonumber\\
    &\le
    (t!)^2
    \qquad \text{by \cref{lemma:comb_bound_partitions}} . \nonumber
    \nonumber
\end{align}
Importantly, the coefficients $c(\Pi_1, \Pi_2, \Pi_3)$ are \emph{independent of $N'$}.
The factors $\sqrt{\abs A}^{\abs{\Pi_1}}$ and $\sqrt{\abs C}^{\labs{\Pi_2}}$ are also independent of $N'$
but are inserted to make later formulas more convenient.

Similarly, for the second term in \cref{eqn:fwrittenout}, 
we have from \cref{lem:partition_props} \cref{item:ortho_basis}
\begin{align}
     \avg_{\pi_{AXC}} P(\pi_{AXC})^{\otimes t} 
     &=
     \sum_{\Pi \vdash [t]} \proj{\nor O'^{(AXC)}_\Pi} 
     = 
     \sum_{\Pi \vdash [t]} \frac{ \proj{\unn O'^{(AXC)}_\Pi} }{ \braket{\unn O'^{(AXC)}_{\Pi} | \unn O'^{(AXC)}_{\Pi}}}
\end{align}
where
\begin{align}
    \ket{\unn O'^{(AXC)}_\Pi} 
    &= \sum_{\Sigma \succeq \Pi } K^{-1}_{\Pi\Sigma}\ket{\unn O^{(AXC)}_{\Sigma}}
    = \sum_{\Sigma \succeq \Pi } \sum_{\chi \succeq \Sigma} 
    K^{-1}_{\Pi\Sigma} 
    \ket{\unn O^{(A)}_{\Sigma}} \ket{\unn O'^{(X)}_{\chi}} \ket{\unn O^{(C)}_{\Sigma}}.
\end{align}
Inserting these expansions into $\bra{\phi_\text{left}^{N'}} \avg [P(\pi_{AX C})^{\otimes t}] \ket{\phi_\text{right}^{N'}}$, 
we obtain
\begin{multline}
    \bra{\phi_\text{left}^{N'}} \avg [P(\pi_{A X C})^{\otimes t}] \ket{\phi_\text{right}^{N'}} \\
    = \sum_{\Pi_{1}, \Pi_2, \Pi_{3}, \Pi_{4},\Pi_5 \vdash [t]} 
    d(\Pi_1, \Pi_2, \Pi_3, \Pi_4,\Pi_5) 
    \frac{
        \sqrt{\labs{A}}^{\labs{\Pi_1}}
        \sqrt{\labs{C}}^{\labs{\Pi_2}}
        \sqrt{\labs{A}\labs{C}}^{2\labs{\Pi_3}}
        \norm*{\ket{\unn O'^{(X)}_{\Pi_4}}}^2 
        \norm*{\ket{\unn O'^{(X)}_{\Pi_5}}}^2
    }{
        \norm*{\ket{\unn O'^{(AX)}_{\Pi_1}}} \cdot 
        \norm*{\ket{\unn O'^{(XC)}_{\Pi_2}}} \cdot 
        \norm*{\ket{\unn O'^{(AXC)}_{\Pi_3}}}^2
    } \label{eq:denominators_ABC}
\end{multline}
where the $N'$-independent coefficients are 
\begin{align}
d(\Pi_1, \Pi_2, \Pi_3, \Pi_4, \Pi_5) 
&=
\parens*{
    \sum_{\Sigma_A, \Pi_A} K^{-1}_{\Pi_1 \Sigma_A} K^{-1}_{\Pi_3 \Pi_A} 
    \frac{\braket{ \unn O^{(A)}_{\Sigma_A} | \unn O^{(A)}_{\Pi_A}}}{\sqrt{\abs A}^{\abs{\Pi_1} + \abs{\Pi_3}}}
    \frac{\braket{ \unn u ^{(C)}_{\Pi_1} | \unn O^{(C)}_{\Pi_A} }}{\sqrt{\abs C}^{\abs{\Pi_3}}}
    \indicator[\Pi_4 \succeq \Sigma_A]
} \times \nonumber \\
&\quad\qquad \parens*{
    \sum_{\Sigma_C, \Pi_C} K^{-1}_{\Pi_2 \Sigma_C} K^{-1}_{\Pi_3 \Pi_C} 
    \frac{\braket{ \unn O^{(C)}_{\Pi_C} | \unn O^{(C)}_{\Sigma_C} }}{\sqrt{\abs C}^{\abs{\Pi_3}+\abs{\Pi_2}}}
    \frac{\braket{ \unn O^{(A)}_{\Pi_C} | \unn v ^{(A)}_{\Pi_2} }}{\sqrt{\abs A}^{\abs{\Pi_3}}}
    \indicator[\Pi_5 \succeq \Sigma_C]
} \, ,
\\
\abs{d(\Pi_1, \Pi_2, \Pi_3, \Pi_4, \Pi_5)} 
&\le
(t!)^4 \, .\nonumber
\end{align}

\paragraph{Step $3$: Low-degree approximation.}

It is now apparent that the $N'$-dependence of~$f$ comes through the norms,
which can be explicitly written in a simple formula if $N' \ge t$
since in this regime the number of blocks ($\le t$) in any partition does not exceed~$N'$.
Indeed, \cref{lem:partition_props} \cref{item:low-degree} gives nondivergent formulas
\begin{align}
    \norm*{\ket{\unn O'^{(X)}_{\Pi}}} 
    &= \sqrt{N'(N'-1) \cdots (N'-\abs{\Pi}+1)} 
    = \sqrt{N'}^{\abs{\Pi}} \sqrt{
        \parens*{1-\frac{1}{N'}}
        \parens*{1-\frac{2}{N'}}
        \cdots 
        \parens*{1-\frac{\labs{\Pi}-1}{N'}}
    } \,,
    \label{eq:normexpr}\\
    \frac{1}{\norm*{\ket{\unn O'^{(AX)}_{\Pi}}}} 
    &= 
    \frac{1}{\sqrt{\labs{A}N'}^{\labs{\Pi}}} \frac{1}{\sqrt{(1-\frac{1}{\labs{A}N'})(1-\frac{2}{\labs{A}N'})\cdots (1-\frac{\labs{\Pi}-1}{\labs{A}N'})}}\,. \nonumber
\end{align}
We have remarked that the domain of definition of~$f(x)$ is a discrete set,
but it is now obvious how to extend the domain to a continuous interval~$[0,1/\sqrt t] \ni x$ by~\cref{eq:normexpr}.
The $N'$-dependent part of~\cref{eq:denominators_ABBC} is then
\begin{align}
    \frac{
        \sqrt{\labs{A}}^{\labs{\Pi_1}}\sqrt{\labs{C}}^{\labs{\Pi_2}}\cdot\norm*{\ket{\unn O'^{(X)}_{\Pi_3}}}^2
    }{
        \norm*{\ket{\unn O'^{(AX)}_{\Pi_1}}} \cdot \norm*{\ket{\unn O'^{(XC)}_{\Pi_2}}}
    } 
    &=
    \frac{1}{\sqrt{N'}^{\labs{\Pi_1}+\labs{\Pi_2}-2\labs{\Pi_3}}}
    \cdot 
    \underbrace{\left(
        \frac{\sqrt{\labs{A} N'}^{\labs{\Pi_1}}}{\norm*{\ket{\unn O'^{(AX)}_{\Pi_1}}}}
        \cdot 
        \frac{\sqrt{\labs{C} N'}^{\labs{\Pi_2}}}{\norm*{\ket{\unn O'^{(XC)}_{\Pi_2}}}}
        \cdot
        \frac{\norm*{\ket{\unn O'^{(X)}_{\Pi_3}}}^2}{N'^{\labs{\Pi_3}}} 
    \right)}_{=: \xi_2(1/N') = \xi(1/ \sqrt{N'})} \label{eq:NprimeDependentPart1}
\end{align}
where $\xi$ and $\xi_2$ satisfy $\xi(0)=\xi_2(0)=1$ and $\xi(x)=\xi_2(x^2)$ for all $x \in [0,1/\sqrt t]$.
The real functions $\xi, \xi_2$ have polynomial approximations supplied by the following.

\begin{lemma}[Taylor approximation]\label{lemma:Taylor} 
    For integers $p \geq r \geq 0$ and reals
    $\{a_1, \ldots, a_r, b_1,\ldots, b_{p-r} \} \subset [0, w] \subset \RR$, 
    we let
    \begin{align}
        g(z):= \frac{\prod^{r}_i\sqrt{1 - a_i z}}{\prod^{p-r}_j\sqrt{1 - b_j z}}.
    \end{align}
    Then, the Taylor expansion~$g_q(z)$ of $g(z)$ up to order~$q$ satisfies
    \begin{align}
        \labs{ g(z) - g_q(z)}
        \le 3^p 2^{-q}\quad \text{for all}\quad z \in \left[0,\frac{1}{5 w}\right].
    \end{align}
\end{lemma}

\noindent
The proof is straightforward by Taylor expansion and is found below.
Note that the error bound depends on the total number~$p$ of factors in $g$,
regardless of whether a factor is in the numerator or denominator.

There are $4t$ factors in $\xi_2(z) = \xi_2(1/N')$ where a factor is of form~$(1 - c z)^{\pm 1/2}$.
Therefore, $\xi_2(z)$ can be approximated by a degree-$q$ real polynomial 
with error $3^{4t} 2^{-q}$ for $z \in [0,(5 t)^{-1}]$.
For $\xi(x) = \xi_2(x^2)$, the same approximation error bound holds with a polynomial of degree $2q$.
The prefactor $N'^{-(\abs{\Pi_1} + \abs{\Pi_2} - 2\abs{\Pi_3})/2}$ in \cref{eq:NprimeDependentPart1} 
has degree between zero and $2t$ as a monomial in $N'^{-1/2}$
because $\Pi_3 \succeq \Pi_1$ and $\Pi_3 \succeq \Pi_2$, 
and hence $\abs{\Pi_3} \le \abs{\Pi_1}$ and $\abs{\Pi_3} \le \abs{\Pi_2}$.
Therefore, we have a polynomial $h_{2q+2t}(z) \in \CC[z]$ of degree at most $2q + 2t$
such that for all $N' \ge t$
\begin{align*}
    \abs*{
        \braket{\phi_\text{left}^{N'}|\phi_\text{right}^{N'}} - h_{2q+2t}(\frac{1}{\sqrt{N'}})
    } \le 
    \left( 
        \sum_{\Pi_1, \Pi_2, \Pi_3 \vdash [t]} |c(\Pi_1, \Pi_2, \Pi_3)| 
    \right)  \frac{3^{4t}}{2^q}  \, .
\end{align*}

To bound the sum of $\abs{c(\Pi_1,\Pi_2,\Pi_3)}$ we recall \cref{eq:denominators_ABBC}.
This sum contains at most~$(t!)^3$ terms by \cref{lemma:comb_bound_partitions}.
We see that the degree-$(2q + 2t)$ polynomial $h$ satisfies
\begin{align}
\labs{\braket{\phi_\text{left}^{N'}|\phi_\text{right}^{N'}} - h(\frac{1}{\sqrt{N'}})} \le (t!)^5 \cdot 
\frac{3^{4t}}{2^q}
\label{eq:poly_ABBC}
\end{align}
for all $N' \ge t$. 
Using analogous steps to approximate the norms that appear in~\cref{eq:denominators_ABC}, 
we can also find a 
polynomial~$g$ of degree $q+ \cO(t)$ that satisfies
\begin{align}
    \labs{\bra{\phi} \avg [ P(\pi_{ABC})^{\otimes t}]\ket{\psi} - g(\frac{1}{\sqrt{N'}})} 
    \le (t!)^{9} 
    \frac{3^{\cO(t)}}{2^q}
    \,. \label{eq:poly_ABC}
\end{align}
Combining~\cref{eq:poly_ABBC,eq:poly_ABC}, 
we get that for any integer $q \ge 1$, the polynomial $p(x) = h(x) + g(x)$ of degree $q + \cO(t)$ satisfies 
\begin{align}
    \labs{f(\frac{1}{\sqrt{N'}}) - p(\frac{1}{\sqrt{N'}})} \leq (t!)^{9} 
    \frac{3^{\cO(t)}}{2^q}
\end{align}
for all $N' \ge t$. 

\paragraph{Step $4$: Bounding $|f(1/\sqrt{N})|$ using Markov's other inequality.}

We choose $q = \Theta(t \log t + \log N)$.
Then, the polynomial~$p$ has degree $r = \Theta(t \log t + \log N)$ and achieves an approximation error
\begin{align}
    \labs{f(\frac{1}{\sqrt{N'}}) - p(\frac{1}{\sqrt{N'}})} \leq \frac{1}{N} \quad \text{for all}\quad N'\ge t. 
\end{align}
Combined with~\cref{eqn:ftozero,eqn:fbounded}, this implies 
\begin{align}
    \labs{\lim_{N' \to \infty} p(\frac{1}{\sqrt{N'}})} \leq \frac{1}{N}
    \quad \text{and} \quad 
    \sup_{N' \ge t} |p(\frac{1}{\sqrt{N'}})| \leq 2 + \frac{1}{N} \leq 3 \,. \label{eqn:P_props}
\end{align}
Triangle inequality implies
\begin{align}
    \labs{f(\frac{1}{\sqrt{N}})} \leq \labs{p(\frac{1}{\sqrt{N}})} + \frac{1}{N} \leq \labs{\Re p(\frac{1}{\sqrt{N}})} + \labs{\Im p(\frac{1}{\sqrt{N}})} + \frac{1}{N} \,. \label{eqn:fboundp}
\end{align}
Both real and imaginary parts of~$p$ are real polynomials of degree at most~$r$.

\begin{lemma}\label{cor:1/sqrtN_spacing}
Let $p \in \RR[z]$ be a real polynomial in a variable~$z$ of degree~$q$.
Let $N_0 = 4 q^2$.
Then, for all integers $N \geq N_0$,
\begin{align}
    \labs{p(\frac{1}{\sqrt{N}}) - p(0)} \leq \frac{8 q^3}{\sqrt{N}} \sup_{N' \in \N, N' \ge N_0}\labs{p(\frac{1}{\sqrt{N'}})}.
\end{align}
\end{lemma}

\noindent This is proved below. T
ogether with~\cref{eqn:P_props} we find that 

\begin{align}
    \labs{\Re p(\frac{1}{\sqrt{N}})} \leq \labs{\lim_{N' \to \infty} p(\frac{1}{\sqrt{N'}})} + \frac{10 r^3}{\sqrt{N}} \sup_{N' \ge t} \labs{p(\frac{1}{\sqrt{N'}})} \leq \cO(\frac{r^3}{\sqrt{N}})
\end{align}
and analogously 
\begin{align}
    \labs{\Im p(\frac{1}{\sqrt{N}})} \leq \cO(\frac{r^3}{\sqrt{N}}) \,.
\end{align}
Plugging this into~\cref{eqn:fboundp} we complete the proof of~\cref{thm:permutationoverlap}.
\end{proof}

\subsection{Deferred proofs}

\begin{proof}[Proof of \cref{lemma:Taylor}]
    We begin with Taylor expansions
    \begin{align}
        \sqrt{1-y} &= 1 - \frac{1}{2} y  - \frac{1}{8} y^2 \cdots, \qquad
        \frac{1}{\sqrt{1-y}} = 1 + \frac{1}{2 }y  +  \frac{3}{8}y^2+\cdots  .
    \end{align}
    where all the coefficients are bounded by~$1$ in magnitude.
    Therefore, both series are dominated by $1, y, y^2, \ldots$ where $ y \in [0,1)$,
    whose sum is $1/(1-y)$.
    
    Normalizing~$x$ by~$w$, 
    we may assume that $\abs{a_i}, \abs{b_j} \le 1, \abs{x} \le \frac 1 {5}$.
    It follows that the Taylor expansion of the target function 
    $(\prod_i (1 - a_i x)^{1/2})(\prod_j(1 - b_j x)^{-1/2})$ where $x \in [0, 1)$
    is dominated by the power series of $(1-x)^{-p}$.
    Hence, the truncation error of the Taylor series of the target function is
    upper bounded by that of $(1-x)^{-p}$.
    Taylor's theorem implies that if $h_q(x)$ is the Taylor series of $(1-x)^{-p}$ up to order~$q$,
    then for all $x \in [0,\frac 1 5]$ we have 
    \begin{align}
        \abs{(1-x)^{-p} - h_q(x)} 
        &\le 
        \max_{z \in [0,\frac 1 5]} \frac{p(p+1)(p+2) \cdots (p+q)}{(q+1)!}\frac{z^{q+1}}{(1-z)^{p+q+1}} \\
        &=
        \binom{p+q}{p-1}\frac{(1/5)^{q+1}}{(4/5)^{p+q+1}} \le 2^{p+q}\frac{(1/5)^{q+1}}{(4/5)^{p+q+1}} = \frac 1 {2^{q+2}}\parens*{\frac 5 2}^p. \nonumber \qedhere
    \end{align}
\end{proof}

To prove \cref{cor:1/sqrtN_spacing}, we need some other lemmas.

\begin{lemma}[Markov's other inequality{~\cite[Theorem 5.1.8.]{polynomials_inequalities}}]\label{lemma:Markovs}
        Let $f \in \RR[x]$ be a real polynomial of degree~$q$. Then,
    \begin{equation}
        \max_{x\in [0,1]} \abs{f'(x)} \leq 2q^2 \max_{x\in [0,1]} \labs{f(x)}.
    \end{equation}
\end{lemma}

\noindent
Instead of using the maximum over the continuous interval $[0,1]$ in \cref{lemma:Markovs}, 
we can restrict to a subset of $[0,1]$ using the following simple fact.

\begin{lemma}[Uniform bounds from samples]\label{lemma:uniform_sparse}
    Let $f \in \RR[x]$ be a real polynomial of degree~$q$ 
    and let $U\subseteq [0,1]$ be a subset.
    Suppose that for any $x \in [0,1]$ we have $\mathrm{dist}(x,U) = \inf_{u \in U} \abs{u - x} \le \Delta < \frac 1 {2q^2}$.
    Then,
    \begin{align}
        \max_{x\in [0,1]} \abs{f(x)} \leq \frac{\sup_{u \in U} \abs{f(u)}} {1-2q^2\Delta}. 
    \end{align}
\end{lemma}
\begin{proof}
    Suppose $\max_{x \in [0,1]} \abs{f(x)} = \abs{f(y)}$ for some $y \in [0,1]$.
    Let $y' \in U$ be a nearest point to~$y$.
    If $\abs{f(y)} = \abs{f(y')}$, there is nothing to prove.
    Assume that $\abs{f(y)} > \abs{f(y')}$ so $0 < \abs{y - y'} \le \Delta$.
    Then, by \cref{lemma:Markovs}
    \begin{align}
        \frac{\labs{f(y)}-\labs{f(y')}}{\labs{y-y'}} &\le \labs{\frac{f(y)-f(y')}{y-y'}} 
        \le \max_{x\in [0,1]}|f'(x)| 
        \le 2q^2 \labs{f(y)}
    \end{align}
    Rearranging, we complete the proof.
\end{proof}

\begin{proof}[Proof of \cref{cor:1/sqrtN_spacing}]
Let $U = \{\sqrt{N_0/N} \; | \; N \in \N, N \geq N_0 \} \subset [0,1]$.
Note that the distance from any $x \in [0,1]$ to a nearest element in $U$ is at most
\begin{align}
    \Delta \leq \sqrt{N_0/N_0} - \sqrt{N_0/(N_0+1)} \leq 1/N_0 \leq 1/(4 q^2) \,. \label{eqn_U_gap_bound}
\end{align}
Define a rescaled function $\tilde f(x) = f(x/\sqrt{N_0})$.
Then 
\begin{align}
    \labs{f(\frac{1}{\sqrt{N}}) - f(0)} 
    &= \labs{\tilde f(\sqrt{N_0/N}) - \tilde f(0)} \nonumber\\
    &\leq \sqrt{N_0/N} \; \max_{x \in [0,1]} |\tilde f'(x)| \\
    &\leq 2q^2 \sqrt{N_0/N} \; \sup_{x \in [0,1]} |\tilde f(x)| &\text{by \cref{lemma:Markovs}} \nonumber\\
    &\leq 4 q^2 \sqrt{N_0/N} \; \sup_{x \in U} |\tilde f(x)| &\text{by \cref{lemma:uniform_sparse} and~\cref{eqn_U_gap_bound}} \nonumber\\
    &= 4 q^2 \sqrt{N_0/N} \;\; \sup_{N' \in \N, N' \geq N_0} |f(\frac{1}{\sqrt{N'}})| & \text{by defn. of $U$ and $\tilde f$.} \nonumber
\end{align}
The lemma now follows.
\end{proof}

\section{Permutation \texorpdfstring{$t$}{t}-designs with multiplicative error \texorpdfstring{$\eps$}{ε} in circuit depth \texorpdfstring{$\cO(nt + \log(1/\eps))$}{O(nt+log(1/ε))}} \label{sec:nt_depth_permutations}

Our main focus in this work are random circuits.
However, using our analysis of Kassabov's generators in \cref{section:efficientreversiblecircuits}, we can also obtain permutation $t$-designs in circuit depth $\cO(nt + \log(1/\eps))$ (i.e.~without logarithmic factors).
This was claimed informally in~\cite{alon2012almost}, but as far as we know no proof has appeared in the literature.
As this result might be of independent interest, we give a proof here.

\begin{theorem}\label{thm:linear_perm_design}
    For any integers $n\geq 1$ and $t\leq \cO(2^{n/6.1})$, there exists an explicit set $S$ such that each element in $S$ is a product of NOT, CNOT, and Toffoli gates with circuit depth $\cO(1)$ assuming all-to-all connectivity, and
    \begin{align}
        g(\mu(S), \ \tau, \ \Alt(2^n)) = 1 - \Omega(1),
    \end{align}
    where $\tau : \pi \mapsto P(\pi)^{\otimes t}$ and $P(\pi) \ket z = \ket{\pi(z)}$ for any $z \in \{0,1\}^n$ and $\pi \in \Sym(2^n)$. As a result, there exists a constant $C>0$ such that $\mu(S)^{*k}$ is an approximate permutation $t$-design on $n$ bits with multiplicative error $\eps$ when $k \geq C(nt+\log(1/\eps))$. 
\end{theorem}

Even though we are equipped with depth-$1$ implementations of the expanding generators (\cref{sec:kas_depth_one}),
we encounter additional hurdles toward the proof of a linear-depth permutation $t$-design as in~\cref{thm:linear_perm_design}.
The first is that Kassabov's generators implemented by reversible gates on $n$ bits
do not generate $\Alt(2^n)$, 
but only a subgroup of $\Alt(2^n)$ that maps onto $\Alt((2^{3s}-1)^6)$ for $s \le n/18$ (more in~\cref{sec:kas_exterior}).

Instead of trivializing the action on the exterior, in this appendix, we overcome this hurdle by showing that the alternating group $\Alt((2^{3s}-1)^6)$ 
combined with some random bit flips is good enough.
We then extend our result for $18s$ bits only to any $n\geq 1$ bits using the overlap lemma in~\cref{section:permutationoverlap}.

\subsection{Extending almost full alternating groups with random bit flips}

\newcommand{\projgood}{{\Pi_\mathrm{good}}}
\newcommand{\projbad}{{\Pi_\mathrm{bad}}}

The ``almost full'' alternating group $\Alt((2^{3s}-1)^6)$ is a subgroup of $\Alt(2^n)$
and permutes bit strings in
\begin{align}
    \mathbf K_s = \left\{(z_1, \cdots, z_6) \in (\{0,1\}^{3s})^{\times 6}~\middle|~ z_1, z_2,\ldots,z_6 \neq 0^{3s} \right\} 
    \subset \{0,1\}^{18s}.
\end{align}
Note that $|\mathbf K_s| = (2^{3s}-1)^6$.
Let $\projgood$ be the orthogonal projector onto the $\CC$-span of $\{ \ket{z} ~|~ z \in \mathbf K_s \}$,
and $\projbad = \one - \projgood$ denote the orthogonal complement.
The result of \cref{sec:kas_ckt} is that each Kassabov generator is 
a circuit consisting of gates, each of which commutes with $\projgood$.

Since we want $\Alt(2^{18s})$, not just $\Alt(\mathbf K_s)$, 
we need some permutation that does not commute with $\projgood$.
The simplest is a bit flip, which is a Pauli operator~$X$ on some bit.
Let $P_X = \{I, X\}^{\otimes n}$ be the group of all bit flip operators,
whose action on~$\{0,1\}^n$ gives a subgroup of $\Alt(2^n)$.
Define
\begin{align}
    \cX : \rho \mapsto \avg_{\pi \sim \mu(P_X)} P(\pi)^{\otimes t}\, \rho \, P(\pi)^{\dag \otimes t} .
\end{align}
Since \cref{cor:AdditiveToMultiplicativeError} is stated in terms of channels,
we identify the vector space~$(\CC^2)^{\otimes nt}$ with 
a vector space of all diagonal $2^{nt} \times 2^{nt}$ matrices:
a vector $\ket{\psi} = \sum_{x\in\{0,1\}^{nt}} \psi_x \ket{x}$ 
is identified with a matrix $\rho_{\psi} = \sum_{x} \psi_x \ket{x}\!\bra{x}$.
The vector norm of the former is the Schatten $2$-norm of the latter.
The action of $P(\pi)^{\otimes t}$ on $\ket{\psi}$ is equivalent to 
$P(\pi)^{\otimes t}\rho_{\psi}P(\pi)^{\dagger \otimes t}$.
In addition, it will be convenient to define a channel $\cD$ that strips off all the off-diagonal elements
and keeps the diagonal elements only; this is just a notational tool to focus on diagonal matrices, i.e.~those for which $\rho = \cD(\rho)$.

\begin{lemma}\label{lem:SmallBad}
    Let $n = 18s \ge 18$
    and $\rho=\sum_{x\in \{0,1\}^{nt}}\rho_x\ket{x} \! \bra{x}$ be a diagonal matrix.
    Then,
    \begin{align}
        \norm*{\left(\one - \projgood^{\otimes t} \right) \cdot \cX (\rho)}_{1}
        \le
        6 t \cdot 2^{-n/6} \norm{\rho}_1 \, .
    \end{align}
    In addition, 
    for any self-adjoint moment operator~$H = \avg_{\pi \sim \nu} P(\pi)^{\otimes t}$ 
    over a distribution~$\nu$ on~$\Alt(2^n)$,
    if $H$ commutes with~$\projgood^{\otimes t}$,
    then we have 
    \begin{align}
        &\norm*{
            M_X H (\one - \projgood^{\otimes t}) M_X
        }_\infty
        \le
        6t \cdot 2^{-n/6}\\
        &\quad \text{ where } \quad M_X = M(\mu(P_X), \tau) \deq \avg_{\pi \sim \mu(P_X)} P(\pi)^{\otimes t}\, . \nonumber
    \end{align}
    Here and throughout this section, we use the representation $\tau: \pi \mapsto P(\pi)^{\ot t}$ as in \cref{sec:gap_revckt}.
\end{lemma}

In words, random bit flips bring a bit string into the ``good'' subspace almost always.

\begin{proof}
    Write $x = (x_1,\ldots,x_t)$ where $x_i \in \FF_2^n$.
    We can assume $\rho \succeq 0$ and $\norm{\rho}_1 = \Tr(\rho) = 1$.
    By linearity of expectation,
    $
        \cX(\rho) 
        = \sum_{x_1,\cdots, x_t \in \FF_2^n} \rho_x \, 
            \avg_{y \in \FF_2^n} 
                \ket{x_1 + y, \cdots, x_t + y} \! \bra{x_1 + y, \cdots, x_t + y}
    $.
    The operator $(\projbad \otimes \one^{\otimes (t-1)})\cX(\rho)$ is diagonal, 
    so the $1$-norm is just the trace:
    \begin{align}
        \norm{
            (\projbad \otimes \one^{\otimes (t-1)})\cX(\rho)
        }_1 
        &= 
        \sum_{x_1,\cdots, x_t \in \FF_2^n} \rho_x \, \avg_{y \in \FF_2^n} 
            \bra{x_1+y} \projbad \ket{x_1+y} 
            \braket{x_2 + y | x_2 + y} 
            \cdots 
            \braket{x_t + y | x_t + y} \nonumber\\
        &= 
        \sum_{x_1,\cdots, x_t \in \FF_2^n} \rho_x \, \frac{\Tr(\projbad)} {2^n}
        = 
        \frac{2^n - (2^{3s}-1)^6}{2^n} 
        \le 
        \frac{6}{2^{n/6}}.
    \end{align}
    An obvious operator inequality
    \begin{equation}
        \one - \projgood^{\otimes t} \preceq \sum_{i=1}^{t} \one^{\otimes (i-1)}\otimes \projbad \otimes \one^{\otimes (t-i)}
    \end{equation}
    gives the first claim.
    
    For the second claim, we identify the operator~$M_X H(\one - \projgood^{\otimes t}) M_X$ 
    with a self-adjoint superoperator~$\Phi$ acting on the diagonal matrices. 
    Spelling out $\Phi$, we have
    \begin{align}
        \Phi(\rho) = \avg_{x,z \sim \mu(P_X),~ y \sim \nu} P(x)^{\otimes t} P(y)^{\otimes t}(\one - \projgood^{\otimes t}) P(z)^{\otimes t} (\cD \rho) P(z)^{\dagger \otimes t} (\one - \projgood^{\otimes t})P(y)^{\dagger \otimes t} P(x)^{\dagger \otimes t} .
    \end{align}
    By definition, $\norm{\Phi}_{2\to 2} = \norm{M_X H(\one - \projgood^{\otimes t}) M_X}_\infty$.
    The claim then follows from~\cref{cor:AdditiveToMultiplicativeError}.
\end{proof}

\begin{lemma}\label{lemma:gapXRX}
    If $n=18s$, then
    \begin{align}
        \norm[\big]{
            M_X M_K M_X - M_A
        }
        \le 36 t / 2^{n/6} \, ,
    \end{align}
    where
    \begin{align}
        M_A &= M(\mu(\Alt(2^n)), \tau ) \deq \avg_{\pi \sim \mu(\Alt(2^n))} P(\pi)^{\otimes t},\\
        M_K &= M(\mu(\Alt(\mathbf K_s)), \tau ) \deq \avg_{\pi \sim \mu(\Alt(\mathbf K_s))} P(\pi)^{\otimes t}\, . \nonumber
    \end{align}

\end{lemma}

This means that the spectral gap for $\mu(P_X) * \mu(\Alt(\mathbf K_s)) * \mu(P_X)$ 
is a constant as long as $t \ll 2^{n/6}$.

\begin{proof}
    Define
    \begin{align}
        \cR &: \rho \mapsto \avg_{\pi \sim \mu(\Alt(\mathbf K_s))} P(\pi)^{\otimes t}\, \rho \, P(\pi)^{\dag \otimes t},\\
        \cA &: \rho \mapsto \avg_{\pi \sim \mu(\Alt(2^n))} P(\pi)^{\otimes t}\, \rho \, P(\pi)^{\dag \otimes t} \, .\nonumber
    \end{align}
    
    Now, the essential norm of interest can be written as
    \begin{align}
        \norm*{
            \avg_{\pi\in \mu(P_X) * \mu(\Alt(\mathbf K_s)) * \mu(P_X)} P(\pi)^{\otimes t} 
            -
            \avg_{\pi\in \mu(\Alt(2^n))} P(\pi)^{\otimes t}
        }
        &=
        \norm*{
            \cD (\cX \cR \cX - \cA) \cD
        }_{2\to 2}\\
        &\le
        \norm*{
            (\cR \cX - \cA) \cD
        }_{1 \to 1} & \text{by~\cref{cor:AdditiveToMultiplicativeError}} .\nonumber 
    \end{align}
    The left and right invariance of~$\mu(\Alt(2^n))$ implies $\cX \cA = \cA = \cA \cX$,
    and hence, assuming $\norm{\cD \rho}_1 = 1$,
    \begin{align}
        \norm*{ ( \cR\cX - \cA )(\cD \rho) }_1 
        &=
        \norm*{ ( \cR - \cA ) \circ  \cX (\cD \rho) }_1 \\
        &\leq 
        \norm*{(\cR - \cA)\left(\projgood^{\otimes t} \cdot \cX(\cD\rho)\right)
        }_1 
        + 
        \norm{\cR - \cA}_{1 \to 1} \cdot 6 t \cdot 2^{-n/6}  & \text{by \cref{lem:SmallBad}} \nonumber\\
        &\leq 
        \norm[\big]{(\cR - \cA)\big(\underbrace{\projgood^{\otimes t} \cdot \cX(\cD\rho)}_{\tilde \rho} \big)
        }_1 
        + 
        12 t \cdot 2^{-n/6} \, .\nonumber
    \end{align}
    Clearly, $\tilde \rho$ is supported on the span of $x = (x_1,\ldots,x_t)$ with all $x_i \in \mathbf K_s$
    and has $1$-norm at most~$1$.
    
    We have to bound $\norm{(\cR - \cA)(\tilde \rho)}_1$ or $\norm{(\cR - \cA)(\ket x \! \bra x)}_1$ 
    for any $x \in \mathbf K_s^{\times t}$.
    Since we have the invariant probability measures on the two groups~$\Alt(\mathbf K_s)$ and $\Alt(2^n)$,
    we see that for each $x = (x_1,\ldots,x_t)$, both $\cR(\ket x\!\bra x)$ and $\cA(\ket x \! \bra x)$ are
    proportional to some projectors, normalized to have trace 1.
    The images of these projectors are the spans of the orbits of~$x$ under the action~$P(\pi)^{\otimes t}$.
    
    To understand the orbit, we recall that for any integer $N \ge 3$ 
    the group $\Alt(N)$ is $(N-2)$-transitive.
    That is, if $m \le N - 2$, any tuple of $m$ distinct entries chosen from $N$ letters
    can be mapped by $\Alt(N)$ to any other such tuple.
    Since the claim of the lemma is vacuous unless $t \le 2^{n/6}$,
    we may assume that,
    for any $m \in \{1,2,\ldots, t\}$,
    the alternating groups $\Alt(2^n)$ and $\Alt(\mathbf K_s)$ act transitively
    on
    \begin{align}
        \mathrm{Diff}_m &= \{(x_1,\ldots, x_m) \in (\FF_2^n)^{\times m} ~|~ x_a \neq x_b \text{ if } a \neq b \} \quad \text{ and } \quad 
        \mathrm{Diff}_m \cap \mathbf K_s^{\times m} ,\label{eq:diffgood}
    \end{align}
    respectively.
    Every $t$-tuple $x = (x_1,\ldots,x_t)$ corresponds to a tuple $[x] \in \mathrm{Diff}_m$ 
    where $m$ is the number of distinct entries of $x$
    and $[x]$ is obtained by removing any repeated entries from~$x$.
    This correspondence gives a linear map $\ket{x} \mapsto \ket{[x]}$,
    which commutes with the permutation group action:
    \begin{align}
        P(\pi)^{\otimes t} \ket x  = \ket{\pi(x)} \mapsto \ket{[\pi^{\times t}(x)]} = \ket{\pi^{\times m}([x])} = P(\pi)^{\otimes m} \ket{[x]}
    \end{align}
    It is now clear that 
    $\norm{
        (\cR - \cA)(\ket{x}\!\bra{x})
    }_1$ is equal to
    \begin{align}
        \norm[\Big]{
            \underbrace{
                \avg_{\pi \sim \mu(\Alt(\mathbf K_s))} P(\pi)^{\otimes m}\ket{[x]}\!\bra{[x]}P(\pi)^{\dagger \otimes m} 
            }_{E_\cR}
            - 
            \underbrace{
                \avg_{\pi \sim \mu(\Alt(2^n))} P(\pi)^{\otimes m}\ket{[x]}\!\bra{[x]}P(\pi)^{\dagger \otimes m} 
            }_{E_\cA}
        }_1 \, .
    \end{align}
    Since the action of the alternating groups is transitive in the respective sets in~\cref{eq:diffgood},
    the operators $E_\cR$ and $E_\cA$ are trace-normalized identity operators
    on the spans over $\mathrm{Diff}_m \cap \mathbf K_s^{\times m}$
    and over $\mathrm{Diff}_m$, respectively.
    Put $r = \rank E_\cR = \abs{\mathrm{Diff}_m \cap \mathbf K_s^{\times m}}$ 
    and $a = \rank E_\cA = \abs{\mathrm{Diff}_m}$.
    Then,
    \begin{align}
        \norm{
            (\cR - \cA)(\ket{x}\!\bra{x})
        }_1
        =
        \norm{E_\cR - E_\cA}_1
        =
        r \cdot \left(\frac 1 r - \frac 1 a \right) + (a-r)\cdot \frac 1 a
        =
        \frac{2(a-r)}{a} .
    \end{align}
    We complete the proof with simple estimates for $r$ and $a$.
    \begin{align}
        &1 - \frac{1}{2^n} \frac{m(m-1)}{2}
        \le 
        \frac{a}{2^{nm}} =\left( 1 - \frac{1}{2^n}\right) \cdots \left( 1 - \frac{m-1}{2^n}\right)
        \le
        1, \\
        &\frac{r}{2^{nm}} = \frac{1}{2^{nm}} \abs{\mathrm{Diff}_m \cap \mathbf K_s^{\times m}} 
        \ge 1 - \frac{6m}{2^{n/6}} - \frac{m(m-1)}{2^{n+1}}
        \ge 1 - \frac{12m}{2^{n/6}},\nonumber
    \end{align}
    where the second line is because
    $2^{-nm}\abs{\mathbf K_s^{\times m}} = (1-2^{-3s})^{6m} \ge 1 - 6m \cdot 2^{-n/6}$.
    Therefore, $1 - \frac r a \le 12m / 2^{n/6} \le 12 t / 2^{n/6}$.
\end{proof}

\subsection{Composting Kassabov's generators}

Consider $n=18s$ bits where $s\geq 1$.
Kassabov's result (\cref{sec:kas_depth_one}) says that there exists a generating set of a fixed size,
with respect to which the Kazhdan constant of $\Alt(\mathbf K_s)$ is bounded below by a constant.
Each generator here is implemented by a product of $\cO(n)$ nonoverlapping CNOT and Toffoli gates (so depth is $1$ with all-to-all connectivity),
each of which preserves the set $\mathbf K_s$ of all ``good'' bit strings.

Let $S_0$ be the set of Kassabov generators for~$\Alt(\mathbf K_s)$. 
$S_0$ generates a subgroup~$B$ of~$\Alt(2^n)$ 
such that there is a surjective group homomorphism~$B \to \Alt(\mathbf K_s)$;
this does not mean that $B$ is a product of $\Alt(\mathbf K_s)$ and some other group.
Nonetheless we can replace $\Alt(\mathbf K_s)$ in \cref{lemma:gapXRX} 
with~$S$.

\begin{proof}[Proof of~\cref{thm:linear_perm_design}]
We first prove the theorem for $n=18s$ with $s\geq 1$. 
Let $S = \{xyz: x,z\in P_X, y\in S_0\}$. Each element in $S$ has a constant circuit depth using NOT, CNOT, and Toffoli gates. Hence, it suffices to show that for any $t = \cO( 2^{n/6} n^{-3} )$, 
\begin{align}
    g\big(
        \mu(P_X) * \mu(S_0) * \mu(P_X),~ \tau,~ \Alt(2^n) 
    \big)
    = 1 - \Omega(1)\, .
\end{align}
Write $M_\nu = \avg_{\pi \sim \nu_S} P(\pi)^{\otimes t}$ for brevity.
The claim is that $\norm{M_X M_\nu M_X - M_A} = 1 - \Omega(1)$ in the designated regime of~$t$.

As we have remarked, $\abs{S_0} = \cO(1)$.
The short product \cref{lem:shortproduct} then
implies that $\cK(\Alt(\mathbf K_s); S_0|_{\mathbf K_s}) = \Omega(1)$,
where the restriction means that each gate viewed as a function~$\{0,1\}^{18s} \to \{0,1\}^{18s}$ 
is restricted to the domain~$\mathbf K_s$.
For linear operators on~$(\CC^2)^{\otimes 18 s t}$, 
the restriction amounts to multiplying by $\Pi = \projgood^{\otimes t}$:
\begin{align}
    \norm[\Big]{
        M_\nu \Pi 
        -
        M_K \Pi
    }
    \le
    1 - \frac{\cK(\Alt(\mathbf K_s); S_0|_{\mathbf K_s})^2}{\abs{S_0}+1}
    =
    1 - \Omega(1)
\end{align}
where we used \cref{lem:gapviakazhdan}. 
The extra ``$+1$'' in the denominator amounts to adjoining the identity to $S_0$.
Now we use triangle inequality.
\begin{align}
    M_X M_\nu M_X - M_A \nonumber
    &=
    (M_X M_\nu \Pi M_X - M_X M_K \Pi M_X) + (M_X M_K M_X - M_A) \nonumber\\
    & \qquad + M_X M_\nu (\one - \Pi) M_X + M_X M_K (\Pi - \one) M_X \nonumber \\
    \norm{M_X M_\nu M_X - M_A}
    &\le
    \norm{M_X (M_\nu \Pi - M_K \Pi) M_X}
    +
    \norm{M_X M_K M_X - M_A} \\
    &\qquad + \cO(t 2^{-n/6}) & \text{by \cref{lem:SmallBad} twice}\nonumber\\
    &\le \norm{M_X}^2 \norm{M_\nu \Pi - M_K \Pi} + \cO(t 2^{-n/6}) & \text{by \cref{lemma:gapXRX}}
    \nonumber\\
    &\le 1 - \Omega(1) + \cO(t 2^{-n/6}) \, .\nonumber & 
\end{align}

We now extend the spectral gap from $n=18s$ qubits to any $n\geq 1$. 
Let there be $n=18s+n'$ bits where $0 \leq n' \leq 17$ and $s \ge 1$.
We work with $t = \cO(2^{n/6}n^{-3})$. 
(We will eventually take even smaller $t$.)
Let $A \sqcup B \sqcup C$ be a partition of $n$ bits (bit indices to be more precise) 
where $\abs A = \abs C = n'$,
so $\abs B = n - 2n'$.
Let $\nu_{AB}$ and $\nu_{BC}$ be the distributions of applying a uniformly random gate from $S$ to $A\sqcup B$ and $B\sqcup C$. 
Denote by $\Alt(2^J)$ for any set $J$ of bit indices the alternating group~$\Alt(2^{\abs J})$
that permutes bit strings on~$J$ and leaves other bits intact.
We have just shown that
\begin{align}
    \norm[\big]{
        \avg_{\pi \sim \nu_{AB}} \tau_{AB}(\pi) - \avg_{\pi \sim \mu(\Alt(2^{AB}))} \tau_{AB}(\pi)
    } = 1 - \Omega(1)
\end{align}
where $\tau_{J}(\pi) = (P(\pi)_J)^{\otimes t}$ for any set $J$ of bit indices.
Since tensoring $\one$ to an operator does not change the operator norm,
we have
\begin{align}
    \norm[\big]{
        M(\nu_{AB}^{*k}, \tau_{AB} ) \otimes \one_C - M(\mu_{\Alt(2^{AB})}, \tau_{AB}) \otimes \one_C
    } &= (1 - \Omega(1))^k \, ,\\
    \norm[\big]{
        \one_A \otimes M(\nu_{BC}^{*k}, \tau_{AB} ) - \one_A \otimes M(\mu_{\Alt(2^{AB})}, \tau_{AB}) 
    } &= (1 - \Omega(1))^k \nonumber
\end{align}
for any $k \ge 1$.
Triangle inequality together with the fact that any moment operator has norm at most~$1$,
gives
\begin{align}
    &\norm*{
        \big( M(\nu_{AB}^{*k},\tau_{AB}) \otimes \one_C \big)
        \big( \one_A \otimes M(\nu_{BC}^{*k}, \tau_{BC}) \big)
        - M(\mu_{\Alt(2^{ABC})}, \tau_{ABC})
    }\\
    &\le
    2(1 - \Omega(1))^{k} + 
    \norm*{
         \big(M(\mu_{\Alt(AB)}, \tau_{AB}) \otimes \one_C\big)
         \cdot \big(\one_A \otimes M(\mu_{\Alt(BC)}, \tau_{BC})\big)
         - M(\mu_{\Alt(2^{ABC})}, \tau_{ABC})
    }\, . \nonumber
\end{align}
The overlap theorem (\cref{thm:permutationoverlap}) asserts that
the last norm is $\cO\left( \frac{(t \log t + n- 2n_0)^3}{\sqrt{2^{n - 2n_0}}} \right)$.
If we further assume that $t = \cO(2^{n/6.1})$ where $6.1 > 6$ is arbitrarily chosen,
then this smaller than $1/4$.
Setting $k$ to be a sufficiently large constant, we ensure that the overall essential norm of $\nu_{AB}^{*k} * \nu_{BC}^{*k}$
be at most $3/4$, so the gap is a constant.
\Cref{lem:local-vs-parallel} implies that this gap implies 
for $\frac{1}{2k} (\sum^k \nu_{AB} + \sum^k \nu_{BC}) = \frac 1 2(\nu_{AB} + \nu_{BC})$
\begin{align*}
    g\left( \frac{\nu_{AB} + \nu_{BC}}{2}, ~\tau, ~\Alt(2^n) \right) = 1 - \Omega(1) \,.
\end{align*}
\end{proof}

\section{Other efficient generators for permutation groups} \label{app:caprace_kassabov}

Here, we discuss small generating sets for alternating groups~$\Alt(n_s)$
for some ``dense'' sequence of integers~$\{n_s\}$, different from those in~\cref{section:efficientreversiblecircuits}.
This family of generating sets is found in a recent work of Caprace and Kassabov~\cite{caprace2023tame}.
The associated Cayley graphs of $\Alt(n_s)$ are shown to be a uniform family of expander graphs.
The density above means that there exists a constant $c \ge 1$
such that for any positive integer $N$ 
there is an $n_s$ satisfying $n_s \le N \le c n_s$.
Using \cref{lem:GeneratingLargerSymmetricGroupFromSmallerOnes}, for each $n_s$ we can then extend these generators to ones for $\Alt(N)$ for $N \leq c n_s$, yielding generators for $\Alt(N)$ for any $N \in \N$.
We explain this in more detail below.

Though the Caprace--Kassabov generators~\cite{caprace2023tame} give expander graphs,
we remark that it is unclear how to adapt them to very efficient reversible 
classical (or quantum) circuits without any ancillas
that implement an approximate permutation design on $n$-bit strings.
This is the reason we focused on Kassabov's~\cite{Kassabov_2007_alt,Kassabov_2007_lattices} generator in the main text.
However, for applications that allow ancillas, the construction from~\cite{caprace2023tame} is much simpler and circuits for the generators follow easily; this was already used in~\cite{metger2024simple,chen2024efficient}.

Let $p$ be an odd prime.
Define three bijective functions on $\FF_p^3$:
\begin{align}
        \sigma(x,y,z) = (y,z,x), \qquad
        \alpha(x,y,z) = (x+y, y, z), \qquad
        \beta(x,y,z) &= (x+y^2, y, z) .
\end{align}
It is obvious that $\sigma,\alpha,\beta$ act on $\FF_p^3 \setminus\{(0,0,0)\}$.
Viewed as elements of~$\Sym(p^3-1)$, 
they are even permutations because they have orders $3,p,p$, which are all odd,
so their disjoint cycle representations can only have odd-length cycles.

\begin{theorem}[Theorem 1.5 of~\cite{caprace2023tame}]\label{thm:generating_set_Alt}
    The permutations~$\sigma,\alpha,\beta$ generate $\Alt(p^3-1)$.
    Moreover, there exists a real number $\eps > 0$ 
    such that for each odd prime~$p$ the Kazhdan constant satisfies
    \begin{align}
        \cK\left(\Alt(p^3-1); \{\sigma,\alpha,\beta,\sigma^{-1},\alpha^{-1},\beta^{-1}\}\right)
        \ge 
        \eps.
    \end{align} 
    So, the associated Cayley graph is an $\eps'$-expander of degree~$6$
    where $\eps' > 0$ depends only on~$\eps$.
    This can be extended from the alternating to the symemtric group: with any transposition~$\tau$, 
    say the one that swaps $(1,0,0)$ and $(0,1,0)$ but leaves all others invariant,
    the Kazhdan constant $\cK\left(\Sym(p^3-1); \{\tau,\sigma,\alpha,\beta,\sigma^{-1},\alpha^{-1},\beta^{-1}\}\right)$
    is at least~$\eps/2$.
\end{theorem}

Given any input $(x,y,z) \in \FF_p^3$ 
specified as a binary string of length~$3\lceil \log_2(p) \rceil$, 
the images $\sigma(x,y,z)$, $\alpha(x,y,z)$, and $\beta(x,y,z)$ (and their inverses)
can be evaluated at using $\CO(1)$ calls to finite field arithmetic.
Every finite field arithmetic operation $(+,-,*,/)$ can be implemented from $\CO(1)$
integer arithmetic on $\CO(n)$-bit integers,
which takes $\tilde\CO(n)$ bit operations~\cite{HarveyHoeven2019}.

We now explain how to extend \cref{thm:generating_set_Alt} from $\Sym(p^3-1)$ to $\Sym(2^n)$.
Bertrand's postulate, which is a theorem,
states that there is a prime in the interval $[n+1,2n]$ 
for any integer $n \ge 1$.
Hence, for any integer~$n \ge 2$,
there is an odd prime~$p$ such that $2^{n-1} < p < 2^n$,
which means $2^{3n - 3} \le p^3 - 1 < 2^{3n}$.
Hence, for any integer $n \ge 10$,
there is an odd prime~$p$ such that $ 2^{n-6} < p^3 - 1 < 2^n$.
(The requirement $n \ge 10$ is just a convenient choice.)
Consider embeddings (one-to-one group homomorphisms)~$\pi_j$ 
of~$\Sym(a = p^3-1)$ into~$\Sym(2^n)$
where $\pi_j(\Sym(a))$ ($j = 0,1,2,\ldots,m+1$) permutes points on intervals
\begin{align}
[1,a],\nonumber\\
[1 + 1 \cdot 2^{n-7}, a + 1 \cdot 2^{n-7}],\nonumber\\
[1 + 2 \cdot 2^{n-7}, a + 2 \cdot 2^{n-7}],\\
\cdots, \nonumber\\
[1 + m \cdot 2^{n-7}, a + m \cdot 2^{n-7}],\nonumber\\
[1+2^n-a, 2^n] \nonumber
\end{align}
where $m < 128$ is the greatest integer such that $a + 2^{n-7} m < 2^n$.
Then, applying \cref{lem:GeneratingLargerSymmetricGroupFromSmallerOnes} recursively,
we have that every element of~$\Sym(2^n)$ is a product of at most $2^{m+1} + 1$ elements 
from the embedded symmetric groups.

The action of the seven generators for~$\pi_j(\Sym(p^3-1))$, viewed as bijections~$\{0,1\}^n \to \{0,1\}^n$,
has circuit complexity~$\tilde \cO(n)$ for the following reason.
Given an $n$-bit string~$x$, one interprets it as an integer~$x$,
chooses  among $m+2$ intervals the left-most interval~$j$ to which $x$ belongs,
and computes the distance of $x$ from the left boundary of the interval~$j$.
Represent this distance, an integer, in base~$p$ to obtain $(x,y,z) \in \FF_p^3$,
act on it by the chosen generator, and put the result back into the interval~$j$.

The symmetric group~$\Sym(2^n)$ is now endowed with a generating set of size $7 \cdot (m+2) < 910$.
The Kazhdan constant for $\Sym(2^n)$ with this generating set 
is degraded by a factor of at most~$2^{m+1}+1$ (a constant independent of~$n$)
from that of $\Sym(p^3-1)$ with respect to the generating set of the seven elements.

We finally note that Bertrand's postulate says nothing on how to find an $(n - \CO(1))$-bit prime~$p$,
but we may resort to a probabilistic method with run time~$\poly(n)$.
The prime number theorem, that there are $\sim N / \log N$ primes less than $N$,
implies that a random $n$-bit integer 
is a prime with probability~$ \propto 1/n$.
Hence, after $\CO(n)$ random trials equipped with an efficient primality testing algorithm~\cite{AKS2004},
we find a desired odd prime~$p$ with high probability.

\end{document}